\documentclass[11pt]{amsart}

\usepackage{macros-fivebrane}

\usepackage{breakurl}

\begin{document}

\title{A holographic approach to the six-dimensional superconformal index}
\author{Surya Raghavendran and Brian R. Williams}
\maketitle

\begin{abstract}
We present a conjectural description of the space of local operators on a stack of finitely many fivebranes in $M$ theory at the level of the holomorphic twist.
Our approach is through the lens of twisted holography and utilizes a description of the minimal twist of eleven-dimensional supergravity. 
We find that the spaces of local operators are modules for the exceptional linearly compact super Lie algebra $E(3|6)$.
From the conjectural description of local operators we deduce closed formulas for the superconformal index of six-dimensional $\cN=(2,0)$ theories of type~$A_{N-1}$.
\end{abstract}

\setcounter{tocdepth}{1}
\tableofcontents

%
%


\section{Introduction}
Superconformal field theories admit a plethora of exactly computable, protected quantities, giving them a distinguished role in supersymmetric physics. 
Since it is believed that all supersymmetric field theories flow to superconformal fixed points, such quantities provide robust invariants of supersymmetric field theories. 
Examples of such a quantity is the superconformal index, which is a generating function for some collection of $R$-charges of certain BPS local operators. 
Crucial steps towards a more mathematical understanding of superconformal indices were taken in \cite{SWchar}, articulated through the construction of twisting.

Introduced by Witten \cite{WittenTwist} and further developed by Costello \cite{CostelloHol}, twisting refers to a localization, or fixed-point, construction for theories equipped with an action of a supersymmetry algebra. 
Operationally, one modifies the BRST differential of the theory by a nilpotent supercharge---the result is a theory for which the infinitesimal translations in the image of the nilpotent supercharge act homotopically trivially. 
When a twist exists (which is almost always) a supersymmetric field theory admits a so-called minimal twist; the resulting theory is a holomorphic-topological field theory that is holomorphic in the maximal number of spacetime directions.

One of the insights of \cite{SWchar} was that superconformal indices count exactly local operators in the minimal twist---accordingly, we may think of the space of local operators in the minimal twist as categorifying the index.
Moreover, the space of local operators, together with its algebra structure under operator product expansion, is part of the richer structure of a \textit{factorization algebra}.
Whilst the former governs the behavior of observables supported at points, the latter organizes observables that are supported on any open set including for example, non-local operators obtained from local ones via a descent procedure.

A principal goal of the present paper is to initiate the study of the factorization algebras of observables associated to the minimal twists of the six-dimensional~$\cN=(2,0)$ superconformal field theories. The full factorization structure is extremely rich, and in this paper we only discuss a small part of it, namely the costalk at a point as a vector space.
This is otherwise known as the space of local operators at the point. 
However, this information is enough to extract familiar quantities like the superconformal index.


Despite its ubiquity, the six-dimensional~$\cN=(2,0)$ supersymmetric theory is quite elusive---it admits no known Lagrangian formulation, and outside of the abelian case, is not known to admit a field theoretic realization.
In light of this, we propose to access the minimal twist and its local operators via the proposal of \textit{twisted holography}.

\subsection{Our approach via twisted holography}
Introduced by Costello and Li in \cite{CLsugra}, the twisted holography proposal posits an avatar of the AdS/CFT correspondence that holds at the level of supersymmetric twists.
Fundamentally, it conjectures a duality between the algebra of observables of a twisted gravitational theory and the algebra of observables of a twisted gauge theory describing the dynamics of a stack of branes coupled to the gravitational theory. There is an exciting body of work being developed around this program including tests of this proposal from both the gravitational and gauge theory sides.

To describe the kinds of theories on either side of the correspondence, we should comment on what it means to twist a theory of supergravity. We will provide more details in \S \ref{s:twisted}. Since the action of supersymmetry is gagued in theories of supergavity, the aforementioned twisting procedure does not quite make sense. Instead, Costello and Li define twisted supergravity to be the theory in perturbation theory around a background where the bosonic ghost for local supertranslations takes a nonzero nilpotent vacuum expectation value. 
Coupling the worldvolume theory of a brane to such a background has the effect of twisting the worldvolume theory in the usual sense.

It is expected that the six-dimensional $\cN=(2,0)$ superconformal field theory of type $A_{N-1}$ describes the low energy dynamics of a stack of fivebranes in M-theory on flat space. Accordingly, there is a minimal twist of the low energy limit of M-theory, eleven-dimensional supergravity, which induces the minimal twist on a stack of fivebranes. We will be interested in a twisted version of the $AdS_{7}/CFT_{6}$ correspondence which will relate the large $N$ limit of the minimally twisted six-dimensional $\cN=(2,0)$ theory with minimally twisted eleven-dimensional supergravity on $AdS_{7}\times S^{4}$.

The twisted holography propoosal moreover posits that holography can be understood as a concrete algebaic operation at the level of observables. We mentioned that the precise mathematical structure modeling observables of a quantum field theory is that of a factorization algebra. The predicted type of duality between the factorization algebras associated to a gravitational theory and to the worldvolume theory of a number of branes is a general version of \textit{Koszul duality}.
Ordinary Koszul duality for associative algebras (so quantum mechanical systems) associates to an (augmented) algebra $A$ a dual algebra $A^!$ whose appropriate derived category of modules is the same as that of $A$.
Following the work of \cite{CLsugra,CP1} (see also the review in \cite{PWkoszul}) there is a simple physical interpretation of Koszul duality.
If $A$ is the algebra of operators of some bulk quantum field theory (perturbatively we can even consider a theory of gravity) then $A^!$ is the algebra of operators on the universal topological line defect.
Universal here means that algebra of operators on any other line defect which couples to the bulk system admits a unique map of algebras from~$A^!$.

The general theory of Koszul duality for factorization algebras has not been developed, and we do not do so in this paper.
This sort of duality would allow one to make sense of universality statements as above for higher dimensional, possibly non-topological, defects in an arbitrary bulk quantum field theory.
Nevertheless, we can make the following ansatz for the Koszul dual $\cF^!$, which we refer to as the $!$-dual in this paper to avoid confusion, of a factorization algebra~$\cF$ of observables of some bulk quantum field theory.
It is the universal factorization algebra, along a specified defect, which couples to the bulk quantum field theory.
While this heuristic definition sounds natural, it does not lend itself to an explicit description.
However, for particular kinds of factorization algebras of bulk quantum field theories an explicit description is furnished by a local version of Noether's theorem, see~\S\ref{s:noether}.

Let us now make a more concrete, yet slightly informal, statement of twisted holography which fits into the approach of this paper.
Let $X$ be a smooth manifold, and let $\Obs_{grav}$ denote a factorization algebra on $X$ that we view as the observables of a bulk gravitational theory.
Suppose we have, in addition, a stack of $N$ branes, wrapping a closed submanifold $Y\hookrightarrow X$ whose worldvolume theory has a factorization algebra of observables $\Obs_{N-branes}$.
In the context of branes, it is natural to posit that the universal theory along the brane is given by the large $N$ limit of the factorization algebras $\Obs_{N-branes}$.

Note that $\Obs_{grav}$ is a factorization algebra on $X$, while $\Obs_{N-branes}$ is a factorization algebra on the closed submanifold $Y$ so we cannot yet compare them.
We can, however, restrict $\Obs_{grav}$ to a factorization algebra just on $Y$, which we denote by $\Obs_{grav}|_Y$.\footnote{In general this is given by some limit construction like for sheaves, but we will use a particularly nice model in the context of holomorphic-topological factorization algebras in this paper.}

\begin{expect}[Twisted holographic principle following \cite{CLsugra}]
After taking into account the backreaction, there is a map of factorization algebras
\[
  (\Obs_{grav}|_{Y})^{!}\to \Obs_{N-branes}
\]
that becomes an equivalence in the large $N$ limit.
\end{expect}

It is natural to wonder how this statement relates to more traditional formulations of the AdS/CFT correspondence, where local operators of the CFT are related to certain states of the gravitational theory on AdS. A precise relation will be articulated in \S\ref{sec:states} where we argue that one can indeed understand $(\Obs_{grav}|_{Y})^{!}$ as a factorization algebra enhancement of the space of (twisted) multi-particle states of the gravitational theory on the backreacted geometry. Explicitly, this geometry refers to the manifold $X - Y$ obtained by subtracting the brane together with the data of a certain Maurer-Cartan element deforming the geometric structure used to define the gravitational theory on $X - Y$.

The above expectation can be tested in instances where both sides of the duality admit explicit descriptions.
This has been carried out in many examples including:
\begin{itemize}
  \item A stack of $D3$ branes in twisted $\Omega$-deformed type IIB supergravity on flat space. The theory on the stack of $D3$ branes is dual to the closed string B-model on the deformed conifold \cite{costello2021twisted}. This can be understood as a twisted $\Omega$-deformed version of the physical AdS/CFT duality between 4d $\cN=4$ super Yang-Mills and type IIB string theory on $AdS_{5}\times S^{5}$. Here, the duality is formulated in terms of vertex algebras which are avatars of holomorphic factorization algebras on a Riemann surface.

  \item A system of twisted $D1-D5$ branes in a $T^{4}$-compactification of a twist of type IIB string theory. Upon compactifying along $T^4$ this can be understood as a variant of the above example where the deformed conifold is replaced by a certain superspace \cite{CP}. The duality is understood as a twisted version of the duality between type IIB supergravity on $AdS_{3}\times S^{3}\times T^{4}$ and the symmetric orbifold CFT $Sym^{N}(T^{4})$.

  \item Membranes and fivebranes in twisted $\Omega$-deformed $M$-theory on Taub-NUT space \cite{CostelloM5,CostelloM2}.
In the particular $\Omega$-background, membranes are localized to a topological quantum mechanical system where the duality can be phrased in terms of associative algebras and ordinary Koszul duality. Moreover, the $\Omega$-background localizes fivebranes to a complex plane and the observables of the localized theory are an affine $W_{N}$ vertex algebra. The holomorphic factorization algebras we seek in the present paper may be thought of as enhancements of these vertex algebras to three-complex dimensions.
\end{itemize}

In each of the above examples, there exist methods to characterize both sides of the duality independently and intrinsically. Indeed, conjectures of \cite{CLsugra} suggest that certain twists of type II superstrings are equivalent to certain topological strings. Utilizing such a description, one may characterize D-banes and associated open-string field theories in terms of categorical data.

For membranes and fivebranes in $M$-theory, in the nonminimal twist or its further $\Omega$-deformation, one may appeal to dimensional reduction arguments to give similar characterizations of the twisted subsectors of the relevant worldvolume theories.

However, for the minimal twist, we are not so lucky. The dimensional reduction of the nonminimal twist will necessarily involve subtle, nonperturbative effects making an intrinsic characterization along the lines of the above more difficult. We will offer some speculations in this direction at the end of the section below.

Consequentially, rather than prove a version of the twisted holographic principle, the objective of this paper is to \textit{use} the twisted holographic principal to conjecture an explicit description of the factorization algebra associated to the worldvolume theory on a stack of minimally twisted fivebranes. Such an analysis reveals that twists of worldvolume theories of branes enjoy certain infinite dimensional symmetries, enhancing finite dimensional symmetries present in the untwisted theories. In the minimal twist of eleven dimensional supergravity, these infinite dimensional symmetries are provided by certain exceptional simplie lie superalgebras.

\subsection{Infinite dimensional symmetry enhancement by exceptional simple super Lie algebras}

In \S\ref{s:twisted} we will recall a description of the minimal twist of eleven-dimensional supergravity, following \cite{RSW}. In this twist, the theory is holomorphic in a maximal number of directions, which is five complex directions, and topological in the remaining real direction. The relation between the minimal twist and other twists is summarized in the following diagram:

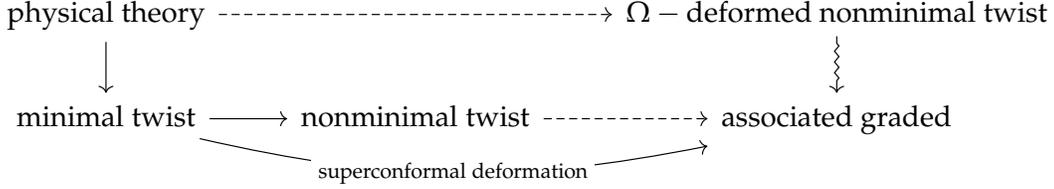
\[\begin{tikzcd}
	{\text{physical theory}}\ar[d]\ar[rr, dashed] & & {\Omega-\text{deformed nonminimal twist}}\ar[d, squiggly] \\
	{\text{minimal twist}}\ar[r]\ar[rr, "\text{superconformal deformation}" description, bend right = 12] & {\text{nonminimal twist}}\ar[r, dashed] & {\text{associated graded}}
\end{tikzcd}\]
We will elaborate on the meaning of the bottom arrow labeled ``superconformal deformation'' below.

Each of the above twists of supergravity on flat space admits a certain infinite dimensional algebra of symmetries. To begin with, the associated graded of the $\Omega$-deformed nonminimal twist is a holomorphic-topological theory in five dimensions. The theory on $\R\times \C^{2}$ depends on a holomorphic symplectic structure on $\C^{2}$, and the equations of motion include the Maurer-Cartan equation for an integrable deformation of such. Accordingly, the theory carries an action of the infinite dimensional lie algebra $\operatorname{Ham}(\C^{2})$ of hamiltonian vector fields on $\C^{2}$.

Surprisingly, there is a lift of this relationship to the minimal twist. In \cite{RSW} a certain exceptional simple super lie algebra called was shown to act on the minimal twist on $\R\times \C^{5}$. The super lie algebra is a certain $L_{\infty}$ extension of an exceptional simplie lie super algebra called $E(5|10)$. The algebras of observables of various twists of eleven dimensional supergravity on flat space are recorded in the diagram below; the twist is indicated by the position of the entry in comparison with the diagram above.

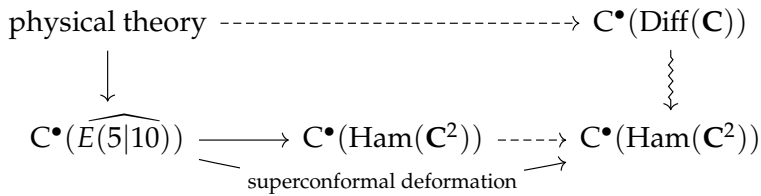
\[\begin{tikzcd}
	{\text{physical theory}}\ar[d]\ar[rr, dashed] & & {\clie^{\bu}(\operatorname{Diff} (\C))}\ar[d, squiggly] \\
	{\clie^{\bu}(\widehat {E(5|10)})}\ar[r]\ar[rr, "\text{superconformal deformation}" description, bend right = 12] & {\clie^{\bu}(\operatorname{Ham} (\C^{2}))}\ar[r, dashed] & {\clie^{\bu}(\operatorname{Ham}(\C^{2}))}
\end{tikzcd}\]

We note that the associated graded of the $\Omega$-deformed nonminimal twist and the nonminimal twist only differ by a tensor factor of the deRham complex on $\R^{6}$ - their algebras of local operators are quasi-isomorphic.


Strikingly, fivebranes in the minimal twist bring another exceptional super lie algebra into the spotlight. In the minimal twist, fivebranes are completely holomorphic objects which, in flat space, wrap three complex directions in the eleven-dimensional bulk theory
\beqn
\C^3 \subset \R \times \C^5 .
\eeqn
Recall that the wordvolume theory associated to a stack of fivebranes in $M$-theory on flat space is a superconformal theory with $\cN=(2,0)$ supersymmetry. In six dimensions (after complexifying) the superconformal algebra is the super Lie algebra $\lie{osp}(8|4)$, whose even part is $\lie{so}(8) \times \lie{sp}(4)$. Twisting involves the choice of a holomorphic supercharge $Q$ which leaves three directions invariant and breaks this super Lie algebra down to the smaller super Lie algebra $\lie{osp}(6|2)$. The $Q$-twist of any six-dimensional superconformal field theory has a symmetry by this super Lie algebra.

As we just pointed out, the super Lie algebra $\widehat {E(5|10)}$ describes the ghost system for symmetries of eleven-dimensional supergravity after the minimal twist.
In \cite{RSW}, we wrote down an explicit realization of the residual superconformal algebra $\lie{osp}(6|2)$ on $\widehat {E(5|10)}$.
In fact, $\lie{osp}(6|2)$ sits inside of a small (but still infinite-dimensional) super Lie algebra called $E(3|6)$. 

\begin{conj}[with Ingmar Saberi]
After the holomorphic twist, the six-dimensional $\cN = (2,0) $ superconformal algebra gets enhanced to the exceptional simple super Lie algebra $E(3|6)$.
As a consequence, after twisting, the space of local operators of any six-dimensional $\cN=(2,0)$ superconformal theory is a representation for $E(3|6)$.
\end{conj}

A consequence of this conjecture is that we can interpret, for example, the superconformal index as a character for $E(3|6)$. Comparison to typical formulae in the literature is faciliated by a judicious choice of Cartan for $E(3|6)$.

The superconformal deformation which takes the minimal twist to the nonminimal twist in the above diagrams is a certain superconformal transformation that is a Maurer-Cartan element in $\lie{osp}(6|2)$. This Maurer-Cartan element deforms the super Lie algebra $E(3|6)$ to a familiar object in chiral CFT: the Lie algebra of vector fields on the (formal) disk.
So, if we were to accordingly deform the above conjecture, we would recover the familiar consequence that a chiral conformal field theory has as part of its symmetries the Lie algebra of vector fields on the formal disk.
Of course, there is a richer algebra around---the Virasoro algebra---which extends the Lie algebra of vector fields on the {\em punctured} disk. 
We will say more about a six-dimensional lift of this object later on in this introduction.

\subsection{The fivebrane decomposition}
We will argue that the value of the factorization algebra associated to the worldvolume theory on a finite number of minimally twisted fivebranes on an arbitrary open set is a particular piece of the $!$-dual of the factorization algebra associated to the bulk gravitational theory.
After taking a certain limit, this argument relies on a presentation of the $!$-dual as
\beqn
(\Obs_{grav}|_{Y})^{!} \simeq \cF_{-1} \otimes \cF_0 \otimes \cF_1 \otimes \cdots 
\eeqn
for some factorization algebras $\cF_{-1},\cF_0, \cF_{1},\ldots$ which are defined using a certain weight decomposition of the bulk gravitational theory.
In fact, this decomposition is intimately related to a certain decomposition of the exceptional simple super Lie algebra $E(5|10)$ that we have already mentioned plays an important role in the minimal twist of eleven-dimensional supergravity. This decomposition is a natural lift of the grading on $\operatorname{Ham}(\C^{2})$ induced from taking the associated graded with respect to the order filtration on $\operatorname{Diff} (\C^{2})$. The strange indexing conventions will be explained in \S\ref{s:fact}.

Using this presentation, our conjecture for the value of the factorization algebra associated to the worldvolume theory on a stack of $N$ twisted fivebranes on an open set $U$ satisfies
\beqn\label{eqn:finiteTensor}
\Obs_{N-branes} (U) \simeq \cF_{-1}(U) \otimes \cF_2 (U) \otimes \cdots \otimes \cF_{N-2} (U) .
\eeqn
In other words, to obtain the space of observables supported on an open set $U$ we simply truncate the tensor decomposition at order $N$.
The surprising part is that the values of the factorization algebras $\cF_k$ are directly identifiable from the point of view of the gravitational theory!

A related theory is the minimal twist of the six-dimensional $\cN=(2,0)$ superconformal theory associated to a Lie algebra of type $A_{N-1}$.
We denote the corresponding factorization algebra by $\Obs_{A_{N-1}}$ for now. 
This is obtained from the worldvolume theory simply by throwing away the modes propagating transverse to the brane, which in our presentation above corresponds to stripping off the first factor $\cF_{-1}$. 
Thus, at the level of factorization algebras we have a similar proposed decomposition
\beqn
\label{eqn:AN-1}
\Obs_{A_{N-1}} (U) \simeq \cF_0 (U) \otimes \cF_1 (U) \otimes \cdots \otimes \cF_{N-2} (U).
\eeqn

We want to emphasize that these expressions only hold after taking this classical limit and taking some truncation of the differential present on the left-hand side.\footnote{An instructive avatar of this decomposition to keep in mind is the description of the $W_N$ vertex algebra as being generated by a collection of operators of spins $2,3,\ldots, N$.
Only in the classical limit does the $W_N$ vertex algebra decompose as a pure tensor product.}
This decomposition as a tensor product of graded vector spaces will break down at the quantum level and when we take into account the factorization algebra structure.
Nevertheless, the description provides a very effective way to compute certain protected quantities like the superconformal index.

Though the formalism we employ of Koszul duality and factorization algebras is rather abstract, it has concrete and fruitful consequences at the level of the superconformal index.
On one hand, we have the index of multi-particle states of the supergravity theory on the backreacted geometry; this is a generating function for masses of states. We denote this by $\chi_{grav}({\bf x})$, where ${\bf x}$ are fugacities given by coordinate functions on the four-dimensional Cartan subalgebra of~$\lie{osp}(6|2)$.
Typically, this multi-particle index can be obtained from the single-particle index $f_{grav}({\bf x})$ of gravitational states through the plethystic exponential. 
Since the index is only really sensitive to the minimal twist, it is reasonable that we can compute the index using our formulation of twisted eleven-dimensional supergravity.
We will obtain known formulas as a character of the Hilbert space of eleven-dimensional supergravity on $AdS_7\times S^{4}$ in \S\ref{sec:states} using our description of twisted supergravity.

Now we turn to the computation of superconformal indices of the type $A_{N-1}$ six-dimensional $\cN=(2,0)$ theory. As we have already pointed out, the superconformal index can be computed as a charater of local operators of the twisted theory. Moreover, the space of local operators is readily recovered from the data of the factorization algebra associated to the theory. We will see how the descriptions above lead to a presentation for the superconformal index $\chi_N({\bf x})$ of the worldvolume theory on a finite stack of $N$ fivebranes.

If we are working simply on the worldvolume $\C^3 \cong \R^6$, our description in \eqref{eqn:finiteTensor} posits that after taking the classical limit the decomposition of factorization algebras associted to a stack of $N$ twisted fivebranes induces a decomposition at the level of local operators
\beqn
\cF_{-1}(0) \otimes \cF_0(0) \otimes \cdots \otimes \cF_{N-2}(0) ,
\eeqn
where $\cF_k(0)$ stands for the space of local operators associated to the factorization algebra $\cF_k$.
Forgetting about the differentials for the time being, the local operators $\cF_k(0)$ are given as the free symmetric algebra on some vector space $V_k$ of linear local opeartors.
Thus, to compute the supersymmetric index of the space of local operators of $\Obs_{N-branes}$ it suffices to compute the index of each of the $V_k$---this is the single particle index---and then apply the plethystic exponential.
For now, denote by $g_k({\bf x})$ the index of the vector space $V_k$.
This quantity is directly computable from the gravitational side.

Putting all of this together, our approach is based on the observation that we can express the single particle supergravity index as
\beqn
f_{grav}({\bf x}) = \sum_{k=-1}^\infty g_k({\bf x}) .
\eeqn
Then, the avatar of our description in \eqref{eqn:finiteTensor} at the level of the single particle superconformal index is
\beqn
f_N ({\bf x}) = \sum_{k = -1}^{N-2} g_k({\bf x}) .
\eeqn
To obtain the full superconformal index we simply apply the plethystic exponential
\beqn
\chi_N({\bf x}) = \prod_{k=-1}^{N-2} {\rm PExp}\left[g_k({\bf x})\right] .
\eeqn
This is a formula for the index on a finite stack of $N$ fivebranes. 
A related quantity is the index of the superconformal field theory associated to the Lie algebra $A_{N-1}$. 
To obtain this we simply throw away contributions coming from a single fivebrane, which results in the expression
\beqn
\chi_{A_{N-1}}({\bf x}) = \prod_{k=0}^{N-2} {\rm PExp}\left[g_k({\bf x})\right] .
\eeqn
From these formulas it is manifest that the holographic relation $\chi_{grav} = \lim_{N \to \infty} \chi_N$ holds.

We highlight the expression that our prescription yields for the theory of type~$A_{1}$.

\begin{conj}\label{conj:6dtwo}
The superconformal index of the six-dimensional $\cN=(2,0)$ theory of type $A_1$ is
\[
\chi_{A_1} (y_i,y,q) = {\rm PExp} \left[f_{A_1}(y_i,y,q) \right] .
\]
where the single particle index $f_{A_1}(y_i,y,q)$ is
\beqn\label{eqn:A1}
\frac{q^4(y_1+y_2+y_3) + q^2 (y^2 + q + q^2 y^{-2}) - q^{3} (y + q y^{-1})(y_1^{-1} + y_2^{-1} + y_3^{-1})}{(1-y_1q) (1-y_2 q) (1-y_3 q)}.
\eeqn
We follow the same conventions for fugacities as in \cite{Kim:2013nva}.
\end{conj}

We next provide some comparitive evidence for our claim that $\chi_N({\bf x})$ really is the superconformal index associated to fivebranes.
In ~\S\ref{s:finite} we will show that for small values of $N$, appropriate expansions of our closed formulas agree with expressions which have been found in the literature.
We will also check that a number of specializations of the fugacities agree with certain unrefined indices which have been computed via other means.

Even stronger evidence would involve understanding what the gravitational side can say about the factorization algebras $\Obs_{N-branes}$. 
We will return to a treatment of this in future work.

\subsection{Modules and instantons}
We conclude this introduction by outlining a further consistency check of our proposal that $\Obs_{A_{N-1}}$ describes the observables for the minimally twisted type $A_{N-1}$ six-dimensional $\cN=(2,0)$ theory, this time involving dimensional reduction. The statements outlined below will be pursued further in future work. We begin by giving a flavor of the desired statement at the level of the $\Omega$-deformed nonminimal twist.

Recall that the $\Omega$-deformed nonminimal twist is a holomorphic theory in one complex dimension. We may place this theory on $\C^{\times}$ and attempt to dimensionally reduce along $S^{1}\subset \C^{\times}$. The result should be an $\Omega$-deformed A-twist of five-dimensional $\cN =2$ gauge theory for $SU(N)$- a perturbatively trivial theory. However, naively computing the dimensional reduction using a perturbative description of the fields yields an incorrect answer. This is not surprising: the gauge coupling in five-dimensions goes like the inverse of the radius of the circle we are reducing along. Since the twist still depends on the radius, the perturbative description is untrustworthy.

The precise relationship between the $\Omega$-deformed nonminimal twist of the six-dimensional $\cN=(2,0)$ $A_{N-1}$ theory and five-dimensional $\cN=2$ gauge theory for $SU(N)$ is articulated through the AGT correspondence \cite{AGT}. As we intimated, the observables of the $\Omega$-deformed nonminimal twist are expected to be a one-dimensional holomorphic factorization algebra which agrees with the $W_{N}$ vertex algebra. In the language of factorization algebras, the dimensional reduction along $S^{1}\subset \C^{\times}$ amounts to passing to the algebra of modes. The AGT correspondence posits that this mode algebra acts on the equivariant cohomology of the moduli of instantons of the five dimensional theory.

We propose similar statements at the level of the minimal twist. 
We fix a complex surface $S$ together with a pair of complex vector bundles $L,W \to S$ of ranks one and two respectively, satisfying the condition that $X = \text{Tot} (L \oplus W)$ is a Calabi--Yau fivefold.
The twisted eleven-dimensional supergravity lives on the eleven-manifold $\R \times X$. 

\subsection*{The abelian $\cN=(2,0)$ theory}

The abelian six-dimensional $\cN=(2,0)$ theory is also thought of as the worldvolume theory on a single fivebrane in $M$-theory. 
It admits a field theoretic description as the so-called $\cN=(2,0)$ tensor multiplet. 
In \cite{SWtensor} an explicit description of the minimal, holomorphic, twist of this theory is given.
At the level of this twist, the theory is defined on any complex three-fold (equipped with a square root of its canonical bundle). 
In this section, we consider the $\cN=(2,0)$ theory on the threefold $Z = \text{Tot}(L \to S)$ as a single fivebrane supported on the closed submanifold
\beqn
Z \hookrightarrow \R \times X .
\eeqn
We will denote the corresponding factorization algebra of observables on $Z$ by $\Obs_1$.

In the main body of the paper, we show that this agrees, in the classical limit, with the factorization algebra we denoted $\cF_{-1}$ above. 
We emphasize that we only consider the classical limit of this factorization algebra in this paper. 
However as this is a free theory, there is a simple description of its quantization as a twisted enveloping algebra.
From the holographic point of view, this quantum parameter corresponds to including effects from the backreaction, which we plan to do in future work.

Choosing a fiberwise hermitian metric affords a norm map 
\[
\pi \colon Z \to \R_{+}
\] 
whose fiber away from $0$ is the total space of the unit sphere bundle.
The pushforward factorization algebra $\pi_{*}\Obs_{1}$ is a stratified factorization algebra on the non-negative half-line $\R_+ = \{t \geq 0\}$. 
We can tweak this factorization algebra slightly which gives rise to a constructible factorization algebra on $\R_+$.
In turn, this is equivalent to a pair of an associative algebra (which may be thought of as the algebra of modes of $\Obs_{1}$ along the five-manifold which is the total space of the unit sphere bundle) together with a module for this algebra.

In \cite{SWtensor} it is shown that in the case where $L$ is the trivial bundle, the naive dimensional reduction (where we discard higher Kaluza--Klein modes along the circle fiber) along 
\beqn
S \times \C \to S \times \R_+
\eeqn
agrees with the factorization algebra associated to the minimal twist of five-dimensional $\cN = 2$ supersymmetric Yang--Mills theory on $\R\times M$ for the group $U(1)$ in perturbation theory around the zero instanton sector.

If we keep track of Kaluza--Klein modes, our expectation is that the full dimensional reduction $\pi_{*}\Obs_{1}|_{\R_{>0}}$ includes contributions from instanton operators.\footnote{The instanton charge turns out to be identical to the winding number around $S^1$.} Moreover, the module at zero $\pi_{*}\Obs_{1}|_{0}$ ought to be identified with the Hilbert space of the minimal twist of five-dimensional $\cN = 2$ gauge theory. We are left with the following conjecture.

\begin{conj}
\label{conj:AGT1}
Let $\operatorname{Higgs}_{GL(1)} (S,W)$ denote the moduli space of $GL(1)$ Higgs bundles on the complex surface $S$ where the Higgs fields has coefficients in the rank two bundle $W$.
The mode algebra $\pi_{*}\Obs_{1}|_{\R_{>0}}$ that we have just discussed acts on the space of functions on this moduli space
\beqn
\cO \left(\operatorname{Higgs}_{GL(1)} (S,W)\right) .
\eeqn
\end{conj}

We emphasize that the moduli space in question is really sensitive to the complex geometry of the entire fivefold $X$, at least in the neighborhood of the total space of~$L$.

While we leave the precise relationship of this conjectural action with the usual abelian AGT correspondence to future work, we can highlight some compatibilities. Above, we have mentioned a particular deformation of the observables of the six-dimensional $\cN = (2,0)$ theory by an explicit Maurer--Cartan element in the residual superconformal algebra $\lie{osp}(6|2)$. We point out that this deformation is of the same spirit considered in \cite{BeemEtAl}. For a particular choice of the fivefold $X$, one can identify (though we do not do that in this paper) this deformation of $\pi_{*} \Obs_1|_{\R_{>0}}$ with the Heisenberg algebra.

We expect that one may judiciously choose $X$ such that the superconformal element will also deform the Hilbert space of the minimal twist in five dimensions to the equivariant cohomology of the moduli of abelian instantons on $\C^{2}$, and that the action of~$\pi_{*}\Obs_{1}|_{\R_{>0}}$ compatibly deforms to the Heisenberg action studied by Grojnowski-Nakajima.

\subsection*{The $\cN=(2,0)$ theory of type $A_1$} 


The next simplest application of our holographic approach is to describe the observables of the six-dimensional theory of type~$A_1$.
Our expectation is that the factorization algebra $\Obs_{A_1}$ is related to the gravitational-side factorization algebra that we denoted by $\cF_0$ above.
Precisely, on each individual open set of the threefold $U \subset Z$, our conjecture in this paper is that the cohomology of $\cF_0(U)$ and of $\Obs_{A_1}(U)$ are the same.
To obtain the full factorization algebra structure from the gravitational side we should again include effects of the backreaction. 

Proceeding as in the abelian case, we can place the factorization algebra $\Obs_{A_1}$ on $Z$ and pushforward along the norm map to obtain an associative dg algebra $\pi_{*} \Obs_{A_1}|_{\R_{>0}}$.

\begin{conj}
Let $\operatorname{Higgs}_{SL(2)} (S,W)$ denote the moduli space of $SL(2)$ Higgs bundles on the complex surface~$S$ where the Higgs fields has coefficients in the rank two bundle~$W$.
The associative dg algebra $\pi_{*} \Obs_{A_1}|_{\R_{>0}}$ acts on
  \beqn
  \cO \left ( \operatorname{Higgs} _{SL(2)}(S, W)\right) .
  \eeqn
\end{conj}

As in the abelian case, we may highlight some compatibilities with the usual AGT correspondence. 
For a particular choice of the input geometric data, the same Maurer-Cartan element in $\lie{osp}(6|2)$ will deform the associative algebra $\pi_{*}\Obs_{A_{1}}|_{\R_{>0}}$ to the Virasoro algebra $U(\operatorname{Vir})$. We again expect that for a judicious choice of $X$, the Hilbert space deforms to the equivariant cohomology of the moduli of rank 2 instantons on $\C^{2}$, and the action compatibly deforms to the action in the rank 2 case of the AGT correspondence.

A result in the main body of the paper may be viewed as a check of these conjectures in the case where $Z = \operatorname{Tot}(\cO(-1)\to \C\PP^{2})$ and $W = K^{1/2}_{Z}\otimes \C^{2}$. In this case, the unit circle bundle in $L$ may be identified with the Hopf fibration $S^{5}\to \C\PP^{2}$. In \cite{Kim:2013nva}, the superconformal index of the six-dimensional $\cN=(2,0)$ theory is computed by identifying it with the $S^{1}\times S^{5}$ partition function, and reducing along the Hopf fiber of $S^{5}$ to further identify it with the superconformal index of five-dimensional $\cN=2$ gauge theory on $\C\PP^{2}$. The latter quantity is computed in terms of the Nekrasov instanton partition function, and is almost manifestly the character of $\cO \left ( \operatorname{Higgs} _{SL(2)}(S, W)\right)$. 
The fact that we recover the same expression as the character of the space of local operators $\Obs_{A_1}(0)$ corroborates the claim that $ \cO \left ( \operatorname{Higgs} _{SL(2)}(S, W)\right)$ is the vacuum module of $\pi_{*}\Obs_{A_1}|_{\R_{>0}}$.


\subsection*{Acknowledgements}
We thank K.~Costello, B.~Davison, D.~Gaiotto, N.~Garner, F.~Hahner, Si Li, N.~Paquette, and P.~Yoo for conversations and discussions about ideas that helped lead to the completion of this work. 
Special thanks are due to Ingmar Saberi for his collaborations on past, current, and future related projects.
The work of S.R. is supported by the Perimeter Institute for Theoretical Physics. Research at Perimeter Institute is supported in part by the Government of Canada, through the Department of Innovation, Science and Economic Development Canada, and by the Province of Ontario, through the Ministry of Colleges and Universities.
The work of B.W. was supported by Boston University and the University of Edinburgh. 
%

%
%

\section{Twists in dimensions eleven and six}
\label{s:twisted}

In this section we review the descriptions of twists of the key players in $M$-theory including eleven-dimensional supergravity as well as the six-dimensional superconformal field theory associated to the abelian Lie algebra~$\lie{gl}(1)$.

\subsection{A model for minimally twisted supergravity}


Recently in \cite{CLsugra}, Costello and Li gave a rigorous definition of twisted supergravity.
The idea is roughly the following.
In any theory of supergravity, the action of supersymmetry is gauged. 
When one treats the theory in the BRST formalism, the fields then include a collection of bosonic ghosts which homologically enforce invariance under the odd gauge symmetries.
By definition, twisted supergravity is supergravity in the background where a bosonic ghost for supersymmetry takes a nonzero value $Q$, where $Q$ is a nilpotent supercharge. 
Such backgrounds have the feature that the worldvolume theories of branes in their presence are twisted in a way that is compatible with the usual notion of twisting supersymmetric gauge theories.

Eleven-dimensional supersymmetry admits two inequivalent classes of twists characterized by the dimension of the subspace of translations spanned by the image of bracketing with the twisting supercharge.
\begin{itemize}
\item 
The minimal twist. 
This twist leaves six real directions invariant and is preserved by the subgroup $SU(5)$ of the Lorentz group. 
\item 
The non-minimal twist. 
This twist leaves nine real directions invariant and is preserved by the subgroup $SU(2) \times G_2$ of the Lorentz group. 
\end{itemize}

In this paper we will focus primarily on the minimal twist. 
In \cite{SWspinor} a complete description of the free limit of this twist of eleven-dimensional supergravity is given.
Within the BV formalism, the theory is only $\Z/2$ graded, a point we will elaborate on further soon. 
In \cite{RSW} we have given a proposal for the minimal twist of fully interacting eleven-dimensional supergravity as a $\Z/2$ graded interacting BV theory in a background with global symmetry group $SU(5)$.

The eleven-dimensional theory described in \cite{RSW} exists on any manifold which is locally of the form 
\[
S \times X 
\] 
where $X$ is a Calabi--Yau fivefold and $S$ is a real oriented smooth one-dimensional manifold.
More generally, the eleven-dimensional theory can be constructed on any eleven-manifold equipped with a transversely holomorphic foliation which is equipped with an appropriate holomorphic volume form on the leaves of the foliation.
We will elaborate on this further in \S \ref{sec:states}.

We recall the fields and equations of motion of this eleven-dimensional theory. 

%

\parsec[s:sugrafields]

Let $\T_X$ stand for the holomorphic tangent bundle of the Calabi--Yau fivefold $X$ and $\Bar{\T}_X$ its complex conjugate.
The complexified tangent bundle of the eleven-manifold $S \times X$ decomposes as
\[
\T_S \oplus \T_X \oplus \Bar{\T}_X 
\]
where $\T_S$ is the complexified tangent bundle of the smooth one-manifold $S$.
Locally, we choose coordinates $(t;z,\zbar)$ where $t$ is a smooth coordinate for $S$ and $z$ is a holomorphic coordinate for $X$. 

Momentarily we will give a full presentation of the fields and equations of motion within the Batalin--Vilkovisky (BV) formalism.
But first, we will give a more direct description of the equations of motion and gauge symmetries of the model.
 
One of the primary fields of eleven-dimensional supergravity describes deformations of the metric $g$.
Since we work in a background which has $SU(5)$ holonomy we can take the background metric on $S \times X$ to be the product of the flat metric on $S$ with a Calabi--Yau metric $g_{X}$ on $X$. 

Deformations of the background metric which survive the twist can be identified with Beltrami differentials, that is, sections of 
\[
\Bar{\T}_X^* \otimes \T_X \cong \Bar{\T}^*_X \otimes \Bar{\T}_X^* . 
\]
Here, we have used the isomorphism granted by the K\"ahler metric $g_{X}$.
In local holomorphic coordinates $\mu$ admits a description like
\[
\mu^{\Bar{i}}_j (t;z,\zbar) \d \zbar_{\Bar{i}} \otimes \del_{z_j} .
\]
More generally, the fields of this model include sections of the holomorphic tangent bundle with coefficients in Dolbeault forms on $X$ of arbitrary Dolbeault type and de Rham forms on $S$, which we can write in superfield notation as 
\[
\mu = \mu^{\Bar{I}}_j (t;z,\zbar) \d \zbar_{\Bar{I}} \otimes \del_{z_j} + \mu^{\Bar{I}}_{t,j} (t;z,\zbar) \d t \d \zbar_{\Bar{I}} \otimes \del_{z_j}
\]
where there is a sum over the multi-index $\Bar{I}$ as it ranges over subsets of $\{1,\ldots,5\}$.
In this notation, the component $\mu^{\Bar{I}}_j$ is {\em odd} if $|\Bar{I}|=0,2,4$ is even and {\em even} if $|I| = 1,3,5$ is odd.
Similarly, the component $\mu^{\Bar{I}}_{t,j}$ is {\em even} if $|\Bar{I}|=0,2,4$ is even and {\em odd} if $|I| = 1,3,5$ is odd.
In particular we have odd fields $\mu_j(t;z,\zbar) \del_{z_j}$ which play the role of ghosts for infinitesimal changes of coordinates.
We will see that the equations of motion dictate that we only see changes of holomorphic coordinates.

The field $\mu$ satisfies a condition that it be {\em divergence-free} for the Calabi--Yau structure on the fivefold~$X$.
This means that it is required to satisfy the equation
\beqn
\div \mu = 0
\eeqn
where $\div(-)$ is the divergence with respect to the holomorphic volume form $\Omega$ on $X$.
By including the full space of BV fields, this condition does arise from an action functional principle---it is not an additional constraint like the one present in the six-dimensional superconformal theory.
The field to add in this situation is a Dolbeault form $\nu \in \Omega^{0,\bu}(X)$ which imposes the equation above upon setting $\nu = 0$.
Another piece of the equations of motion for~$\mu$ dictate that it be locally constant along~$S$. 

Ordinarily, for $\mu$ to describe a deformation of complex structure in the $X$ direction we would require that it satisfy the Maurer--Cartan equation.
For twisted supergravity, however, we find a modification of this equation which involves the twisted analog of another familiar field.

The other primary bosonic field of eleven-dimensional supergravity is a three-form.
In the BV formalism we see all types of differential forms corresponding to a tower of ghosts for this higher gauge field together with antifields and antighosts.
 
On $S \times X$ the space of complex-valued differential forms decomposes into de Rham--Dolbeault forms as in
\[
\oplus_{k,p,q} \Omega^k(S) \otimes \Omega^{p,q}(X) .
\]
We refer to the homogenous pieces as the space of $(k;p,q)$ forms, these admit local presentations as
\[
\sum_{P,Q} f_{PQ} (t;z,\zbar) \d z_P \d \zbar_Q + \sum_{P',Q'} g_{P'Q'}(t;z,\zbar) \d t \d z_{P'}\d \zbar_{Q'} 
\]
where the first sum is over tuples~$P=(p_1,p_2,p_3),Q=(q_1,q_2,q_3)$ and
the second sum is over tuples~$P'=(p'_1,p'_2,p'_3),Q=(q'_1,q'_2,q'_3)$

In the minimal twist only forms with $p_i ,p_i'\leq 1$ survive. 
Thus, we have as part of the complex of fields in the twisted theory the space of forms
\[
\left(\oplus_{k,q} \Omega^k(S) \otimes \Omega^{0,q}(X)\right) \oplus 
\left(\oplus_{k',q'} \Omega^{k'}(S) \otimes \Omega^{1,q'}(X)\right) .
\]
We denote a general (non-homogenous) element in the first summand by $\beta$ and an element in the second summand by $\gamma$.
As above we can expand such fields in local coordinates using multi-indices by 
\begin{align*}
\beta & = \beta^{\Bar{I}} \d \zbar_{\Bar{I}} + \beta_t^{\Bar{J}} \d t \d \zbar_{\Bar{J}}\\
\gamma & = \gamma^{i \Bar{I}} \d z_i \d \zbar_{\Bar{I}} + \gamma_t^{j \Bar{J}} \d t \d z_j \d \zbar_{\Bar{J}}.
\end{align*}
Fields $\beta^{\Bar{I}}, \gamma_t^{i \Bar{I}}$ with $|\Bar{I}| = $odd (resp. even) are {\em even} (resp. {\em odd}) and fields $\beta_t^{\Bar{I}}, \gamma^{i \Bar{I}}$ with $|\Bar{I}| = $odd (resp. even) are {\em odd} (resp. {\em even}). 

The most important equation of motion of the eleven-dimensional theory involves both the fields $\mu$ and $\gamma$. 
When $\div \mu = 0$ it takes the form
\beqn\label{eqn:eom2}
\dbar \mu + \frac12 [\mu,\mu] = \del \gamma \del \gamma .
\eeqn
Because of the term on the right hand side this equation is not exactly the usual Beltrami equation for deformations of complex structures.
On the left hand side we are implicitly using an identification between the holomorphic tangent bundle $\T_X$ and the bundle $\wedge^4 \T^*_X$ granted by the holomorphic volume form on $X$.
This allows us to view this equation as taking place entirely in the graded space $\Omega^\bu(S) \hotimes \Omega^{4,\bu}(X)$. 
In particular, this equation is inhomogenous in form type.
As an example of the homogenous form type $(0;4,2)$ piece of this equation we have 
\[
\ep^{ijklm} \del_{\zbar_{\bar{j}}} \mu_m^{\bar{i}} + \ep^{ijkmn} \del_{z_l} (\mu_{m} \mu_{n}) =\del_{z_i} \gamma^{j \bar{i}} \del_{z_k} \gamma^{l \bar{j}} 
\]


\parsec[s:Lsugra]

We move on to give a complete description of our eleven-dimensional model within the Batalin--Vilkovisky (BV) formalism.
In the superfield notation above we have implicitly provided all components of the ghosts, fields, anti-fields, and anti-ghosts.
One of the key structures in the BV formalism is the graded skew-symmetric pairing between fields and antifields (ghosts and antighosts).
Once such a pairing $\omega_{BV}$ is introduced, one can opt to describe solutions to the equations of motion as Maurer--Cartan elements of a certain sheaf of ($\Z/2$-graded) $L_\infty$ algebras on $X \times S$.

An $L_\infty$ algebra is a graded vector space~$L=\oplus_{p\in \Z} L^p$ equipped with a collection of operations $\{[-]_k\}_{k = 1,2,\ldots}$ where 
\[
[-]_k \colon L^{\times k} \to L[2-k] 
\]
are multilinear maps which satisfy the ordinary higher Jacobi relations defining an $L_\infty$ algebra.
The shift $[2-k]$ means that $[-]_k$ is of cohomological degree $k-2$. 

The BV fields arise as the sections of a sheaf of $L_\infty$ algebras $\cL$ over the spacetime manifold $M$.
There is an overall shift in the relationship between the cohomological grading of $\cL$ and the ghost number in physics terminology. 
In cohomological degree zero of $\cL$ sit the ghosts, in cohomological degree one sit the fields, etc..
In this sense it is more accurate to refer to the BV fields as sections of the shift of the sheaf of $L_\infty$ algebras $\cL[1]$. 
We refer to \cite{CG2,ESW} for a further review of these conventions.
Together with the skew symmetric pairing $\omega_{BV}$ of cohomological degree $-1$ on $\cL[1]$, the structure maps $[-]_k$ organize together to define the full BV action which takes the general form
\beqn
\label{eqn:Sbv}
S(\Phi) = \sum_{k \geq 1} \omega_{BV} \left(\Phi , [\Phi,\ldots,\Phi]_{k}\right) ,
\eeqn
which we can think of as a functional on $\cL[1]$ of cohomological degree zero.

One feature of our proposal for the minimal twist of eleven-dimensional supergravity is that it only carries an overall $\Z/2$ grading. 
This $\Z/2$ grading totalizes the original ghost grading (which is by the group $\Z$) and the fermion number which is present in the untwisted theory.
We thus make use of $\Z/2$ graded versions of the usual BV formalism.
In this paper $\Z/2$ graded $L_\infty$ algebra means that we just have a $\Z/2$ graded vector space.
There are operations $\{[-]_k\}_{k = 1,2,\ldots}$ which satisfy the same higher Jacobi identities.
The operation $[-]_k$ is even if $k$ is even and odd if $k$ is odd.\footnote{This is not to be confused with a super $L_\infty$ algebra which is an $L_\infty$ algebra internal to the category of super vector spaces.
A $\Z/2$ graded $L_\infty$ algebra is simply what one gets when they apply the forgetful function from $\Z$ graded vector spaces to $\Z/2$ graded vector spaces.}

To build the $L_\infty$ algebra associated to the eleven-dimensional theory, we first introduce a sheaf of $\Z/2$ graded $L_\infty$ algebras on the Calabi--Yau fivefold $X$.
As a sheaf of super vector spaces, it will be given as the holomorphic sections of a super vector bundle that we denote by $L_X$. 
The even part of this super vector bundle is
\[
\T_X \oplus \C_X 
\]
where $\T_X$ is the holomorphic tangent bundle and $\C_X$ is the trivial bundle.
We will denote even sections by $(\mu, \beta)$ according to the decomposition. 
The odd part is 
\[
\C_X \oplus \T^*_X .
\]
We will denote odd elements by $(\nu, \gamma)$ according to the decomposition. 

The $L_\infty$ structure on the sheaf of holomorphic sections of $L_X$ is described as follows. 
First, $\d = [-]_1$, the differential, is simply given by 
\begin{align*}
\d \beta & = \del \beta \in \Omega^1_X \subset L_{X,-} \\
\d \mu & = \div \mu \in \cO_{X,\nu} \subset L_{X, -}
\end{align*}
and $\d \gamma = \d \nu = 0$. 
For $k \geq 2$ the general formula for the $k$-ary brackets is 
\begin{align*}
[\nu_1, \ldots, \nu_{k-2}, \mu_1,\mu_2]_{k} & = \div(\nu_1 \cdots \nu_k \mu_1 \wedge \mu_2) \in \PV^1_X \\
[\nu_1,\ldots, \nu_{k-3}, \mu_1,\mu_2,\gamma]_k & = \nu_1 \cdots \nu_{k-3} (\mu \wedge \mu') \vee \del \gamma \in \cO_{X,\beta} .\\
[\nu_1,\ldots,\nu_{k-2}, \mu, \gamma]_k & = \nu_1 \cdots \nu_{k-2} \mu \vee \del \gamma \in \Omega^1_X \\
[\gamma_1,\gamma_2] & = \Omega_X \vee (\del \gamma_1 \wedge \del \gamma_2) \in \PV^1_X .
\end{align*}
In \cite{RSW} it is shown that this endows $L_{X}$ with the structure of a sheaf of $\Z/2$ graded $L_\infty$ algebras.
There is a sub sheaf given by the sections $\mu$ and $\nu$ which is $L_\infty$-equivalent to the sheaf of holomorphic divergence-free vector fields on $X$. 

As all operations above are given in terms of holomorphic polydifferential operators, the $L_\infty$ algebra structure above induces a $\Z/2$ graded $L_\infty$ algebra structure on the Dolbeault resolution $\Omega^{0,\bu}(X,L_X)$ of the holomorphic vector bundle $L_X$.

Finally, to obtain a local $L_\infty$ algebra on $S \times X$, where $S$ denotes a real oriented smooth one-manifold, we simply tensor with the de Rham complex along~$S$.
We denote by $\cL_{sugra}$ the following local $\Z/2$ graded $L_\infty$ algebra on $X \times S$
\beqn
\cL_{sugra} \define \Omega^\bu(S) \hotimes \Omega^{0,\bu}(X,L_X) .
\eeqn

It is shown in \cite{RSW} that $\cL_{sugra}$ is equipped with a non-degenerate skew symmetric odd pairing $\omega_{BV,sugra}$ which is compatible with the $L_\infty$ structure above.
All together, this data prescribes the structure of a $\Z/2$ graded theory in the BV formalism where the BV action is as in \eqref{eqn:Sbv}.

The free part of the BV action is easy to describe
\beqn
\label{eqn:free}
\int_{S \times X} \left(\gamma (\d_{dR} + \dbar) \mu\right) \Omega_X + \int_{S \times X} (\beta (\d_{dR} + \dbar) \nu) \Omega_X + \int_{S \times X} (\beta \div \mu) \Omega_X ,
\eeqn
where $\Omega_X$ is the holomorphic volume form on~$X$.
Notice that in the first two terms no holomorphic derivatives appear in the direction of $X$ due to the presence of this holomorphic volume form.

The interacting part of the action is more complicated, partly for the reason that in our presentation above it is given by a formal series in the space of fields, rather than just a polynomial in the fields. 
Explicitly, the interaction can be written as
\beqn
\label{eqn:int}
\frac12 \int_{S \times X} (\mu^2 \del \gamma) \frac{\Omega_X}{1- \nu} + \frac16 \int_{S \times X} \gamma \del \gamma \del \gamma .
\eeqn
Notice that upon varying the field $\gamma$ and assuming that $\d_{dR} \mu = 0$ and $\nu = \div \mu = 0$ we recover the equation of motion \eqref{eqn:eom2} from this action functional. 
%

\parsec[s:e510]

In the case that $X = \C^5$ and $S = \R$ the eleven-dimensional model we have just described bears a close relationship to the exceptional simple super Lie algebra $E(5|10)$ classified in \cite{KacClass}.
The even part of $E(5|10)$ is the Lie algebra of divergence-free vector fields on the formal five-disk $\Hat{D}^5$. 
The odd part of $E(5|10)$ is the space of closed two-forms $\Omega^{2,cl}(\Hat{D}^5)$ on the formal five-disk.
There is a natural action of divergence-free vector fields on closed two forms. 
Additionally there is a Lie bracket between two closed two-forms $\alpha,\alpha'$ which produces a divergence-free vector field via the standard holomorphic volume form on the five-disk:
\beqn
[\alpha,\alpha'] = \Omega_{\Hat{D}^5}^{-1} \vee ( \alpha \wedge \alpha' ) .
\eeqn

Given a vector bundle $E \to M$, the bundle of $\infty$-jets $J^\infty E \to M$ is a $\infty$-dimensional pro vector bundle whose sections consist of $\infty$-jets of sections of $E$.
If $M = \R^d$ and $E$ is translation invariant, then the bundle of $\infty$-jets can be identified with $E_0 \times \C[[x_i]]$ where $\{x_i\}$ is a chosen coordinate on~$\R^d$. 
The process of taking $\infty$-jets is well-behaved for local Lie algebras:
the $\infty$-jets of a translation invariant local Lie algebra on $\R^d$ at a point carries the natural structure of a Lie algebra.

In \cite{RSW} it is shown that the $\infty$-jets of $\cL_{sugra}$ at $0 \in \R \times \C^5$ is quasi-isomorphic to a central extension $\Hat{E(5|10)}$ of the super Lie algebra $E(5|10)$.
The central extension is defined by the (totally even) three-linear cocycle
\beqn\label{eqn:e510central}
(\mu, \mu' , \alpha) \mapsto \<\mu \wedge \mu' , \alpha\>_{z=w=0} \in \C .
\eeqn
Since this is a three-linear functional the model we use for the central extension is a super $L_\infty$ algebra (with zero one-ary operation) rather than just a super Lie algebra.

\subsection{A more general background}
\label{s:thfmflds}
The eleven-dimensional model for twisted supergravity can be extended to more general geometries than products $S \times X$ where $S$ is a real one-dimensional manifold and $X$ is a Calabi--Yau fivefold.
More generally, we can define the model on an eleven-manifold $M$ equipped with a transversely holomorphic foliation (THF) equipped with the data of a non-vanishing volume form on the leaf space.
For general background on the theory of THFs we refer the reader to \cite{DuchampKalka, KamberTondeur, Rawnsley}.

A THF structure on a smooth manifold $M$ is an integrable subbundle $F \subset \T_M \otimes \C$ such that $F + \Bar{F} = \T_M \otimes \C$.
Equivalently, this allows you to choose an atlas of coordinates $(t_i, z_j)$ on $M$ which transform smoothly in the $t_i$-variables and holomorphically in the $z_j$-variables. 
Accordingly, we may take a local patch in a THF manifold to be of the form $\R^m \times \C^n$.
The bundle $F$ is locally spanned by the vector fields $\partial / \partial t_i$'s and $\del/\del \zbar_j$'s.
(Notice that when $F \cap \Bar{F} = 0$ we are just describing an ordinary complex structure on $M$.)
We will say that $M$ is of dimension $(m,n)$ (so it is a manifold of real dimension $m + 2n$).

There is a notion of the sheaf of `holomorphic functions' $\cO_F$ on any manifold $M$ equipped with a THF~$F$---in the language of foliations these are the functions on $M$ which are constant along the leaves of $F$.
Locally, on a coordinate patch it is given by functions which only depend on the holomorphic variables $z_i$.
There is a resolution $\cA^\bu_F$ for $\cO_F$ by locally free sheaves on $M$ which is the THF analog of the de Rham and Dolbeault complexes. 
In degree zero $\cA^0_F$ is just the sheaf of smooth functions on $M$.
In degree $k$ one takes $\cA^k_F$ the sheaf of sections of the bundle $\wedge^k F^*$.
The derivative along the leaves of the foliation defined by $F$ defines a map of sheaves
\[
\thfd \colon \cA^{q}_F \to \cA^{q+1}_F  .
\]
By integrability one has $\thfd^2 = \thfd \circ \thfd = 0$ and one can show that the natural map $\cO_F \xto{\simeq} \cA^\bu_F$ provides a resolution.
Locally in a split THF structure like $\R^m \times \C^n$ the operator $\thfd$ is of the form $\d_{dR} + \dbar$ where $\d_{dR}$ is the de Rham differential along $\R^d$ and $\dbar$ is the Dolbeault operator along $\C^n$.

The fields of our eleven-dimensional theory may be described as sheaves of flat sections of some natural partially flat bundles on $M$. Associated to a foliation $F$, its normal bundle is the quotient bundle $Q = T^{\C}_{M}/F$. The normal bundle $Q$ and its dual $Q^{*}$ admit natural connections defined along the leaves of $F$, the so-called Bott connection \cite{KamberTondeur}.

These furnish the notion of the sheaf of holomorphic vector fields $\vartheta_F$ and a sheaf of holomorphic one-forms $\Omega^1_F$ on a manifold $M$ equipped with a THF~$F$. These are the sheaves of sections of $Q$ and $Q^{*}$ that are flat with respect to the Bott connection. Locally, on $\R^m \times \C^n$ these are of the form $f^i \del / \del z_i$ and $f^{i}dz_{i}$ where each $f^i \in \cO_F$.

Let $\cA^k_{F}(Q), \cA^{k}_{F}(Q^{*})$ be the sheaves of $C^\infty$-sections of the bundles $\wedge^k F^{*} \otimes Q$ and $\wedge^{k}F^{*}\otimes Q^{*}$ respectively. Then we have fine resolutions

\beqn
\vartheta_F \xto{\simeq} \cA_F^\bu(Q) \ \ \ \ \ \ \ \Omega^1_F \xto{\simeq}  \cA_F^\bu(Q^*)
\eeqn
The sheaf $\vartheta_F$ is equipped with a Lie bracket which extends to endow $\cA^\bu(Q)$ the structure of a dg Lie algebra. Likewise, the sheaf $\Omega^{1}_{F}$ carries the contragradient action of $\vartheta_{F}$ which likewise extends to equip $\cA^{\bu}(Q^{*})$ with the structure of a dg-module over $\cA_{F}^{\bu}(Q)$.

The holomorphic deRham and divergence operators also have analogues in the THF setting, which are now defined on the leaf space. First, for a vector bundle with  $\cA_{F}^{p,q}(Q)$ denote the sheaf of $C^{\infty}$ sections of $\wedge^{q} F^{*}\otimes \wedge^p Q$.
The fields of the eleven-dimensional theory, which we phrased in terms of a mixed type of de Rham and Dolbeault cohomology, can now be elegantly repackaged in terms of the above.

For illustration, let us focus on the fields $\beta,\gamma$ which on $S \times X$ combine to form the complex
\beqn\label{eqn:drdol}
\Omega^{\bu}(S) \otimes \Omega^{0,\bu}(X) \xto{1 \otimes \del} \Omega^{\bu}(S) \otimes \Omega^{1,\bu}(X) .
\eeqn
As usual, we leave the $\d_{dR}$ and $\dbar$ operators implicit.
There is an analog of this operator $1 \otimes \del$ on a manifold $M$ equipped with a THF~$F$.
It is a differential operator 
\beqn
D \colon \cA^{p,\bu}_F \to \cA^{p+1,\bu}_F
\eeqn
which locally is just the holomorphic de Rham operator. 
This operator commutes with $\thfd$.
The generalized $\beta$ and $\gamma$ fields will be sections of the two-term complex of sheaves $\cA^\bu_F \xto{D} \cA^{1,\bu}_F$.

\newcommand\Div{D_\Omega}

To make sense of a THF analogue of the holomorphic divergence operator, we will assume that in addition to having a THF structure that $M$ is equipped with the data of a trivialization $\wedge^{n} Q^* \cong \C_{M}$ which we will denote by~$\Omega_F$. We require that the trivialization is constant along the leaves of $F$, i.e. that $D\Omega_{F} = 0$. If $F \cap \Bar{F} = 0$, then $\Omega_F$ returns the usual notion of a non-vanishing holomorphic top form. For our purposes, this is the appropriate notion of a transverse Calabi--Yau structure

Next we should say where the generalized $\mu$ and $\nu$ fields live. 
The volume form $\Omega_F \colon \C \cong \wedge^n Q^*$ determines an isomorphism $\cA^{p,q}_F \cong \cA^{n-p,q}_F$ for each $p,q$.
Via this isomorphism we can identify $D \colon \cA^{n-1,\bu}_F \to \cA^{n,\bu}_F$ with a differential operator
\beqn
\Div \colon \cA_F^{\bu}(Q) \to \cA^\bu_F .
\eeqn
The $\mu$ and $\nu$ fields will be sections of the two-term complex of sheaves $\cA^\bu_F(\theta_F) \xto{\Div} \cA^\bu_F$. 
This resolves the THF version of the sheaf of holomorphic vector fields on $M$ which preserve the form $\Omega_F$. 

Now we are in a place to describe a BV theory defined on an eleven-manifold $M$ equipped with a THF $F$ and a volume form $\Omega_F$ on the leaf space.
The space of fields of the theory consist of fields $\beta,\gamma, \mu, \nu$ as before but where
\beqn
\beta \in \Pi \cA_F^\bu , \quad \gamma \in \cA^{\bu}_F (Q^*) , \quad \mu \in \Pi \cA^\bu_F(Q), \quad \nu \in \cA_F^\bu  .
\eeqn

Notice that there is a natural odd pairing between the fields $\mu, \gamma$ and $\beta, \nu$.

For convenience, we will let $L$ denote the dg vector bundle
\beqn
\Pi (\C \oplus Q)\to Q^* \oplus \C
\eeqn

with differential given by the block matrix $\operatorname{diag}(D, \Div)$ and succinctly write the space of fields as $\cA^{\bu}_{F}(L)$.

The free part of the action functional (cf. \eqref{eqn:free}) is 
\beqn
\int_M (\gamma \thfd \mu) \Omega_F + \int_M (\beta \thfd \gamma) \Omega_F + \int_M (\beta \Div \mu) \Omega_F .
\eeqn
The interaction (cf. \eqref{eqn:int}) is
\beqn
\label{eqn:int}
\frac12 \int_{M} (\mu^2 D \gamma) \frac{\Omega_F}{1- \nu} + \frac16 \int_{M} \gamma D \gamma D \gamma .
\eeqn
If we assume that $\nu = \Div \mu = 0$ we recover the following equation of motion upon varying $\gamma$
\beqn
\thfd \mu + \frac12 [\mu,\mu] = D \gamma D \gamma .
\eeqn
This is a generalization of equation \eqref{eqn:eom2} for a general THF $F$ on the eleven-manifold~$M$.

\subsection{Twisted fivebranes}

It is expected that the worldvolume theory of fivebranes in $M$-theory on flat space is equipped with $\cN=(2,0)$ superconformal symmetry.
Six-dimensional $\cN=(2,0)$ supersymmetry admits two inequivalent classes of twists characterized by the number of directions which are left invariant:
\begin{itemize}
\item
The holomorphic, or minimal, twist.
This twist leaves three real directions invariant and is stabilized by the double cover of the group $U(3)$.
The holomorphic twist of any $\cN=(2,0)$ theory can be defined on any complex three-fold $Z$ equipped with a square-root of its canonical bundle.
\item
The non-minimal twist.
This twist leaves five real directions invariant and is stabilized by the group $SO(4) \times U(1)$.
The non-minimal twist of any $\cN=(2,0)$ theory can be defined on a six-manifold of the form $M^4 \times \Sigma$ where $M^4$ is a smooth four-manifold and $\Sigma$ is a Riemann surface.
\end{itemize}

An explicit characterization of these twists in the case of the stack of a single fivebrane has been given in \cite{SWtensor}.
We recall the description of the minimal twist.

\parsec[s:single]

The holomorphic twist of the fivebrane theory is defined on any complex three-fold $Z$ (which is not necessarily equipped with a Calabi--Yau structure).
In the particular twist we will use we must assume, however, that $Z$ is equipped with a square-root of the canonical bundle $K_Z^{1/2}$.

The theory is $\Z \times \Z/2$ graded where $\Z$ is the ghost number (or cohomological degree) and $\Z/2$ is parity.
There are four fundamental fields of ghost number zero $(\alpha, \omega, \phi_1,\phi_2)$ which consist of even Dolbeault forms of type $(2,1)$ and $(3,0)$ on $Z$:
\[
\alpha \in \Omega^{2,1}(Z), \quad \omega \in \Omega^{3,0}(Z),
\]
and a pair of odd $(0,1)$ forms twisted by the line bundle $K^{1/2}_Z$:
\[
(\phi_1,\phi_2) \in \Pi \Omega^{0,1}(Z , K^{1/2}_Z) \otimes \C^2 .
\]
The equations of motion read
\beqn
\label{eqn:eom}
\begin{split}
\del \alpha + \dbar \omega & = 0 \\
\dbar \alpha = \dbar \phi_i & = 0 .
\end{split}
\eeqn

There are gauge symmetries for the fields $\mu, \Omega$ determined by a ghost $b$ which is a Dolbeault form of type $(2,0)$ which acts simply by
\beqn
\label{eqn:ghost}
\begin{split}
\mu & \mapsto \mu + \dbar b  \\
\Omega & \mapsto \Omega + \del b .
\end{split}
\eeqn
There are also odd gauge symmetries for the odd fields $\phi_i$ given by
\beqn
\phi_i \mapsto \phi_i + \dbar \chi
\eeqn
where $\chi \in \Omega^0(Z, K_{Z}^{1/2}) \otimes \C^2$.

This theory notoriously does not admit a Lagrangian description.
However, it can still be put in a degenerate form of the BV formalism where the space of fields above is equipped with an odd Poisson bivector \cite{SWtensor}.
In the (degenerate) BV formalism, the space of fields of the theory is
\beqn
\cE_{\lie{gl}(1)} = \Omega^{\geq 2, \bu}(Z)[1] \oplus \Pi \Omega^{0,\bu}(Z, K^{1/2}_Z) \otimes \C^2 [1] ,
\eeqn
and the linear BRST differential is described by the following diagram
\beqn\label{eqn:weight-1a}
\begin{tikzcd}
\ul{-1} & \ul{0}\\
\Pi \Omega^{0,\bu}(Z, K_Z^{1/2} \otimes \C^2) & \\
\Omega^{2,\bu}(Z) \ar[r,"\del"] & \Omega^{3,\bu}(Z) ,
\end{tikzcd}
\eeqn
where the $\dbar$ operator acts implicitly on each summand.

\section{Twisted supergravity states}
\label{sec:states}

The first entry of the AdS/CFT dictionary in traditional treatments is a matching between \textit{supergravity states} and local operators in the CFT.
The goal of this section is to provide constructions of spaces of twisted supergravity states in our eleven-dimensional model, via geometric quantization. 
The state space on the twisted version of $AdS_{7}\times S^{4}$ has a remarkable property---it is naturally a module for a certain infinite-dimensional exceptional super Lie algebra. 
We conclude the section by computing characters for these modules and comparing them with large $N$ indices for fivebranes and membranes in the literature.

Before proceeding with the construction, let us first give some feel for the situation we hope to describe.
Suppose we consider a gravitational theory on $AdS_{d+1}\times S^{d^{\prime}}$, which we compactify to view as a theory on $AdS_{d+1}$ with all Kaluza-Klein harmonics included. Let $M^{d}$ denote the conformal boundary of $AdS_{d+1}$. A supergravity state is traditionally defined to be a solution to linearized equations of motion with a given boundary value \cite{WittenAdS}. 
Typically, this definition is made in situations where the relevant boundary value problem has a unique solution, in which case one may label states by the corresponding boundary values. Moreover, one may think of such boundary values as arising from modifications of a vacuum boundary condition at a point.

\subsection{Twisted Backreactions}
We begin by describing the relevant backgrounds. In eleven-dimensional supergravity, the $AdS_7 \times S^4$ background is obtained by backreacting a number of fivebranes in flat space \cite{Maldacena:1997re,WittenAdS}.
In \cite{RSW} we gave descriptions of twisted versions of this background. We will recall this construction, adapted to a slightly more global situation than considered previously.

We will consider the eleven-dimensional theory on eleven-manifolds that arise as total spaces of vector bundles. Placing the theory in the backreacted geometry is a two step procedure:

\begin{itemize}
  \item Place the eleven-dimensional theory on the complement of the zero section. 
  To do so, we will wish to describe the complement of the zero-section in a way that facilitates a description of the perturbative theory in terms of natural operations on holomorphic-topological local $L_{\infty}$-algebras.

  \item Deform the theory on the complement of the zero section by a certain Maurer--Cartan element.
  The Maurer--Cartan element is thought of as the flux sourced by branes wrapping the zero section.
\end{itemize}

This procedure is implemented at the level of the $\Omega$-deformed nonminimal twist on flat space in the appendix of \cite{CostelloM5}.

\parsec[s:brfive]

We now carry out the above procedure to backreact fivebranes in our eleven-dimensional model.
Let $Z$ be a three-fold that the fivebranes wrap.
We also fix a rank 2 holomorphic vector bundle $W\to Z$ such that $\wedge^{2} W \cong K_{Z}$;
this condition ensures that the total space of $W$ is a Calabi-Yau five-fold. In the main body of the paper we will choose $W$ to be the bundle $K_{Z}^{1/2}\otimes \C^{2}$.

Consider the bundle $\R\oplus W$; this bundle has a canonical partially flat connection. We wish to consider our eleven dimensional model on $M = \text{Tot} (\R\oplus W)$ which is the total space of the \textit{real} rank five bundle $\R\oplus W$ over $Z$. The partially flat connection on $\R\oplus W$ equips $M$ with a canonical THF structure $F_{M}\subset T_{M}^{\C}$.

We place a stack of $N$ fivebranes wrapping the zero section in $\R\oplus W$.
Denote the complement of the zero section by
\[
\mathring{M} = \text{Tot}(\R\oplus W) - 0(Z).
\]
We choose fiber coordinates of the bundle $t, w_{1}, w_{2}$ of $\R \oplus W$ over $Z$ and a fiberwise partially hermitian metric.
The corresponding norm defines a map
\begin{align*}
 h \colon  M & \to \R_{+} \\
  (t, w_{i}, \bar{w_{i}}, p)& \mapsto t^{2} + |w_{1}|^{2}+|w_{2}|^{2}
\end{align*}
Letting $\pi \colon M \to Z$ be the natural projection, we obtain the projection
\[
p \define (h,\pi) \colon M \to \R_{+}\times Z
\]
which restricts to an $S^4$ bundle $p|\mathring {M} \colon \mathring{M} \to \R_{>0} \times Z$.
These embeddings and projections fit inside of the following commutative diagram
\[
\begin{tikzcd}
\mathring{M} \ar[d,"p|\mathring M"'] \ar[r,hook] & M \ar[d,"p"] & \ar[l,hook',"0"'] Z \ar[d,"="] \\
\R_{>0} \times Z \ar[r,hook] & \R_{+} \times Z & \ar[l,hook',"0 \times \id"] Z.
\end{tikzcd}
\]
The inclusions on the left are the natural embeddings.
The top right inclusion is the zero section of $M = {\rm Tot}(\R \oplus V)$ and the bottom right inclusion is the embedding at radius $r = 0$.

As we elaborated in \S \ref{s:thfmflds}, the eleven-dimensional theory is defined on the THF manifold $\mathring M$---in the BV formalism this is encoded, in part, by the sheaf of $L_\infty$ algebras $\cL_{sugra}$ on $\mathring M$. Compactification of this theory along the $S^4$ link corresponds to pushing forward this sheaf along $p|\mathring M$. The resulting sheaf of $L_\infty$ algebras $(p|\mathring M)_*\cL_{sugra}$ describes, in the BV formalism, the compactified theory on the seven-manifold $\R_{>0} \times Z$.

We will compute the pushforward using the prescription for pushing forward a $\cA^{\bu}_{F}$-module along a compatible submersion outlined in \cite{KormanThesis}. Recall that we have a THF structure on $M$, given by an involutive subbundle $F_{M}\subset T^{\C}_{M}$, induced from a partially flat conneciton on $\R\oplus W$. The situation is summarized in the diagram below.
\[
  \begin{tikzcd}
    0 \ar[r] & F_{M}\cap T_{M/Z}\ar[r]\ar[d] & F_{M}\ar[r]\ar[d]  & \pi^{*}T^{0,1}_{Z}\ar[r]\ar[d] & 0 \\
    0 \ar[r] & T_{M/Z}\ar[r] & T_{M}\ar[r] & T_{Z}\ar[r] & 0.
  \end{tikzcd}
\]
We will assume that the top left corner $F_{M}\cap T_{M/Z}$ is again a vector bundle. Moreover, we adopt the convention that all vector bundles are complex unless otherwise mentioned.

The existence of a partially flat connection amounts to a choice of splitting \[\sigma \colon \pi^{*}T^{0,1}_{Z}\to F_{M}\] of the top row. Both $F_{M}$ and $\pi^{*}T^{1,0}_{Z}$ are involutive, and the flatness of the connection implies that $\sigma$ preserves the lie brackets on sections.

We have a similar diagram expressing that the sphere bundle $\mathring M \to \R_{>0}\times Z$ is a map of THF manifolds when both the domain and codomain are equipped with the induced THF structures.
\[
  \begin{tikzcd}
    0 \ar[r] & F_{\mathring M}\cap T_{\mathring M/\R_{>0}\times Z}\ar[r]\ar[d] & F_{\mathring M}\ar[r]\ar[d]  & h^{*}T_{\R_{>0}}\oplus\pi^{*}T^{0,1}_{Z}\ar[r]\ar[d] & 0 \\
    0 \ar[r] & T_{\mathring M/\R_{>0}\times Z}\ar[r] & T_{\mathring M}\ar[r] & h^{*}T_{\R_{>o}}\oplus T_{Z}\ar[r] & 0.
  \end{tikzcd}
\]

It follows from the assumption that $F_{M}\cap T_{M/Z}$ is a vector bundle that the top left hand corner of the above diagram is as well. Moreover, the splitting $\sigma$ can be extended to a splitting $\mathring \sigma : h^{*}T_{\R_{>0}}\oplus \pi^{*}T^{0,1}_{Z}\to F_{\mathring M}$; the extension is an adapted connection in the language of \cite{KamberTondeur}.

For each of the involutive subbundles $F$ appearing in the top row of the above diagram, we wish to describe certain  To make the burden less cumbersome, we will drop mention of the foliation and instead refer to the whether the leaves lie along the base or fiber of the submersion $p$.

Explicitly, we set
\[
  \cA^{\bu}_{\mathring M/\R_{>0}\times Z} = \cA^{\bu}_{F_{\mathring M}\cap T^{\C}_{\mathring M/ \R_{>0}\times Z}}
\]
for the complex resoling functions on $\mathring M$ constant in the directions of $F_{\mathring M}$ that are tangent to the fibers of $(p|\mathring M)$.

Likewise, we let
\[
  F_{\R_{>0}\times Z} = T_{\R_{>0}}\oplus T^{0,1}_{Z}
\] 
be the subbundle of the tangent bundle of the base locally spanned by $\del_{r}, \del_{z_{i}}$ and $(p|\mathring M)^{*}F_{\R_{>0}\times Z} = h^{*}T_{\R_{>0}}\oplus \pi^{*}T^{0,1}_{Z}$ the pullback, which is a subbundle of $T_{\mathring M}$. We will let $\cA^{\bu}_{\R_{>0}\times Z} = \cA^{\bu}_{F _{\R_{>0}\times Z}}$ denote the THF complex resolving functions constant along the leaves of this foliation on the base.

Our goal is to describe the derived pushforward $(p|\mathring M)_{*}\cL_{sugra}$ as a $\cA^{\bu}_{\R_{>0}\times Z}$-module.

We first note that the splitting $\mathring \sigma$ gives us an isomorphism of $C^{\infty}_{\mathring M}$-modules
\begin{align*}
  \cA^{\bu}_{\mathring M}(L) & \cong \cA^{\bu}_{(p|\mathring M)^{*}F_{\R_{>0}\times Z}}\otimes \cA^{\bu}_{{\mathring M/ \R_{>0}\times Z}}(L) \\
                         & \cong (p|\mathring M)^{*}\cA ^{\bu}_{\R_{>0}\times Z}\otimes \cA^{\bu}_{{\mathring M/ \R_{>0}\times Z}}(L) \\
\end{align*}
where we have used the fact that pullback commutes with taking the exterior power. Here, the tensor product is take in $C^{\infty}_{\mathring M}$-modules.


By the projection formula, there is an isomorphism of $C^{\infty}_{\R_{>0}\times Z}$-modules
\[
\cA ^{\bu}_{\R_{>0}\times Z}\otimes_{C^{\infty}_{\R_{>0}\times Z}} (p|\mathring M)_{*}\cA^{\bu}_{{\mathring M/ \R_{>0}\times Z}}(L) \to (p|\mathring M)_{*}\left ( (p|\mathring M)^{*}\cA ^{\bu}_{\R_{>0}\times Z}\otimes_{C^{\infty}_{\mathring M}} \cA^{\bu}_{{\mathring M/ \R_{>0}\times Z}}(L) \right )\]


The left hand side of the above isomorphism has a canonical $\cA ^{\#}_{\R_{>0}\times Z}$-module structure. 
It remains to discuss compatibility with the differential along the base. The compatibility condition is expressed through the observation that the pro-vector bundle \[(p|\mathring M)_{*} \cA^{\bullet}_{\mathring M/\R_{>0}\times Z}(L)\] will have the structure of a partial flat connection defined along the leaves of $F_{\R_{>0}\times Z}$. We can construct such a connection from analyzing the behavior of the THF differential $\thfd$ under the splitting.

 The THF differential induces a derivation of total degree 1 on the pushforward $(p|\mathring M)_{*}\cL_{sugra}$ so it is determined by its restrictions

\[\thfd_{i} \colon  (p|\mathring M)_{*} \cA^{\bullet}_{\mathring M/\R_{>0}\times Z}(L) \to \cA^{i}_{\R_{\geq >0} \times Z}\otimes _{C^{\infty}_{\R_{>0}\times Z}}(p|\mathring M)_{*} \cA^{\bullet-i+1}_{\mathring M/\R_{>0}\times Z}(L).\]

We may understand the first few terms, along with the relations implied by $\thfd^{2} = 0$:
\begin{itemize}
  \item The first term \[\thfd_{0} : (p|\mathring M_{*})\cA^{\bullet}_{\mathring M/\R_{>0}\times Z}(L)\to (p|\mathring M_{*})\cA^{\bullet+1}_{\mathring M/\R_{>0}\times Z}(L)\] is the internal differential.

  \item The second term \[\thfd_{0} : (p|\mathring M_{*})\cA^{\bullet}_{\mathring M/\R_{>0}\times Z}(L)\to \cA^{1}_{\R_{\geq 0}\times Z}\otimes _{C^{\infty}_{\R_{>0}\times Z}}(p|\mathring M_{*})\cA^{\bullet}_{\mathring M/\R_{>0}\times Z}(L)\] is a partial connection. However, it is not flat. Rather the failure for flatness is captured by \[ \thfd_{1}^{2} + [\thfd_{0}, \thfd_{2}] = 0 \] where \[ \thfd_{2} : (p|\mathring M_{*})\cA^{\bullet}_{\mathring M/\R_{>0}\times Z}(L)\to \cA^{2}_{\R_{\geq 0}\times Z}\otimes _{C^{\infty}_{\R_{>0}\times Z}}(p|\mathring M_{*})\cA^{\bullet-1}_{\mathring M/\R_{>0}\times Z}(L) \] is a nullhomotopy.
\end{itemize}
In particular, $\thfd_{1}$ induces a partial flat connection along $F_{\R_{>0}\times Z}$ the hypercohomology \[\mathbb{R}^\bu (p|\mathring M)_* \left (\cA^\bu_{\mathring M / \R_{>0}\times Z}(L)\right )\].

In \cite{KormanThesis} it is shown that though the above decomposition depends on the choice of splitting $\mathring \sigma$ initially chosen, the flat connection induced by $\thfd_{1}$ on the hypercohomology does not.

  In sum, we have proven the following:

\begin{prop}\label{p:dimred}
  The pushforward $(p| \mathring M)_{*}\cL_{sugra}$ is described by the following complex of sheaves
  \beqn
  \cA^\bu_{\R_{>0}\times Z}\otimes \mathbb{R}^\bu (p|\mathring M)_* \left (\cA^\bu_{\mathring M / \R_{>0}\times Z}(L) \right)
  \eeqn
with the $\cA^{\bullet}_{\R_{>0}\times Z}$-module structure coming from the partial flat connection induced by $\thfd_{1}$.
\end{prop}

The BV pairing naturally descends to the above complex of sheaves. It is induced from an explicit fiberwise pairing, which we may write as follows. First note that the fiber of $\mathbb{R}^\bu (p|\mathring M)_* \left (\cA^\bu_{\mathring M / \R_{>0}\times Z}(L) \right)$ at $(t,p)\in \R_{>0}\times Z$ can be computed as the cohomology of the complex
\[
H^{\bullet}\left (\Gamma \left (S^{4}_{(t,p)}, \cA^{\bullet}_{\mathring M/\R_{>0}\times Z}(L) \right ) ; \thfd_{0}\right ).
\]

Moreover, the deformation retraction along the fibers $(R\times \C^{2})_{(t,p)}- 0 \to S^{4}_{(t,p)}$ induces a quasi-isomorphism
\[
\left (\Gamma \left (S^{4}_{(t,p)}, \cA^{\bullet}_{\mathring M/\R_{>0}\times Z}(L) \right ) ; \thfd_{0}\right) \cong \left ( \Gamma \left ( \cA^{\bullet}_{(\R\times \C^{2})_{(t,p)} - 0}(L|_{{(\R\times \C^{2})_{(t,p)} - 0}}) \right ); \thfd \right )
\]
where the fibers $\R\times \C^{2}_{(t,p)}$ are equipped with the obvious transversely holomorphic foliation.

The cohomology of the last complex can be readily computed.

\begin{lem}
The cohomology of
\[
\left ( \Gamma \left ( \cA^{\bullet}_{(\R\times \C^{2})_{(t,p)} - 0}(L|_{{(\R\times \C^{2})_{(t,p)} - 0}}) \right ); \thfd \right ) 
\]
is the complex given by
\beqn\label{eqn:bigcplx}
\begin{tikzcd}
  \ul{even} & \ul{odd} \\ \hline
H^\bu \left ( \Gamma ( \cA^{\bu}_{(\R\times\C^2)-0}), \thfd \right )\{\partial_{w_a}, \partial_{z_i}\}  \ar[r, dotted, "D_{\Omega}"] & H^\bu \left ( \Gamma ( \cA^{\bu}_{(\R\times\C^2)-0}), \thfd \right )  \\
H^\bu \left ( \Gamma ( \cA^{\bu}_{(\R\times\C^2)-0}), \thfd \right )  \ar[r, dotted, "D"] & H^\bu \left ( \Gamma ( \cA^{\bu}_{(\R\times\C^2)-0}), \thfd \right ) \{dw_a, dz_i\}
\end{tikzcd}
\eeqn
\end{lem}

Now, the THF complex
\beqn\label{eqn:thfcohomology1}
\left ( \Gamma ( \cA^{\bu}_{(\R\times\C^2)-0}) \, , \, \thfd \right )
\eeqn
has a natural even pairing given by the residue along the four-sphere~$S^4$.
This combines with the natural odd pairing between vector fields and one-forms to give an odd pairing on the complex~\eqref{eqn:bigcplx}.
At the level of cohomology, the pairing on \eqref{eqn:thfcohomology1} is even symplectic and it is equipped with a standard polarization. 
In degree zero, the cohomology can be identified with holomorphic functions $\cO(\C^2)$, and the full cohomology is of the form $\T^*[-4] \cO(\C^2)$. 
So, as a $\Z/2$ graded object the cohomology of~\eqref{eqn:thfcohomology1} can be identified with $\T^* \cO(\C^2)$. 
This endows the cohomology of the complex \eqref{eqn:bigcplx}, which has an odd shifted symplectic structure, with a polarization.
Explicitly, we can describe this polarization at the first page in a spectral sequence which simply takes the $\thfd$-cohomology of \eqref{eqn:bigcplx} as follows. 
Let $V$ be the super vector space $\C\{\del_{w_a}, \del_{z_i}\} \oplus \Pi \C$. 
Then, this page in the spectral sequence is isomorphic to
\beqn
\Pi \T^* \left(\cO(\C^2) \otimes V\right) ,
\eeqn
where $\Pi \T^*$ denotes the odd shifted cotangent bundle.
The differential on this page of the spectral sequence, given by turning on $D, D_\Omega$ preserves this polarization.

Thus, we have a description of the compactification of $\cL_{sugra}$ on the link of fivebranes as a free holomorphic-topological BV theory on $\R_{>0}\times Z$. In particular, the result is local on $\R_{>0}\times Z$.

For computing the space of states, it will be useful to have a description of its global sections. The global sections of the sheaf in proposition \ref{p:dimred} appears naturally on the $E_{1}$-page of a Leray-Serre-type spectral sequence \cite{KamberTondeur}
\beqn
E_1^{p,q} = \Gamma \left (\R_{>0}\times Z; \cA^p_{\R_{>0}\times Z}\otimes \mathbb{R}^q(p|\mathring M) _* \left (\cA^\bu_{\mathring M / \R_{>0}\times Z}(L)\right) \right)\implies \mathbb{H}^{p+q} (\cA^\bu_{\mathring M}(L)).
\eeqn\label{thfss}
We emphasize that no page of this spectral sequence involves cohomology with respect to the operators $D, D_{\Omega}$; these differentials are kept internal.

The $E_{2}$ page is similarly given by flat sections with respect to $\thfd_{1}$
\beqn
E_2^{p,q} = \Gamma_{\thfd_1} \left (\R_{>0}\times Z; \cA^p_{\R_{>0}\times Z}\otimes \mathbb{R}^q (p|\mathring M) _* \left (\cA^\bu_{\mathring M / \R_{>0}\times Z}(L)\right) \right)
\eeqn

In \cite{RSW}, this spectral sequence was implicitly used to give a concrete description of the pushforward $(p|\mathring M)_{*}\cL_{sugra}$ in the case where
\beqn
M = \text{Tot} (\R\oplus \C^{2}\to \C^{3}), \ \ \ \ \ \ \mathring M = \R\times \C^2\times \C^3 - (0 \times \C^3).
\eeqn
In this case, the above spectral sequence clearly degenerates at the $E_{2}$-page.

To recall the result, fix holomorphic coordinates $z_{i}, i=1, 2, 3$ on $\C^{3}$ and holomorphic fiber coordinates $w_{a}, a= 1, 2$ on $\C^{2}$.

\begin{prop}
  Up to completion, the cohomology
  \beqn
  \mathbb H^{\bu} (\cA_{\mathring M}^\bu (L) ) = H^{\bu}\left (\Gamma(\cA_{\mathring M}^{\bu}(L)), \thfd\right)
  \eeqn
  is a direct sum of the cohomology on flat space $H^{\bu}(\Gamma(\cA_{M}^{\bu}(L), \thfd)$ with the complex

  \beqn
  \begin{tikzcd}[row sep = 1 ex]
    \ul{even} & \ul{odd} \\ \hline
w_1^{-1} w_2^{-1} \CC[w_1^{-1}, w_2^{-1}][z_1,z_2,z_3] \{\partial_{w_i}\}  \ar[r, dotted, "D_{\Omega}"] & w_1^{-1} w_2^{-1} \CC[w_1^{-1}, w_2^{-1}] [z_1,z_2,z_3] \\
w_1^{-1} w_2^{-1} \CC[w_1^{-1}, w_2^{-1}] [z_1,z_2,z_3] \{\del_{z_i}\} \ar[ur, dotted, "D_{\Omega}"'] \\
w_1^{-1} w_2^{-1} \CC[w_1^{-1}, w_2^{-1}] [z_1,z_2,z_3] \ar[r, dotted, "D"] \ar[dr, dotted, "D"'] & w_1^{-1} w_2^{-1} \CC[w_1^{-1}, w_2^{-1}][z_1,z_2,z_3] \{\d z_i\} \\ & w_1^{-1} w_2^{-1} \CC[w_1^{-1}, w_2^{-1}][z_1,z_2,z_3] \{\d w_i\} .
\end{tikzcd}
\eeqn
\end{prop}

\parsec[s:flux]
The second step in the construction of the backreacted geometry involves deforming the pushforward $(p| \mathring{M})_{*}\cL_{sugra}$ by a certain Maurer-Cartan element. The Maurer-Cartan element is determined by the lowest order term in the coupling between the eleven dimensional theory and the stack of fivebranes. It was argued in \cite{RSW} that the relevant coupling is given by the nonlocal interaction
\beqn\label{eqn:br1}
I_{fivebrane} = N\int_{Z} D_\Omega^{-1}\mu \vee \Omega +\cdots
\eeqn
where $\mu \in \Omega^0 (\R) \hotimes \PV^{1,3}(X)\subset \cA^{\bullet}_{F}(Q)$ is a component of a field in the eleven-dimensional theory which satisfies $\thfd_{\Omega} \mu = 0$.

The relevant Maurer-Cartan element is a solution to the equations of motion for the eleven-dimensional theory deformed by this interaction. After deforming with this interaction, varying with respect to $D_{\Omega}^{-1}\mu$ yields an equation for $\gamma$ that to linear order reads:
\beqn
\thfd D \gamma  = N \delta_Z.
\eeqn

The existence of a solution to this equation is a strong constraint on the THF manifold $M$. The relevant consequence for us is that it implies the degeneration of the spectral sequence \ref{thfss} at the $E_{2}$-page. From this point onwards, we assume that there exists a solution $\gamma_{BR}$ to the above equation.

Accordingly, the dimensional reduction takes the form

\[
\cA^{\bu}_{\R_{>0}\times Z} \otimes \left ( \Pi T^{*} (\cO (\C^{2}) \otimes V) \right)
\]

Deforming by the Maurer-Cartan element turns on a differential given by bracketing with $\thfd D\gamma_{BR}$. Thus the theory in the backreacted geometry is given by

\[
\cA^{\bu}_{\R_{>0}\times Z} \otimes \left ( \cO (\C^{2})\otimes V \xrightarrow{D\gamma_{BR}} \Pi (\cO (\C^{2})\otimes V)^{*} \right )
\]

\parsec[s:hilbertspace]
The description of the $S^{4}$ compactification of eleven dimensional supergravity obtained in the previous subsection facilitates a straightforward construction of the twisted supergravity states. 
The compactification is described by the pro-local Lie algebra $(p|\mathring M)_{*} \cL_{sugra}$ on the seven-manifold $\R_{>0}\times Z$ in proposition \ref{p:dimred}. We will construct the twisted supergravity states by applying geometric quantization ansatz to the phase space at $\infty \in \R_{>0}$.

Our ansatz for geometric quantization will avoid discussion of subtler aspects such as the metaplectic correction. We define the geometric quantization to be given by functions that are constant along the leaves of a lagrangian foliation of the phase space.

The phase space at $\infty\in \R_{>0}$ is computed by restricting to a neighborhood $(a, \infty)\subset \R_{>0}$:
\beqn
(p| \mathring M)_* \cL_{sugra}|_{\infty \times Z}= \Omega^{0,\bu}_Z \otimes \left ( \cO (\C^{2})\otimes V \xrightarrow{D\gamma_{BR}} \Pi (\cO (\C^{2})\otimes V)^{*} \right )
\eeqn

There is an obvious polarization whose base is given by
\beqn
\Omega^{0,\bullet}_Z \otimes \left (\cO (\C^{2})\otimes V \right )
\eeqn

\begin{dfn}
\label{dfn:states}
  The space of twisted supergravity states is
  \[ \cH_{sugra} = \Sym \left(\overline \Omega^{0,\bullet}_{Z, c} \otimes (\cO (\C^{2}) \otimes V)^{*} \right)\]
where the source denotes compactly supported distributional sections.
\end{dfn}

\subsection{Global symmetry for twisted $AdS$}
\label{s:global1}

After complexification, the~six-dimensional superconformal algebra is isomorphic to $\lie{osp}(8|4)$.
The even part of this algebra is $\lie{so}(8) \times \lie{sp}(4)$.
This algebra contains the six-dimensional $\cN=(2,0)$ supersymmetry algebra whose odd part is four copies of $S^{6d}_+$, the positive irreducible complex spin representation of $\lie{so}(6)$.

The holomorphic supercharge is a supertranslation
\[
Q \in \Pi S^{6d}_+ \otimes \C^4 \subset \lie{osp}(8|4)
\]
which is characterized (up to equivalence) by the properties that $Q^2 = 0$ and that its image
\[
{\rm Im}\left(Q|_{\Pi S_+ \otimes \C^4} \right) \subset \R^6 \otimes_\R \C \cong \C^6
\]
is three-dimensional (spanned by the anti-holomorphic translations). 
The supercharge $Q$ acts on $\lie{osp}(8|4)$ by commutator and the resulting cohomology will automatically act on the holomorphic twist of any six-dimensional superconformal field theory. 
This cohomology can readily be identified with the subalgebra $\lie{osp}(6|2)$, see \cite{SWsuco2}. 


In \cite{RSW} we have shown that solutions to equations of motion of our eleven-dimensional theory in the twisted $AdS_{7}\times S^{4}$ background recalled in the previous subsection contains the symmetry algebra $\lie{osp}(6|2)$. This is precisely the twist of the superconformal algebra we just discussed, so this statement may be interpreted as the twisted avatar of the fact that $\lie{osp}(8|4)$ acts as isometries on $AdS_{7}\times S^{4}.$ An easy consequence of the result of \cite{RSW} and the definition of $\cH_{sugra}$ in~\ref{dfn:states} is that the twisted symmetry algebra $\lie{osp}(6|2)$ acts on $\cH_{sugra}$.

We will enumerate states via their weights a Cartan in the bosonic subalgebra of the twisted superconformal algebra.
We recall below how the bosonic piece of the algebra
\beqn
\label{eqn:gut}
\lie{sl}(4) \times \lie{sl}(2) \subset \lie{osp}(6|2) 
\eeqn
embeds as symmetries, or ghosts, of the eleven-dimensional theory in this twisted background.

However, we first note that the state space $\cH_{sugra}$ in fact enjoys an action of a much larger symmetry algebra.

\parsec
Surprisingly, at the twisted level this symmetry by the super Lie algebra $\lie{osp}(6|2)$ is enhanced to a symmetry by an infinite-dimensional super Lie algebra called $E(3|6)$. 
This is an exceptional simple super Lie algebra and was classified by Kac in \cite{KacClass}. 
The even part $E(3|6)_{even}$ of this super Lie algebra is a semi-direct product 
\beqn\label{eqn:evenE36}
\Vect(\Hat{D}^3) \oplus \lie{sl}(2) \otimes \cO(\Hat{D}^3) 
\eeqn
where $\Vect(\Hat{D}^3)$ stands for the Lie algebra of (holomorphic) formal vector fields on the three-disk and $\cO(\Hat{D}^3)$ is the algebra of formal power series in three variables.
The odd part $E(3|6)_{odd}$ is 
\beqn\label{eqn:oddE36}
\Omega^{1}(\Hat{D}^3, K^{-1/2}) \otimes \C^2  .
\eeqn
This is the tensor product of the fundamental $\lie{sl}(2)$ representation with the space of one-forms on the formal three-disk twisted by the line bundle $K^{-1/2}$.
For a description of the Lie bracket we refer to \cite{KacClass}.

The embedding $\lie{osp}(6|2)$ into $E(3|6)$ is described as follows.
The even summand $\lie{sl}(4)$ simply embeds as the constant, linear, and quadratic vector fields which act on a formal neighborhood of $0 \in \C\mathbb{P}^3$.
The even summand $\lie{sl}(2)$ embeds as the constant power series in $\lie{sl}(2) \otimes \cO(\Hat{D}^3)$. 
The odd summand $\C^3 \otimes \C^2$ embeds as the constant one-forms and the other odd summand embeds as the linear one-forms 
\beqn
\frac12 (z_i \d z_j - z_j \d z_i) \otimes w_a
\eeqn
where, here $i \ne j$ and the fundamental $\lie{sl}(2)$ representation is spanned by $w_1,w_2$.
We are leaving out the twist by the $-1/2$ root of the canonical bundle for simplicity.

The embedding of $E(3|6)$ into our twisted description of supergravity immediately yields the following.
\begin{prop}
The state space $\cH_{sugra}$ is a representation for $E(3|6)$.
The representation restricts to the original $\lie{osp}(6|2)$ representation under the embedding above.
\end{prop}


\subsection{Mass Spectroscopy on twisted $AdS_7\times S^{4}$}
\label{s:ads7}

We move on to enumerating the state space $\cH_{sugra}$ on the twisted version of $AdS_7 \times S^4$ in terms of the index. 
We choose coordinates on flat space as
\[
\R \times \C^5 = \R_t \times \C^2_w \times \C_z^3 
\]
with $z = (z_i), i=1,2,3$ and $w = (w_a), a=1,2$.
The stack of fivebranes wrap 
\beqn
w_1=w_2=t=0 .
\eeqn
Important for us is to recall that part of the ghost system for our eleven-dimensional theory consists of divergence-free vector fields on $\C^5$ which are locally constant along $\R$. 

The subalgebra \eqref{eqn:gut} of the twisted superconformal algebra $\lie{osp}(6|2)$ embeds as ghosts in our eleven-dimensional model as follows.
\begin{itemize}
\item
The subalgebra $\lie{sl}(3) \subset \lie{sl}(4)$ embeds as vector fields rotating the plane $\C^3_z$
\beqn
\sum_{ij} A_{ij} z_i \frac{\del}{\del z_j} \in \PV^{1,0}(\C^5)\otimes \Omega^0(\R) , \quad (A_{ij}) \in \lie{sl}(3) .
\eeqn
By definition, these vector fields are automatically divergence-free.
Notice that this vector field is divergence-free and restricts to the Euler vector field along $t=w_{a} = 0$.
\item 
The subalgebra $\lie{sl}(2)$ ($R$-symmetry) is mapped to the triple
\beqn
 w_1 \frac{\del}{\del w_2}, \quad w_2 \frac{\del}{\del w_1}, \quad \frac{1}{2}\left (w_1\frac{\del}{\del w_1}-w_2\frac{\del}{\del w_2}\right) \in \PV^{1,0}(\C^5) \otimes \Omega^0(\R) .
\eeqn
\item Scaling on the plane $\C^3$ embeds as the vector field
\beqn\label{eqn:Delta}
        \Delta = \sum_{i=1}^3 z_i\frac{\del}{\del z_i} - \frac 32\sum_{a=1}^2 w_a\frac{\del}{\del w_a}\in \PV^{1,0}(\C^5)\otimes \Omega^0 (\R).
\eeqn
\end{itemize}

This describes an embedding of the algebra 
\beqn\label{eqn:gut2}
\lie{sl}(3) \times \lie{sl}(2) \times \lie{gl}(1) \subset \lie{sl}(4) \times \lie{sl}(2) 
\eeqn
into the ghosts of our model.
The dimension of a Cartan subalgebra of $\lie{sl}(3) \times \lie{sl}(2) \times \lie{gl}(1)$ is four and accordingly, the equivariant character we study has four fugacities.
We choose these explicitly as follows:
\begin{itemize}
  \item $t_{1}, t_{2}$ denote generators for the Cartan of $\lie{sl}(3)$ which is spanned by the vector fields
  \beqn
  h_1 = z_1 \frac{\del}{\del {z_1}} - z_2 \frac{\del}{\del{z_2}} , \quad h_2 = z_2 \frac{\del}{\del{z_2}} - z_3 \frac{\del}{\del{z_3}}.
  \eeqn
  \item $r$ denotes a generator for the Cartan of a $\lie{sl}(2)$ which is generated by the element 
  \beqn
  \label{eqn:hCartan}
  h = \frac12 \left(w_1 \frac{\del}{\del w_1} - w_2 \frac{\del}{\del w_2}\right) .
  \eeqn
\item $q$ denotes a generator for the Cartan of the~$\lie{gl}(1)$ which is generated by the element $\Delta$ from equation~$\eqref{eqn:Delta}$. 
\end{itemize}

The twisted supergravity states $\cH_{sugra}$ form a representation for $\lie{osp}(6|2)$. 
The weights of twisted supergravity states with respect to the generators of the Cartan subalgebra above are completely determined by the weights of the holomorphic coordinates on $\C^2_w \times \C^3_z$.
These are summarized in table \ref{tbl:sugraM5}.

\begin{table}
\begin{center}
\begin{tabular}{c c c c c c}
  & $z_{1}$ & $z_{2}$ & $z_{3}$ & $w_{1}$ & $w_{2}$ \\
  \hline
  $t_{1}$ & $1$ & 0 & $-1$ & 0 & 0 \\
  $t_{2}$ & 0 & 1 & $-1$ & 0 & 0 \\
  $r$ & 0 & 0 & 0 & 1 & $-1$ \\
  $q$ & $-1$ & $-1$ & $-1$ & $\frac{3}{2}$ & $\frac{3}{2}$
\end{tabular}
\caption{Fugacities for the fields of the holomorphic twist of eleven-dimensional supergravity for the geometry $\R \times \C^5 \setminus \C^3$.}
\label{tbl:sugraM5}
\end{center}
\end{table}

We enumerate single particle supergravity states via computing the super trace of the operator $q^Y t_1^{h_1} t_2^{h_2} r^h$ acting on $\cH_{sugra}$:
\beqn
f_{sugra}(t_1,t_2,r,q) = \Tr_{\cH_{sugra}} (-1)^F t_1^{h_1} t_2^{h_2} r^h q^\Delta .
\eeqn
The super trace means that there is an extra factor of $(-1)^F$, where $F$ is parity (fermion number), when computing the ordinary trace.

The space of supergravity states given in definition \ref{dfn:states} is given in terms of compactly supported Dolbeault forms on $\C^3$. 
On $Z = \C^3$, the compactly supported distributions defined by 
\beqn
\alpha \mapsto \del_{z_1}^{n_1} \del_{z_2}^{n_2} \del_{z_3}^{n_3} \alpha |_{z=0}
\eeqn
combine to span a dense subspace of $\Bar{\Omega}^{0,\bu}_c(\C^3)$.
Using this topological basis for distributions and the explicit description of $\cH_{sugra}$ in definition \ref{dfn:states} the following proposition is an immediate computation.

\begin{prop}
\label{prop:sugraindex1}
The single particle index of the space of twisted supergravity states~$\cH_{sugra}$ is given by the following expression
\beqn
\label{eqn:sugra_index}
f_{sugra} (t_1,t_2, r, q) = \frac{q^4(t_1^{-1}+t_1t_2^{-1}+t_2)-q^2(t_1+t_1^{-1}t_2+t_2^{-1})+(q^{3/2}-q^{9/2})(r+r^{-1})}{(1-t_{1}^{-1}q)(1-t_{2}q)(1-t_{1}t_{2}^{-1}q)(1-rq^{3/2})(1-r^{-1}q^{3/2})}.
\eeqn
The full (multiparticle) index is defined to be the plethystic exponential 
\beqn
{\rm PExp}\left[f_{sugra}(t_1,t_2,r,q)\right] .
\eeqn
\end{prop}

To simplify the form of this index we can introduce a different parametrization of the Cartan of $\lie{sl}(3) \times \lie{sl}(2) \times \lie{gl}(1)$.
First, we can parameterize the Cartan of $\lie{sl}(3)$ by the vector fields
  \beqn\label{eqn:ys}
  -(\log y_1) z_1 \frac{\del}{\del {z_1}} - (\log y_2) z_2 \frac{\del}{\del{z_2}} - (\log y_3) z_3 \frac{\del}{\del{z_3}} .
  \eeqn
where $y_1,y_2,y_3$ are parameters which satisfy the single constraint
\beqn
y_1 y_2 y_3 = 1 .
\eeqn
In terms of the variables $t_1,t_2$ used above we have
\beqn
y_1 = t_1^{-1},\quad y_2 = t_1 t_2^{-1}, \quad y_3 = t_2 .
\eeqn
Second, we can parametrize the Cartan of the remaining subalgebra $\lie{sl}(2) \times \lie{gl}(1)$ by the two vector fields
\beqn
\til{h} = h + \frac12 \Delta \quad \text{and} \quad \Delta
\eeqn
where $\Delta$ is as in equation \eqref{eqn:Delta} and $h$ is as in \eqref{eqn:hCartan}.
We denote by $y$ the generator of the Cartan corresponding to the vector field $\til{h}$ and by $q$ (as above) the generator corresponding to~$\Delta$.
In terms of the variable~$r$ used above we have 
\beqn
y = q^{1/2} r .
\eeqn

Using the paramterization fo the Cartan given by the variables $y_i,y,\Delta$ we obtain the equivalent expression for the index \eqref{eqn:sugra_index} as 
\beqn
\label{eqn:Kim_sugra}
f_{sugra} (y_i, y, q) = \frac{q^4(y_1+y_2+y_3)-q^2(y_1^{-1} + y_2^{-1} + y_3^{-1})+(1-q^3)(yq + y^{-1} q^2)}{(1-y_1 q)(1-y_2 q)(1-y_3 q)(1-yq)(1-y^{-1} q^2)},
\eeqn
We note that this matches exactly with the index computed in \cite[Eq. (3.23)]{Kim:2013nva} with the change of variables.

Our formula \eqref{eqn:sugra_index} also matches with \cite[Eq. (3.24)]{Bhattacharya:2008zy} where we use the change of variables
\beqn
q = x^4, \quad t_1 = y_2, \quad t_2 = y_1, \quad r^2 = z .
\eeqn
(Notice the variables $y_1,y_2$ used in \cite{Bhattacharya:2008zy} differ from the variables we introduced in \eqref{eqn:ys}.)

\parsec 
We consider the specialization of this index 
\beqn
q=r^2, t_2=1 
\eeqn
which is known as the Schur limit.
Applying this limit to \eqref{eqn:special1} yields the plethystic exponential of the following single particle index
\[
f_{sugra}(q, t_1, t_2=1, r = q^{1/2}) = \frac{q}{(1-q)^2}
\]
This plethystic exponential yields the MacMahon function, which is the character of the vacuum module of the $W_{1+\infty}$-algebra.

%


%
%

\section{Factorization algebras in twisted $M$ theory}
\label{s:fact}

In this section we use the formalism of factorization algebras to give a conjectural description of the space of observables of the universal fivebrane theory after performing the holomorphic twist.
For the case at hand, our expectation for the universal theory can be understood as the worldvolume theory on a large number of fivebranes. 
We will see how this relates to the description of the states in twisted supergravity as defined in the last section.
Later, we see how our interpretation is supported by a relationship of the supergravity index to the fivebrane superconformal index.

Our ansatz for the worldvolume theories in the large $N$ limit relies on a proposal for twisted holography proposed by Costello and Li in \cite{CLsugra} and further developed by Costello, Gaiotto, Paquette in \cite{CostelloM2,CostelloM5,costello2021twisted,CP}.
The core idea is that the algebra of operators on both sides of the duality are Koszul dual. 
Many other examples and support for this twisted holographic principle have been carried out in \cite{Oh:2021wes,Oh:2020hph,Gaiotto:2021xce}.

It is absolutely crucial for this proposal that we work in a derived setting using the Batalin--Vilkovisky (BV) formalism. 
The observables of any theory in the BV formalism have the structure of a factorization algebra on spacetime \cite{CG2}. 
Classically, if a theory in the BV formalism is described by a local $L_\infty$ algebra $\cL$ on spacetime $M$ then to an open set $U \subset M$ the factorization algebra of observables assigns the cochain complex
\[
\clie^\bu(\cL(U))
\]
which computes the Lie algebra cohomology of the $L_\infty$ algebra $\cL(U)$. 
The differential on this cochain complex is precisely the Chevalley--Eilenberg differential associated to the~$L_\infty$ structure maps. 

At the quantum level the factorization algebra of observables is deformed.
Perturbatively, Costello and Gwilliam use the BV formalism to give a systematic way to study quantization within the context of factorization algebras~\cite{CG2}. 
Even in perturbation theory, the full quantum behavior of our proposal for the minimal twist of eleven-dimensional theory is an open question. 
We leave the problem to characterize the quantum algebra structure to future work, and in this paper we do not address anything past tree-level in the bulk eleven-dimensional theory.

We start with the factorization algebra of classical observables of our prototype for the minimal twist of eleven-dimensional supergravity which we recalled in \S \ref{s:twisted}. 
This model is defined on any eleven-manifold $S \times X$ where $X$ is a Calabi--Yau fivefold and $S$ a real oriented one-manifold. 
The primary object we used to describe the solutions of the equations of motion was the $\Z/2$ graded local $L_\infty$ algebra $\cL_{sugra}$, see \S \ref{s:Lsugra}. 
In terms of this local $L_\infty$ algebra, the $\Z/2$ graded commutative dg algebra of classical observables supported on an open set $U \subset S \times X$ is 
\beqn\label{eqn:sugraobs}
\clie^\bu\left(\cL_{sugra}(U)\right) .
\eeqn
We will denote the classical factorization algebra on $S \times X$ of classical observables by~$\Obs_{sugra}$.

%
%
%

\subsection{The fivebrane decomposition of twisted supergravity}\label{s:resm5}

Fix a complex three-fold $Z$, where the fivebranes will be supported, and assume that it is equipped with a sqaure-root of its canonical bundle $K^{1/2}_Z$.
We consider the eleven-dimensional theory that we introduced in \S \ref{s:twisted} on $\R \times X_Z$ where $X_Z$ is the Calabi--Yau fivefold
\beqn
X_Z = \text{Tot}(K_Z^{1/2} \otimes \C^2).
\eeqn
This is the total space of the rank two holomorphic vector bundle $K_Z^{1/2} \otimes \C^2$ over~$Z$. 
One can also think about $\R \times X_Z$ as the total space of the {\em real} rank five bundle over~$Z$; although this description obfuscates the geometric structure used to define the eleven-dimensional theory.
Let 
\beqn
\pi \colon \R \times X_Z \to Z 
\eeqn
be the map which projects out the real factor followed by the projection to the base of the total space of the complex rank two bundle $p \colon X_Z \to Z$.

The classical observables of the eleven-dimensional theory form a factorization algebra on $\R \times X_Z$ that we have denoted by $\Obs_{sugra}$.
Like sheaves, factorization algebras can be pushed forward.
Thus, using $\pi$ we obtain a factorization algebra $\pi_* \Obs_{sugra}$ defined on the three-fold~$Z$.
Explicitly, to an open set $U \subset Z$ the factorization assigns the cochain complex
\begin{align*}
(\pi_* \Obs_{sugra}) (U) & = \Obs_{sugra} \left(\pi^{-1} U\right) \\
& = \Obs_{sugra}(\R \times p^{-1}U) .
\end{align*}

This factorization algebra is not the factorization algebra associated to an ordinary sort of field theory on $Z$.
In physics terminology this feature is due to the presence of an infinite tower of so-called Kaluza--Klein modes which means that the space of fields is not given as the sections of a finite rank vector bundle on~$Z$.
Mathematically, this is due to the fact that the map $\pi$ is not proper.
Nevertheless there is a subfactorization algebra $\Bar{\pi}_* \Obs_{sugra}$ which admits a natural grading so that each filtered component can be understood as such. 
This grading is determined by looking at an eigenspace decomposition of a certain compact abelian group acting on the fields of the eleven-dimensional theory on $\R \times X_Z$ which we now describe. 

Like $\Obs_{sugra}$ in equation \eqref{eqn:sugraobs}, the factorization algebra $\Bar{\pi}_* \Obs_{sugra}$ is of the form $\clie^\bu(\cG)$.
Here, $\cG$ is a sheaf of $L_\infty$ algebras on $Z$ which is presented as the $C^\infty$-sections of an $\infty$-dimensional pro vector bundle.

\parsec[]

We are in the situation of a complex three-fold $Z$ embedded as a submanifold of the eleven-manifold~$\R \times \text{Tot}(K_{Z}^{1/2} \otimes \C^2)$. 
Consider the group $U(1) \times U(1)$ rotating the fibers of the rank two bundle $K_Z^{1/2} \otimes \C^2$. 
Using this we can define the subfactorization algebra 
\beqn\label{eqn:taylor1b}
\Bar{\pi}_* \Obs_{sugra} \subset \pi_* \Obs_{sugra}
\eeqn
which to an open set~$U \subset Z$ assigns the subcomplex of observables in~$\Obs_{sugra} (\pi^{-1} U)$ which are finite sums of integral eigenvectors for this $U(1) \times U(1)$ action.

By construction, the restricted factorization algebra is of the form
\[
\Bar{\pi}_*\Obs_{sugra} = \clie^\bu(\cG_Z) 
\]
where $\cG_Z$ is an infinite-rank local $L_\infty$ algebra on $Z$. 
The embedding of factorization algebras \eqref{eqn:taylor1b} is induced by a partial Taylor expansion map
\[
\pi_* \cL_{sugra} \to \cG_Z .
\]

\parsec[s:cstarfive]

We now consider a particular decomposition of the factorization algebra $\Bar{\pi}_* \Obs_{sugra}$. 
In a local chart $U \subset Z$ we can write a generic field of the eleven-dimensional theory on $\pi^{-1} U$ in coordinates as $f(t;w,z)$ where $t \in \R$ is the real coordinate, $w=(w_1,w_2)$ is the holomorphic fiber coordinate of the rank two bundle $X_Z$ over $Z$, and $z = (z_1,z_2,z_3)$ is a local holomorphic coordinate for $U \subset Z$.

With this notation in place, we introduce the following $\C^\times$ action on the fields of the eleven-dimensional theory:
\begin{itemize}
\item On the fields $\mu(t;w,z) \in \Omega^\bu(\R) \otimes \PV^{1,\bu}(\C^2 \times U)$ the action is
\[
\lambda \cdot \mu(t;w,z) = \mu(\lambda t;\lambda w , z).
\]
\item On the fields $\nu(t;w,z) \in \Omega^\bu(\R) \otimes \PV^{0,\bu}(\C^2 \times U)$ the action is
\[
\lambda \cdot \nu(t;w,z) = \nu(\lambda t;\lambda w , z).
\]
\item On the fields $\beta(t;w,z) \in \Omega^\bu(\R) \otimes \Omega^{0,\bu}(\C^2 \times U)$ the action is
\[
\lambda \cdot \beta(t;w,z) = \lambda^{-1} \beta(\lambda t;\lambda w , z).
\]
\item On the fields $\gamma(t;w,z) \in \Omega^\bu(\R) \otimes \Omega^{1,\bu}(\C^2 \times U)$ the action is
\[
\lambda \cdot \gamma(t;w,z) = \lambda^{-1} \gamma(\lambda t;\lambda w , z).
\]
\end{itemize}

The following proposition summarizes that this $\C^\times$-action is compatible with the $L_\infty$ structure on $\cL_{sugra}$.
Its proof is an immediate computation. 

\begin{prop}
The $L_\infty$ structure on $\cL_{sugra}$ equivariant for this $\C^\times$ action. 
The odd symplectic form on $\cL_{sugra}$ is weight $-1$ for this $\C^\times$ action. 
\end{prop}

We consider this decomposition at the level of the local $L_\infty$ algebra $\cG_Z$.
Recall that the Chevalley--Eilenberg cochains of this local $L_\infty$ algebra is the restricted factorization algebra~$\Bar{\pi}_* \Obs_{sugra} = \clie^\bu(\cG_Z)$. 
For each $n \in \Z$ and open set $U \subset Z$, let 
\[
\cG_Z(U)^{(n)}\subset \cG_Z(U)
\]
be the weight $n$ eigenspace with respect to this $\C^\times$ action.
As a corollary of the proposition, we get a product decomposition 
\beqn
\label{eqn:Gdecomp}
\cG_Z = \prod_{n \geq -1} \cG^{(n)} .
\eeqn
In particular, we see that $\cG^{(0)}_Z$ is itself a local Lie algebra (that we will soon describe). 
Moreover, every $\cG^{(n)}$, $n \geq -1$ is a (local) module for this local Lie algebra.

\parsec[s:weight-1]

The first non trivial case is the weight $(-1)$ piece $\cG_Z^{(-1)}$.
We will show that
$\cG_Z^{(-1)}$ is equivalent as an abelian local Lie algebra to one of the form $\Omega^{0,\bu}(Z, \cV^{(-1)})$ where $\cV^{(-1)}$ is the holomorphic vector bundle
\[
\cV^{(-1)} = K^{1/2}_Z \otimes \C^2 \oplus \cO_Z \oplus \Pi \T^*_Z .
\]
This Dolbeault complex is equipped with the natural $\dbar$ operator and the $\del$ operator which goes from $\Omega^{0,\bu}(Z)$ to $\Omega^{1,\bu}(Z)$. 

Explicitly weight $(-1)$ summand $\cV^{(-1)}$ consists of:
\begin{itemize}
\item 
Vector fields which are locally of the form $f^a(z) \del_{w_a}$ where $f^{a}(z)$ is a holomorphic function on $Z$.
Notice that these vector fields are automatically divergence-free.
Since $\del_{w_a}$ is treated as a section of $K^{1/2}_Z$ we see that these are local sections of $K^{1/2}_Z \otimes \C^2$. 
\item 
Next there are holomorphic functions (type $\beta$) and holomorphic $one$-forms (type $\gamma$) which do not depend on $w_a$. 
Since $\del$ is weight zero for the decomposition these forms combine to form the $\Z/2$ graded complex of sheaves
\[
\cO_Z \xto{\del} \Pi \Omega^{1,hol}_Z .
\]
\end{itemize}

\parsec[s:weight0]

The weight zero summand $\cG_Z^{(0)}$ is special because it carries the induced structure of a local $L_\infty$ algebra on $Z$ inherited from the $L_\infty$ algebra $\cL_{sugra}$.
As a local Lie algebra it is equivalent to one of the form $\Omega^{0,\bu}(Z, \cV^{(0)})$.
We will prove that it is equivalent to a local Lie algebra version of the exceptional super Lie algebra $E(3|6)$. 

By local version, we mean a version of $E(3|6)$ which exists as a sheaf of super Lie algebras of the form $\Omega^{0,\bu}(Z, \cV^{(0)})$ on any complex threefold $Z$ equipped with a square-root of its canonical bundle.
The even part of this sheaf is simply a semi-direct product 
\beqn
\Vect^{hol}(Z) \oplus \sl(2) \otimes \cO_Z .
\eeqn
Notice the similarity with the even part of $E(3|6)$ in equation \eqref{eqn:evenE36}. 
The odd part is 
\beqn
\Omega^{1,hol}(Z, K^{-1/2}_Z) \otimes \C^2 
\eeqn
Again, we observe the similarities with the odd part of $E(3|6)$ given in equation \eqref{eqn:oddE36}. 
The Lie bracket on this sheaf of super vector spaces is defined analogously to the bracket on $E(3|6)$. 
In particular, the bracket operations only involve holomorphic differential operators and hence extends to a local dg super Lie algebra structure on the Dolbeault resolution of the above holomorphic vector bundles. 
We denote this local dg super Lie algebra by $\cE(3|6)$.

When $Z = \C^3$ the $\infty$-jets at $0 \in \C^3$ of the local Lie algebra $\cE(3|6)$ is equivalent to the exceptional super Lie algebra $E(3|6)$.

\begin{prop}\label{prop:v0}
As a $\Z/2$ graded local Lie algebra on $Z$, the weight zero summand~$\cG^{(0)}_Z$ is equivalent to $\cE(3|6)$. 
\end{prop}


To see that $\cG^{(0)}$ and $\cE(3|6)$ are equivalent as super vector bundles we must show that
\[
\cV^{(0)} = \T_Z \oplus \lie{sl}(2) \otimes \cO_Z \oplus \Pi \left( \T^*_Z \otimes K^{-1/2}_Z \otimes \C^2 \right).
\]
The local $L_\infty$ structure of $\cL_{sugra}$ endows the Dolbeault complex of this holomorphic vector bundle with a local Lie algebra structure.
The differential turns out to be the natural $\dbar$ operator.
We describe the bracket below.

Explicitly, we enumerate all holomorphic sections of $\cV^{(0)}$:
\begin{itemize}
\item Vector fields which are locally of the form 
\[
\mu = f^i(z) \del_{z_i} - \frac12 \del_{z_i} f^i(z) w_a \del_{w_a}
\]
and vector fields which are locally of the form
\[
\mu = g(z) A^{ab} w_a \del_{w_b}.
\]
The condition that $\mu$ be divergence-free implies that $(A^{ab}) \in \lie{sl}(2)$.
We identify the first such vector fields with sections of $\T_Z$ and the second such vector fields with sections of $\lie{sl}(2) \otimes \cO_Z$. 
\item Holomorphic functions of the form $\beta = w_a f(z)$ and holomorphic one-forms of the form $\gamma = g^{a}(z) \d w_a + h^{a} (z) w_a \d z$. 
Since $\del$ is weight zero, such sections are equipped with a differential 
\[
\Gamma^{hol}(Z, K^{-1/2}\otimes \C^2) \xto{D} \Pi \Gamma^{hol}(Z, K^{-1/2} \otimes \C^2) \oplus \Pi\Omega^{1,hol}(Z, K^{-1/2}\otimes \C^2) 
\]
which sends $w_a f(z) \mapsto (f(z) , w_a \del f (z))$. 
This $\Z/2$ graded complex of sheaves is equivalent to $\Pi\Omega^{1,hol}(Z, K^{-1/2}\otimes \C^2)$. 
\end{itemize}

The Lie bracket on $\cG_Z^{(0)}$ is defined from a Lie algebra structure on holomorphic sections of $\cV^{(0)}$ inherited from the $L_\infty$ algebra $\cL_{sugra}$. 
In fact, there is only a two-ary bracket, and an immediate calculation shows that it agrees with the Lie bracket on $\cE(3|6)$.

\parsec[s:general_decomp]

We move on to give the following general description of the weight $j$ component $\cG^{(j)}$.
Since we have already described $j = -1,0$ we focus on $j \geq 1$.

\begin{prop}
\label{prop:Vj}
Let $j \geq 1$. 
The complex of vector bundles $\cG^{(j)}$ is quasi-isomorphic to
\beqn
\Omega^{0,\bu}(Z, \cV^{(j)}) 
\eeqn
where $\cV^{(j)}$ is the super holomorphic vector bundle 
\beqn
\label{eqn:Vj}
\begin{tikzcd}
\ul{even} & \ul{odd} \\
S^{j}(\C^2) \otimes \T_Z \otimes K^{-j/2}_Z & S^{j-1}(\C^2) \otimes K^{-(j+1)/2}_Z \\
S^{j+2}(\C^2) \otimes K^{-j/2}_Z & S^{j+1}(\C^2) \otimes \T^*_Z \otimes K^{-(j+1)/2}_Z .
\end{tikzcd}
\eeqn
\end{prop}

\begin{proof}
As above, let $\C^2$ stand for the fundamental $\lie{sl}(2)$ representation and so the irreducible highest weight $j$ representation is $S^j(\C^2)$.

The weight $j \geq 1$ complex of vector bundles $\cG^{(j)}$ is readily seen to be of the following form
\beqn
\begin{tikzcd}
\ul{even} & \ul{odd} \\
\Omega^{0,\bu}(Z , \T_Z \otimes K^{-j/2}_Z) \otimes S^j(\C^2) \ar[dr, "D_1"] \\ & \Omega^{0,\bu}(Z, K^{-j/2}_Z) \otimes S^j(\C^2) \\ 
\Omega^{0,\bu}(Z , K^{-j/2}_Z) \otimes S^{j+1}(\C^2) \otimes \C^2 \ar[ur, "D_2"'] & \\
& \Omega^{1,\bu}(Z, K^{-(j+1)/2}_Z) \otimes S^{j+1}(\C^2) \\ 
\Omega^{0,\bu}(Z, K^{-(j+1)/2}_Z) \otimes S^{j+1}(\C^2) \ar[ur, "D_3"] \ar[dr, "D_4"'] \\
& \Omega^{0,\bu}(Z, K^{-(j+1)/2}_Z) \otimes S^j(\C^2) \otimes \C^2 . 
\end{tikzcd}
\eeqn

Here, as usual, the $\dbar$ operator is left implicit.
Let us describe the differentials $D_1,D_2,D_3,D_4$.
Recall that the five-fold $X = \text{Tot}(K_Z^{1/2} \otimes \C^2)$ is equipped with a non-vanishing holomorphic volume form. 
The corresponding divergence operator $\div$ restricted to the weight $j$ subspace is $\div = D_1 + D_2$. 
The differential $D_2$ is given by the identity on $\Omega^{0,\bu}(Z, K^{-j/2}_Z)$ tensored with the natural $\lie{sl}(2)$-equivariant projection
\beqn
S^{j+1}(\C^2) \otimes \C^2 \cong S^{j+2}(\C^2) \oplus S^{j}(\C^2) \twoheadrightarrow S^{j}(\C^2) .
\eeqn

The holomorphic de Rham operator on $X$ is $\del = D_3+D_4$. 
The differential $D_4$ is the identity on $\Omega^{0,\bu}(Z, K^{-(j+1)/2}$ tensored with the natural $\lie{sl}(2)$-equivariant inclusion
\beqn
S^{j+1}(\C^2) \hookrightarrow S^{j-1}(\C^2) \oplus S^{j+1}(\C^2) \cong S^j(\C^2) \otimes \C^2 .
\eeqn

There is a spectral sequence whose first term is computed by the $D_2,D_4$ cohomology. 
This term in the spectral sequence is isomorphic the complex $\Omega^{0,\bu}(Z, \cV^{(j)})$ where $\cV^{(j)}$ is as in equation \eqref{eqn:Vj}.
There are no further terms in the spectral sequence, so the result follows. 
\end{proof}

One can write down an explicit quasi-isomorphism
\beqn
\Phi \colon \Omega^{0,\bu}(Z, \cV^{(j)}) \to \cG^{(j)} 
\eeqn 
as follows.
\begin{itemize}
\item First we describe what the map looks like on the even summands.
Locally a section of $\Omega^{0,\bu}(Z, S^j(\C^2) \otimes T_Z \otimes K^{-j/2}_Z)$ is of the form
\beqn
f(w) \otimes g_i(z) \del_{z_i} 
\eeqn
where $f(w) \in \C[w_1,w_2]_j$ is a homogenous degree $j$ polynomial in the variables $w_1,w_2$ and the $g_i(z)$'s are Dolbeault forms on $Z$.
Define the divergence-free $\mu$-type field $\Phi(f(w) \otimes g_i(z) \del_{z_i})$ by the expression
\beqn
f(w) g(z) \del_{z_i} - \frac{1}{j+2} \left(\del_{z_i} g_i(z)\right) f(w) E_w
\eeqn
where $E_w = w_a \del_{w_a}$ is the Euler vector field in the direction transverse to the brane. 
\item Locally a section of $\Omega^{0,\bu}(Z, S^{j+2} \otimes K^{-j/2}_Z)$ is of the form
\beqn
f(w) \otimes g(z)  
\eeqn
where $f(w) \in \C[w_1,w_2]_{j+2}$ is a homogenous degree $j$ polynomial in the variables $w_1,w_2$ and $g(z)$ is a Dolbeault form on $Z$.
Define the divergence-free $\mu$-type field $\Phi(f(w) \otimes g(z))$ by the expression
\beqn
g(z)(\del_{w_1} f (w) \del_{w_2} - \del_{w_2} f(w) \del_{w_1}) 
\eeqn
\item Next for the odd summands.
Locally a section of $\Omega^{0,\bu}(Z, S^{j-1} \otimes K^{-(j+1)/2}_Z)$ is of the form
\beqn
f(w) \otimes g(z)  
\eeqn
where $f(w) \in \C[w_1,w_2]_{j-1}$ is a homogenous degree $j-1$ polynomial in the variables $w_1,w_2$ and $g(z)$ is a Dolbeault form on $Z$.
Define the $\gamma$-field $\Phi(f(w) \otimes g(z))$ by the expression
\beqn
\frac12 g(z)f(w)(w_1 \d w_2 - w_2 \d w_1) .
\eeqn
\item
Locally a section of $\Omega^{0,\bu}(Z, S^{j+1} \otimes T^*_Z \otimes K^{-(j+1)/2}_Z)$ is of the form
\beqn
f(w) \otimes g^i(z)  \d z_i
\eeqn
where $f(w) \in \C[w_1,w_2]_{j+1}$ is a homogenous degree $j+1$ polynomial in the variables $w_1,w_2$ and the $g^i(z)$'s are Dolbeault forms on $Z$.
Define the $\gamma$-field $\Phi(f(w) \otimes g^i(z) \d z_i)$ by the expression
\beqn
f(w) g^i (z) \d z_i .
\eeqn
\end{itemize}

The complex $\cG^{(j)}$ is manifestly a local module for the local Lie algebra $\cG^{(0)}$. 
With the identification $\cG^{(0)} \simeq \cE(3|6)$ and $\cG^{(j)} \simeq \Omega^{0,\bu}(Z, \cV^{(j)})$ we can read off this module structure completely explicitly. 

Recall that as a complex of vector bundles $\cE(3|6) = \Omega^{0,\bu}(Z, \cV^{(0)})$ where 
\beqn
\cV^{(0)} = \T_Z \oplus \lie{sl}(2) \otimes \cO_Z \oplus \Pi \left( \T^*_Z \otimes K^{-1/2}_Z \otimes \C^2 \right) .
\eeqn
The local Lie algebra structure on $\cE(3|6)$ arises from a Lie algebra structure on the sheaf of holomorphic sections of $\cV^{(0)}$.
Likewise, the $\cE(3|6)$-module structure on $\cG^{(j)} \simeq \Omega^{0,\bu}(Z, \cV^{(j)})$ arises from a holomorphic $\cV^{(0)}$-module structure on the holomorphic sections of $\cV^{(j)}$. 
By holomorphic, we mean that the structure maps are all holomorphic differential operators. 

We describe this $\cV^{(0)}$-module structure on $\cV^{(j)}$ explicitly.
Holomorphic sections of $\T_Z$ will act by Lie derivative on $\cV^{(j)}$.
Holomorphic $\sl(2)$-valued functions will also act in the natural way since each component of $\cV^{(j)}$ is labeled by an irreducible $\lie{sl}(2)$ representation. 

Finally, we need to explain how holomorphic sections of the odd component $\T^*_Z \otimes K^{-1/2}_Z \otimes \C^2$ of $\cV^{(0)}$ act.
We first give a global description of this action, and then we will write down the explicit formula in local coordinates.
For the local coordinate description recall that a general section of $\T^*_Z \otimes K^{-1/2}_Z \otimes \C^2$ has the form $w_a g^i(z) \d z_i$ where $g^i(z)$ is a holomorphic function. 
\begin{itemize}
\item The odd part of $\cV^{(0)}$ acts on the component $S^{j}(\C^2) \otimes \T_Z \otimes K^{-j/2}_Z$ through the composition
\beqn
\begin{tikzcd}
\left(\T^*_Z \otimes K^{-1/2}_Z \otimes \C^2\right) \otimes \left(S^{j}(\C^2) \otimes \T_Z \otimes K^{-j/2}_Z\right) \ar[r,"\cong"] & \left(S^{j-1}(\C^2) \oplus S^{j+1}(\C^2)\right) \otimes \left(\T_Z \otimes \T^*_Z \otimes K^{-(j+1)/2}_Z\right) \ar[dl] \ar[d] \\
S^{j-1}(\C^2) \otimes K^{-(j+1)/2}_Z & S^{j+1}(\C^2) \otimes \T^*_Z \otimes K^{-(j+1)/2}_Z 
\end{tikzcd}
\eeqn
Here, the leftmost downward arrow is the evident $\lie{sl}(2)$ projection together with the canonical pairing between sections of $\T_Z$ and $\T^*_Z$.
The rightmost downward arrow is the other $\lie{sl}(2)$ projection together with the Lie derivative of holomorphic one-forms.  
Given a local section $f(w) \otimes h_i(z) \del_{z_i}$ of $S^j(\C^2) \otimes \T_Z \otimes K_Z^{-j/2}$ an explicit formula for this action is
\begin{align*}
(w_a g^i(z) \d z_i) \cdot (f(w) \otimes h_k(z) \del_{z_k}) & = \ep_{ab} (\del_{w_b} f(w)) (g^i h_i)(z)  \\ & + w_a f(w) L_{h_k \del_{z_k}} (g^i \d z_i) .
\end{align*}
\item The odd part of $\cV^{(0)}$ acts on the component $S^{j+2}(\C^2) \otimes K^{-j/2}_Z$ through the composition
\beqn
\begin{tikzcd}
\left(\T^*_Z \otimes K^{-1/2}_Z \otimes \C^2\right) \otimes \left(S^{j+2}(\C^2) \otimes K^{-j/2}_Z\right) \ar[r,"\cong"] & \left(S^{j+1}(\C^2) \oplus S^{j+3}(\C^2)\right) \otimes \left(\T^*_Z \otimes K^{-(j+1)/2}_Z\right) \ar[d] \\
& S^{j+1}(\C^2) \otimes \T^*_Z \otimes K^{-(j+1)/2}_Z
\end{tikzcd}
\eeqn
where the downward arrow is induced by the evident $\lie{sl}(2)$ projection.
\item The odd part of $\cV^{(0)}$ acts on the component $S^{j-1}(\C^2) \otimes K^{-(j+1)/2}_Z$ through the composition
\beqn
\begin{tikzcd}
\left(\T^*_Z \otimes K^{-1/2}_Z \otimes \C^2\right) \otimes \left(S^{j-1}(\C^2) \otimes K^{-j/2}_Z\right) \ar[r,"\cong"] & \left(S^{j-2}(\C^2) \oplus S^{j}(\C^2)\right) \otimes \left(\T^*_Z \otimes K^{-1}_Z K^{-j/2}_Z\right) \ar[d] \\
& S^{j}(\C^2) \otimes \T_Z \otimes K^{-j/2}_Z
\end{tikzcd}
\eeqn
where the downward arrow is induced by the evident $\lie{sl}(2)$ projection together with the holomorphic de Rham operator taking holomorphic one-forms to holomorphic two-forms. 
\item Finally, the odd part of $\cV^{(0)}$ acts on the component $S^{j+1}(\C^2) \otimes \T^*_Z \otimes K^{-(j+1)/2}_Z$ through the composition
\beqn
\begin{tikzcd}
\left(\T^*_Z \otimes K^{-1/2}_Z \otimes \C^2\right) \otimes \left(S^{j+1}(\C^2) \otimes \T^*_Z \otimes K^{-j/2}_Z\right) \ar[r,"\cong"] & \left(S^{j}(\C^2) \oplus S^{j+2}(\C^2)\right) \otimes \left(\T^*_Z \otimes \T^*_Z \otimes K^{-1}_Z K^{-j/2}_Z\right) \ar[dl] \ar[d] \\
S^j (\C^2) \otimes \T_Z \otimes K_Z^{-j/2} & S^{j+2}(\C^2)  \otimes K^{-j/2}_Z
\end{tikzcd}
\eeqn
where the leftmost downward arrow is induced by the evident $\lie{sl}(2)$ projection.
The rightmost downward arrow is induced by the remaining $\lie{sl}(2)$ projection together with the holomorphic de Rham operator taking holomorphic two-forms to holomorphic three-forms.
\end{itemize}

\parsec[s:kacrelation]

In the case $Z = \C^3$, the product decomposition of $\cG = \cG_{\C^3}$ in equation \eqref{eqn:Gdecomp} is closely related to a decomposition of the exceptional simple super Lie algebra $E(5|10)$ studied in \cite{KR2}. 
Here, the eleven-manifold bulk is just 
\beqn
\R \times {\rm Tot}(K_{\C^3}^{1/2} \otimes \C^2) \simeq \R \times \C^5 .
\eeqn

In \S \ref{s:e510} we recalled the result from \cite{RSW} that the $\infty$-jets of $\cL_{sugra}$ at $0 \in \R \times \C^5$ is quasi-isomorphic to $\Hat{E(5|10)}$, a certain central extension of $E(5|10)$.
It follows that the $\infty$-jets of the local Lie algebra $\cG$ at $0 \in \C^3$ is also quasi-isomorphic to $\Hat{E(5|10)}$. 

In \cite{KR2} the following weight decomposition of $E(5|10)$ is constructed.
Denote by $\{z_i\}$ the local coordinates along the three-fold $Z=\C^3$ that the fivebrane wraps and $\{w_a\}$ for the transverse holomorphic coordinates to the zero section in the five-fold ${\rm Tot}(K_{\C^3}^{1/2} \otimes \C^2)$.
Assign the following weights to the super Lie algebra $E(5|10)$. 
\begin{itemize} 
\item the coordinate $z_i$ has weight zero, $\wt(z_i) = 0$. 
\item the coordinate $w_a$ has weight $+1$, $\wt(w_a) = +1$. 
\item the parity of an element carries an additional weight of $-1$. 
Thus, for example, the odd element $[\d w_1 \d z_1] \in \Omega^{2,cl}(\Hat{D}^5)$ carries weight $+1 - 1 = 0$. 
(If we think about the odd part as the space of closed two-forms then equivalently this grading translates to the one-form symbol $\d(-)$ as carrying weight $-1/2$.)
\end{itemize} 
The weight grading is concentrated in degrees $\geq -1$. 
In particular, there is a decomposition of super vector spaces
\beqn\label{eqn:decomp1}
E(5|10) = \til V_{-1} \times \prod_{n \geq 0} V_n 
\eeqn
with $\til V_{-1}$ being the weight $-1$ subspace and $V_n$ being the weight $n$ subspace for $n \geq 0$.  
It is straightforward to verify that this weight grading is compatible with the super Lie algebra structure on $E(5|10)$.

At the level of $\infty$-jets the the decomposition in equation \eqref{eqn:Gdecomp} induces a weight grading of $\Hat{E(5|10)}$ which extends the one on $E(5|10)$ that we just described by declaring that the central term have weight $-1$.
In this way, we get a related decomposition of super $L_\infty$ algebras
\beqn\label{eqn:decomp2}
\Hat{E(5|10)} = \prod_{n \geq -1} V_n
\eeqn
We will refer to this as the \textit{fivebrane decomposition} of $\Hat{E(5|10)}$.
Here $V_{-1}$ is a $\C$-extension of $\til V_{-1}$ defined in the decomposition \eqref{eqn:decomp1}.
Notice that for $n \geq 0$ the $V_n$'s are the same as in the non centrally extended case.

The decomposition in equation \eqref{eqn:decomp2} has the property that $V_0$ is isomorphic to $E(3|6)$ as super Lie algebras.
In particular, for each $n$, $V_n$ is a module for $V_0 = E(3|6)$.
In fact, each $V_n$ is an irreducible $E(3|6)$-module \cite{KR2}.

This is compatible with our description in proposition~\ref{prop:v0} where we showed that as local Lie algebras $\cG^{(0)} \simeq \cE(3|6)$. 
Indeed, the $\infty$-jets of the local Lie algebra $\cE(3|6)$ at $0 \in \C^3$ is quasi-isomorphic to $E(3|6)$. 
By the compatibility of our decomposition with \cite{KR2} we see that $V_n$ is quasi-isomorphic to the $\infty$-jets of $\cG^{(n)}$ at $0 \in \C^3$ as a module for $E(3|6)$.

\subsection{Koszul duality for factorization algebras: an ansatz}
\label{s:noether}

In quantum field theory Koszul duality naturally appears in the problem of coupling topological line operators to some ambient bulk theory. 
More generally, for higher dimensional topological defects, this problem is encoded by Koszul duality for the theory of $\EE_n$ algebras \cite{FrancisGaitsgory} \cite[\S 5.2]{LurieHA}.

More generally, we anticipate a general theory of Koszul duality for factorization algebras which should encode the problem of coupling arbitrary defects (without the condition of being topological).
Even for factorization algebras of holomorphic-topological nature this theory has not been studied in mathematics. 
Nevertheless, we will emphasize features that we expect this general form of Koszul duality to possess which will allow us to nail down its behavior on a rather general class of factorization algebras. 

In this first part of this subsection we briefly recall how Koszul duality enters in the problem of coupling line operators. 
We refer~\cite[\S 6]{CP1},~\cite[\S 8]{CG1}, or the review~\cite{PWkoszul} for more details. 
Then, we give an ansatz for Koszul duality for factorization algebras of the form $\clie^\bu(\cL)$ where $\cL$ is some local Lie algebra. 
From the point of view of the perturbative BV formalism this is not much of a restriction, all such factorization algebras of classical observables can be cast in this form. 

\parsec[s:lines]
Suppose that we have a bulk theory living on a spacetime of the form 
\[
\R \times M 
\]
where $M$ is some smooth manifold. 
Denote by $\cA$ the corresponding factorization algebra on $\R \times M$. 
The theory could have arbitrary behavior along $M$, but we assume that the theory is topological along $\R$. 
This means that when viewed as a factorization algebra on $\R$ that $\cA$ is locally constant and is hence equivalent to the data of an $\EE_1$ or $A_\infty$ algebra.

Next, assume that $\cB$ is another $\EE_1$ algebra, which we think of as being associated to some quantum mechanical system along the real line.
This is a local model for the desired line operator that we are attempting to couple to the bulk theory.
Koszul duality enters in the problem of coupling the two quantum mechanical systems $\cA$ and $\cB$---where we view $\cA$ simply as an $\EE_1$ algebra. 

A coupling of the two systems is Maurer--Cartan element in the algebra
\[
\alpha \in \cA \otimes \cB .
\]
That is, $\alpha$ is an element of ghost degree one which satisfies the Maurer--Cartan equation
\[
\delta \alpha + \alpha \star \alpha = 0 .
\]
Given such an $\alpha$ we can deform the algebra $\cA \otimes \cB$ by adding the term $[\alpha,-]$ to the differential. 
In other words, at the cochain level only the differential, not the product structure, is modified. 

In principle, there are more general ways to `couple' two $\EE_1$ algebras; generally this is controlled by the Hochschild cohomology which governs algebra deformations of $\cA \otimes \cB$. 
we will elaborate further on this definition. 
In \cite{CG1} (see also \cite{PWkoszul}) it is shown how this notion relates to the physicists description of coupling in terms of local Lagrangians.

To see Koszul duality, the key observation is that the data of the Maurer--Cartan element $\alpha$ is equivalent to the data of a map of $\EE_1$ algebras
\[
\alpha \colon \cA^! \to \cB 
\]
where $\cA^!$ is Koszul dual to the algebra $\cA$. 
We then have the following slogan: the Koszul dual of the algebra of observables of the bulk theory $\cA^!$ is the algebra of operators on the {\em universal} line defect supported on $\RR \times \{x\}$, where $x \in M$.  

\parsec[s:celine]

Before moving towards our definition of Koszul duality for a general class of factorization algebras, we briefly recast the case of duality for $\EE_1$ algebras in terms of factorization algebras. 

We will focus on a slight generalization of the standard Koszul duality between the exterior and symmetric algebras.

\begin{prop}
Let $\lie{g}$ be a Lie algebra and equip the filtered associative (and commutative) dg algebra $\clie^\bu(\fg)$ with the augmentation induced by the tautological homomorphism $0 \to \lie{g}$.
The Koszul dual of $\clie^\bu(\lie{g})$ with respect to this augmentation is equivalent to the universal enveloping algebra $U \fg$. 
\end{prop}

There are explicit models for the associative dg algebras $\clie^\bu(\lie{g})$ and $U \lie{g}$ as locally constant factorization algebras on $\R$.
First, observe that we can tensor $\fg$ with the commutative dg algebra of de Rham forms to obtain a dg Lie algebra $\fg \otimes \Omega^\bu(\R)$. 
This has a natural enhancement to a local dg Lie algebra as this is simply the smooth sections of the bundle of Lie algebras $\fg \otimes \wedge^\bu \T^*_\RR$ equipped with the de Rham operator.

Using this local Lie algebra, we obtain a model for the associative (and commutative) dg algebra $\clie^\bu(\fg)$ as the factorization algebra
\[
\clie^\bu(\fg \otimes \Omega^\bu_\R).
\]
To an open set $U \subset \R$ this produces the Chevellay--Eilenberg complex computing the Lie algebra {\em cohomology} of the dg Lie algebra $\fg \otimes \Omega^\bu(U)$---the $\fg$-valued de Rham forms on $U$.
 
Similarly, a model for $U \fg$ is the locally constant factorization algebra
\[
\clie_\bu(\fg \otimes \Omega^\bu_{\R,c}),
\]
see \cite[\S 3.4]{CG1}.
To an open set $U \subset \R$ this produces the Chevellay--Eilenberg complex computing the Lie algebra {\em homology} of the dg Lie algebra $\fg \otimes \Omega^\bu_c(U)$---the $\fg$-valued compactly supported de Rham forms on $U$.

\parsec[s:generalkoszul]

In analogy with the case of $\EE_1$, or locally constant factorization, algebras above we make the following definition. 

\begin{dfn}
Let $\cL$ be a local $L_\infty$ algebra on a manifold $M$ and consider the factorization algebra $\clie^\bu(\cL)$ which assigns to an open set $U \subset M$ the cochain complex $\clie^\bu(\cL(U))$. 
The \defterm{$!$-dual factorization algebra} is 
\[
\clie^\bu(\cL)^! \define \clie_\bu (\cL_{c}) 
\]
where $\clie_\bu(\cL_c)$ assigns to an open set $U \subset M$ the cochain complex $\clie_\bu(\cL_c(U))$.
In other words, the $!$-dual factorization algebra of $\clie^\bu(\cL)$ is the (untwisted) factorization enveloping algebra of the local Lie algebra $\cL$. 
\end{dfn} 

There are many things lacking in this definition. 
First, we do not define the $!$-dual for an arbitrary factorization algebra, only for ones of the form $\clie^\bu(\cL)$ where $\cL$ is a local Lie algebra. 
Also, we will not prove that $!$-dual satisfies any Koszul duality axioms. 
From the discussion above we see that $!$-dual does agree with Koszul duality in the case of associative algebras.\footnote{It is not difficult to see that $!$-duality for locally constant factorization algebras on $\R^n$ agrees with $\EE_n$ Koszul duality between $\clie^\bu(\fg)$, viewed as an $\EE_n$ algebra, and the $\EE_n$ enveloping algebra $U_{\EE_n} \fg$ \cite{Knudsen, Lurie}}

\parsec[s:noether]

While we don't prove that the factorization algebra $\clie^\bu(\cL)^!$ satisfies any sort of categorical duality, we point out a universality that is satisfied with reference to couplings.
In \cite[Part 3]{CG2} a factorization algebra enhancement of Noether's theorem is formulated. 
The general context is the following: 
\begin{itemize}
\item $\cB$ is the factorization algebra on spacetime $M$ of classical observables for some auxiliary theory in the BV formalism.
\item $\cL$ is a local Lie algebra on $M$ which acts on the theory by local symmetries. 
\end{itemize}
Then, the classical version of Noether's theorem for factorization algebras produces a map of factorization algebras 
\[
\clie_\bu(\cL_c) \to \cB .
\]

In our context, we imagine that the local Lie algebra $\cL$ describes a theory in the BV formalism.
The factorization algebra of classical observables is $\Obs = \clie^\bu(\cL)$. 
A natural way to couple the factorization algebras $\Obs$ and $\cB$ is to ask that $\cL$ act on $\cB$ as above.
Then, this result produces a map of factorization algebras $\Obs^! \to \cB$.
In this sense, $\Obs^!$ is the universal factorization algebra which couples to $\Obs$.

\subsection{Finite $N$ factorization algebras}
\label{sec:factsummary}

We summarize the key points of this section which will lead to a general conjecture for the factorization algebra of observables for the worldvolume theory on a finite number of fivebranes in the holomorphic twist.
In the beginning of this section we defined a factorization algebra $\Bar{\pi}_*\Obs_{sugra}$ which we think about as being the factorization algebra of observables of twisted supergravity restricted to the worldvolume of the fivebrane or membrane.
This factorization algebra is of the form 
\beqn
\label{eqn:factgrad}
\Bar{\pi}_*\Obs_{sugra} = \clie^\bu(\cG)
\eeqn
where $\cG = \cG_Z$ is a sheaf of $L_\infty$ algebras on the worldvolume~$Z$ of the fivebrane.

We have seen that $\cG$ is given as the sheaf of sections of a pro vector bundle. 
In fact, the local $L_\infty$ algebra $\cG$ came with a natural decomposition 
\[
\cG = \prod_{k \geq -1} \cG^{(k)} .
\]
There is a related local $L_\infty$ algebra 
\beqn
\til \cG \define \prod_{k \geq 0} \cG^{(k)}
\eeqn
where we simply forget the weight $-1$ component. 
These product decompositions also hold at the level of compactly supported sections. 

This weight grading on $\cG$ induces a filtration on the factorization algebra \eqref{eqn:factgrad} and on its $!$-dual 
\beqn
\left(\Bar{\pi}_*\Obs_{sugra}\right)^! = \clie_\bu(\cG_{c}).
\eeqn
We will focus just on the $!$-dual factorization algebra.
To construct the filtration, first consider the natural filtration on the local Lie algebra $\cG$ induced by the grading
\beqn
\cG = F^0 \cG \supset F^1 \cG \supset \cdots 
\eeqn
where for $N \geq 1$ we have
\beqn
F^N \cG = \cG^{(\geq N-2)} = \prod_{j \geq N-2} \cG^{(j)} ,
\eeqn
and similarly for $\til \cG$. 
If we define 
\beqn
\cG_1 \define \cG / \cG^{(\geq 0)} 
\eeqn
and for $N > 1$
\beqn\label{eqn:gN}
\cG_N \define \cG / \cG^{(\geq N-1)} , \quad \til \cG_N \define \til \cG / \cG^{(\geq N-1)} ,
\eeqn
then the associated graded local Lie algebras $\op{Gr} \cG$ and $\op{Gr} \til{\cG}$ satisfy $\op{Gr} \cG = \op{colim} \cG_N$ and $\op{Gr} \til \cG = \op{colim} \til \cG_N$.
These filtrations also hold at the level of compactly supported sections. 

The filtration on $\cG$ induces a filtration of the factorization algebra $\clie_\bu (\cG_{c})$ 
\beqn
\clie_\bu (\cG_c) = F^0 \clie_\bu (\cG_c) \supset F^1 \clie_\bu (\cG_c) \supset \cdots 
\eeqn
where $F^N \clie_\bu(\cG_c) = \clie_\bu(F^N \cG_c)$. 
Similarly, we have a filtration on the factorization algebra $\clie_\bu(\til \cG_c)$. 
At the level of the associated graded, we have
\beqn
\label{eqn:lim}
\op{Gr} \clie_\bu (\cG_{c}) = \op{colim}_{N \geq 1} \clie_\bu(\cG_{N,c}) .
\eeqn
Similarly $\op{Gr} \clie_\bu (\til \cG_{c}) = \op{colim}_{N \geq 2} \clie_\bu(\til \cG_{N,c})$.

The first term in the limit \eqref{eqn:lim} is 
\[
\clie_\bu(\cG_{1,c})
\]
where $\cG_1$ is the abelian local Lie algebra $\cG_{1} \cong \cG^{(-1)}$---this is just the weight $(-1)$ piece of the decomposition of $\cG$. 
In \cite{SWtensor}, Saberi and the second author have given an explicit description of the holomorphic twist of the worldvolume theory on a single fivebrane, that is, the six-dimensional superconformal theory associated to the abelian Lie algebra~$\lie{gl}(1)$.
We denote the corresponding factorization algebra of classical observables on the three-fold~$Z$ by~$\Obs^{cl}_1$.

We recollected the description of the holomorphic twist of the theory on a single fivebrane in~\S\ref{s:single}. 
This is a free theory and the underlying $\Z \times \Z/2$ graded cochain complex of fields $\cE_{\lie{gl(1)}}$ with linear BRST differential is
\beqn
\begin{tikzcd}
\ul{-1} & \ul{0} \\
\Omega^{2,\bu}(Z) \ar[r,"\del"] & \Omega^{3,\bu}(Z) \\
\Pi \Omega^{0,\bu}(Z, K_{Z}^{1/2} \otimes \C^2) . 
\end{tikzcd} 
\eeqn
Here we recall in the $\Z \times \Z/2$ bigrading the differential has bidegree $(1,0)$. 
The factorization algebra $\Obs_1$ is given by~$\cO(\cE_1)$.

\begin{prop}
\label{prop:factabelian}
There is a quasi-isomorphism of factorization algebras valued in $\Z/2$ graded commutative dg algebras on the three-fold~$Z$
\[
\clie_\bu(\cG_{1,c}) \xto{\simeq} \Obs^{cl}_{1} .
\]
\end{prop}

\begin{proof}
In weight $(-1)$ the abelian local Lie algebra takes the form
\[
\cG_Z^{(-1)} \simeq \Omega^{0,\bu}(Z,\cV^{(-1)}) 
\]
where $\cV^{(-1)}$ was defined in \ref{s:weight-1}.
As a sheaf of $\Z/2$ gradedcochain complexes the factorization algebra $\clie_\bu(\cG_Z^{(-1)})$ assigns to an open set $U\subset Z$ the graded symmetric algebra on the complex
\beqn\label{eqn:weight-1a}
\begin{tikzcd}
\ul{odd} & \ul{even}\\
\Omega_c^{0,\bu}(U, K_Z^{1/2} \otimes \C^2) & \\
\Omega_c^{0,\bu}(U) \ar[r,"\del"] & \Omega_c^{1,\bu}(U) .
\end{tikzcd}
\eeqn
On the other hand, as a $\Z/2$ graded cochain complex, the factorization algebra of observables of the theory on a single fivebrane is of the form 
\[
\cO(\cE_{1}(U)) = \Sym(\Pi \overline{\cE}_{1,c}^!(U)) .
\]
It is immediate to see that as a $\Z/2$ graded cochain complex $\cE_c^!(U)$ is exactly \eqref{eqn:weight-1a}.
The result then follows by applying ellipticity.
\end{proof}

We remark that the Chevalley--Eilenberg complex of an $L_\infty$ algebra $\clie_\bu(\lie{g})$ does not have the structure of a commutative dg algebra.
However, when $\lie{g}$ is abelian (so, a cochain complex) we can identify this complex with the symmetric algebra on the cochain complex $\fg^*[-1]$.
In order to see the quantum observables on a single fivebrane from our holographic analysis we must include effects from the backreaction, which we do not do here.

We now formulate an expectation about the worldvolume theory on a stack of holomorphically twisted fivebranes.
Evidence for this description will be given in the remaining parts of this paper, and we will formulate a more refined conjecture in the next section at the level of local operators.

Recall that the holomorphic, or minimal, twist of a theory with six-dimensional $\cN=(2,0)$ supersymmetry exists on any complex three-fold $Z$ equipped with $K_Z^{1/2}$.
For $N \geq 1$, let $\Obs_{N}$ be the factorization algebra of observables of the holomorphic twist of the worldvolume theory on a stack of $N$ fivebranes wrapping a threefold~$Z$. 
For each open set $U \subset Z$ our expectation is that there is an isomorphism of vector spaces
\beqn
H^\bu \left(\Obs_N(U) \right) \simeq H_\bu (\cG_{N,c}(U)) ,
\eeqn
where $H_\bu$ is the Lie algebra homology.

Similarly, for $N > 1$, let $\til \Obs_N$ be the factorization algebra of classical observables of the holomorphic twist of the worldvolume theory on a stack of $N$ fivebranes with the center of mass degrees of freedom removed.
Then, for each open set $U \subset Z$ we similarly expect
\beqn
H^\bu \left(\til\Obs_N(U) \right) \simeq H_\bu (\til \cG_{N,c}(U)) .
\eeqn
In the language of superconformal $\cN=(2,0)$ theories our conjecture is that the value of the factorization algebra $\clie_\bu(\til \cG_{N,c})$ on an open set $U$ is equivalent to the space of observables of the holomorphic twist of the six-dimensional superconformal field theory associated to the Lie algebra $\lie{sl}(N)$ supported on $U$.

We leave the study of the full factorization algebra structure present in $\Obs_N$ for future work, and emphasize here that we are only making expectations for the space of observables supported an open set.
In the next section we consider local operators, and we will formulate a refined conjecture of the space of local operators as a module for the exceptional super Lie algebra $E(3|6)$.



%
%

\section{Local operators in twisted $M$ theory}

The notion of a factorization algebra captures both the local operators of a theory together with the non-local operators that on can define from the local ones via descent.
From the data of a factorization algebra, one can recover the space of local operators by the following formal construction. 
Let $\Obs$ be the factorization algebra of observables of some theory defined on a smooth manifold $M$.
The space of local operators at point $p \in M$ is, in a precise sense, the limiting behavior of the factorization algebra evaluated on the system of open sets which contain the point~$p$. 

Generally this limit is difficult to compute, but for certain theories it is possible to give a concise expression which captures the essential features of the theory.
For example, in a holomorphic theory, the algebra of local operators is equivalent to the algebra generated by holomorphic derivatives of fields evaluated at a point.

In this section we recall the essentials of the theory of local operators for holomorphic-topological theories. 
We consider a way of counting operators in a topological-holomorphic theory, called the `local character' of a holomorphic-topological theory \cite{SWchar}, and compare it to the superconformal index.
We then present a few simple examples and then go on to set up the theory of local operators associated to factorization algebras~$\clie_\bu(\cG_{N,c})$ we constructed in \S \ref{s:fact}.

\subsection{Local operators in topological-holomorphic theories}

A factorization algebra encodes the many ways to combine observables supported on arbitrary open sets. 
Local operators, on the other hand, exist just at a point in spacetime.
From the factorization algebra perspective one can recover local operators by looking at observables which are supported on \text{every} open set which contains the given point; mathematically this is computed by a limit. 

Precisely, in \cite[Definition 10.1.0.1]{CG2} the space of local operators of a factorization algebra $\cF$ at a point $p \in M$ is defined by the limit $\cF(p) = \lim_{U \ni p} \cF(U)$ which runs over open sets $U \subset M$ containing~$p$.

We will only consider local operators on affine space $\R^d$. 
In this case, we will have the additional property that the factorization algebras are translation invariant.
At the level of local operators this means that the translation map $\tau_{p \to p'}$ induces an isomorphism $\cF(p) \simeq \cF(p')$. 
Without loss of generality, we will consider expressions for local operators at $0 \in \R^d$.

For topological-holomorphic theories the local operators take a very familiar form.
As an algebra they are generated by (holomorphic) derivatives of the fields evaluated at the specified point. 
More precisely, the local operators depend only on the $\infty$-jets of the fields at a point.
In this section we carefully formulate this result and give some examples.

\parsec[s:free]

%
%
A topological-holomorphic theory exists on spacetimes of the form $S \times X$ where $S$ is a smooth manifold and $X$ is a complex manifold (possibly equipped with some auxiliary geometric structures). 
The typical space of fields of a holomorphic-topological theory in the BV formalism is
\beqn\label{eqn:cE}
\cE = \Omega^\bu (S) \hotimes \Omega^{0,\bu}(X, V) 
\eeqn
where $V$ is a graded holomorphic vector bundle on $X$.
The underlying free theory is described by a differential on the space of fields of the form
\[
\d_{dR} + \dbar + Q^{hol} .
\]
Here $\d_{dR}$ is the de Rham differential acting on $S$, $\dbar$ is the Dolbeault operator acting on $X$, and $Q^{hol} \colon V \to V[1]$ is a holomorphic differential operator of cohomological degree~$+1$.
This means that the free, linear equations of motion for a field $\varphi$ take the form
\[
\d_{dR} \varphi + \dbar \varphi + Q^{hol} \varphi = 0 .
\]
Taking into account linear gauge symmetries corresponds to cohomology---solutions to the equations of motion modulo the image of $\d_{dR} + \dbar + Q^{hol}$.

Notice that $\cE$ is a sheaf of cochain complexes---it makes sense to restrict the fields to any open set $U \subset S \times X$. 
The factorization algebra of observables of the free theory whose fields are as above assigns to an open set $U \subset S \times X$ the cochain complex
\[
\Obs \colon U \mapsto \cO(\cE(U)) = \Sym \left(\cE(U)^\vee \right) 
\]
equipped with the induced differential.

Some remarks are in order:
\begin{itemize}
\item If $V$ is a topological vector space then $\cO(V) = \Sym(V^\vee)$ denotes the algebra of polynomials on~$V$.
Here~$V^\vee$ is the topological dual.
\item The topological dual of $\cE(U)$ is $\cE(U)^\vee \simeq \overline{\cE}^!_c(U)$ where the bar denotes distributional sections, the subscript $c$ denotes compact support, and $!$ denotes the Serre dual. 
Explicitly, if $U = U' \times U'' \subset S \times X$ then 
\[
\overline{\cE}^!_c(U' \times U'') \simeq \overline{\Omega}^\bu(U') \otimes \overline{\Omega}^{n,\bu}(U'',V^*)[n+m] 
\]
where $\dim_\C (X) = n$ and $\dim_\R (S) = m$. 
\end{itemize}

Let's restrict to the case that $S \times X = \R^m \times \C^n$ and suppose that the bundle $V \to \C^n$ is translation invariant with fiber $V_0$ over $0 \in \R^m \times \C^n$.
We also assume that the operator $Q^{hol}$ is translation invariant.   
 
The jet expansion at $0 \in \R^m \times \C^n$ determines a map of cochain complexes
\[
\cE(\C^n \times \R^m) \to V_0 [[x_i, \d x_i,z_j, \zbar_j, \d \zbar_j]] 
\]
The differential on the right hand side is $\d_{dR} + \dbar + Q^{hol} = \d x_i \del_{x_i} + \d \zbar_j \del_{\zbar_j} + Q^{hol}$ where $Q^{hol}$ is some holomorphic differential operator in the $z_j$ variables. 
Since all structure maps are given by holomorphic polydifferential operators, the canonical map 
\[
V_0 [[x_i, \d x_i,z_j, \zbar_j, \d \zbar_j]] \xto{\simeq} V_0 [[z_j]] 
\]
which sends $x_i, \d x_i,\zbar_j \d \zbar_j \mapsto 0$ is a quasi-isomorphism. 
The only remaining differential on the right hand side is~$Q^{hol}$. 
In summary, we see that the jet expansion at $0 \in \R^m \times \C^n$ determines a map of cochain complexes $\cE(\R^m \times \C^n) \to V_0[[z_j]]$. 

\begin{lem}
\label{lem:taylor}
Suppose that $\cE$ is the sheaf of cochain complexes representing the free topological-holomorphic theory on $S \times X = \R^m \times \C^n$ and consider the factorization algebra of observables~$\Obs = \cO (\cE)$. 
Then, the Taylor expansion map
\beqn\label{eqn:taylor}
\cE(\C^n \times \R^m) \to V_0[[z_0,\ldots,z_n]]
\eeqn
induces a quasi-isomorphism of commutative dg algebras
\[
\Obs(0) \simeq \cO \left( V_0[[z_1,\ldots,z_n]] \right) .
\]
\end{lem}
\begin{proof}
Suppose that $D_\R \times D_\C \subset \R^m \times \C^n$ is a product of a real $m$-disk times a complex $n$-disk containing the origin.
The algebra of observables supported on $D_\R \times D_\C$ is quasi-isomorphic to 
\[
\cO\left( \cO^{hol}(D_\C) \otimes V_0 \right) .
\]

Observe that there is a canonical map on fields 
\[
\cO^{hol}(D_\C) \otimes V_0 \to V_0[[z_1,\ldots,z_n]]
\]
given by taking the power series expansion at $0 \in D_\R \times D_\C$. 
If an observables on $D_\R \times D_\C$ depends on only the value of the field and its derivatives at $0 \in D_\R \times D_\C$ then it automatically factors through this map. 
In particular, this means that there is a quasi-isomorphism of local operators with functions on $V_0[[z_1,\ldots,z_n]]$,
\[
\Obs(0) \simeq \cO\left(V_0 [[z_1,\ldots,z_n]]\right).
\] 
\end{proof}

Let's unpack this result explicitly. 
Using the $n$-dimensional residue, we can identify the topological dual of $V_0[[z_1,\ldots,z_n]]$ with the vector space
\beqn
\frac{\d z_1}{z_1} \cdots \frac{\d z_n}{z_n} V_0^* [z_0^{-1}, \ldots,z_n^{-1}] .
\eeqn
This is the space of linear local operators. 
If $\chi \colon V_0 \to \C$ is a dual vector in~$V_0^*$
then we obtain a linear local operator at $0 \in \R^m \times \C^n$ on the space of fields by the assignment
\[
\varphi \mapsto \del_{z_1}^{k_1} \cdots \del_{z_n}^{k_n} \<\chi,\varphi\> (0) 
\]
where $k_i \geq 0$. 
Under the quasi-isomorphism of the lemma above, this corresponds to the linear local operator 
\[
\frac{\d z_1}{z_1^{k_1+1}} \cdots \frac{\d z_n}{z_n^{k_n+1}} \chi .
\]

\parsec[s:interaction]

It is not hard to turn on interactions in the description above. 
An interacting theory in the BV formalism is described by a local $L_\infty$ algebra structure on $\cL = \cE[-1]$, where $\cE$ is the sheaf of fields.
For a topological-holomorphic theory the higher $L_\infty$ structure maps $[\cdot]_k$ of the local $L_\infty$ algebra are required to be given by holomorphic polydifferential operators and $[\cdot]_1 = \d_{dR} + \dbar + Q^{hol}$.  
For more details we refer to the definitions in \cite{GRWthf}.

In this situation, the factorization algebra of classical observables supported on an open set $U \subset S \times X$ is given by the Chevalley--Eilenberg cochains on the $L_\infty$ algebra $\cL(U)$. 
This defines a factorization algebra 
\[
\Obs \colon U \mapsto \clie^\bu(\cL(U)) .
\]
We will now give a concise presentation for the {\em local} operators in a topological-holomorphic theory. 

On $S \times X = \R^m \times \C^n$ we can also ask that all $L_\infty$ structure maps be translation invariant. 
If this is the case, one obtains the induced structure of an $L_\infty$ algebra on the (shift of the) jets of the fields supported at $0 \in \R^m \times \C^n$
\[
V_0 [[z_1,\ldots,z_n]] [-1] .
\]
The $[\cdot]_1$ operation is precisely $Q^{hol}$ as above.
The Taylor expansion map \eqref{eqn:taylor} is a map of $L_\infty$ algebras. 
Combining this with Lemma \ref{lem:taylor}, one gets a quasi-isomorphism of cochain complexes between the local operators of an interacting topological-holomorphic theory in terms of Lie algebra cohomology
\[
\Obs(0) \simeq \clie^\bu\left(V_0[[z_1,\ldots,z_n]][-1]\right) .
\]


\parsec[s:envelope]

There is another way that observables are presented in a degenerate version of the BV formalism.
Suppose that~$\cE$ is the sheaf of sections of some graded vector bundle~$E$ on a manifold~$M$.
We have seen that the observables~$\cO(\cE) = \Sym(\cE^*)$ has the structure of a factorization algebra---we now consider the $!$-dual factorization algebra.
That is, we consider the factorization algebra 
\[
U \subset M \mapsto \Sym \left(\cE_c(U) \right) 
\]
where $U \to \cE_c(U)$ is the cosheaf of compactly supported sections of the bundle~$E$.

%

\begin{lem}
\label{lem:envelope}
Suppose that $\cE$ is the sheaf of fields of a free holomorphic theory as in~\eqref{eqn:cE} and consider the factorization algebra~$\cF = \Sym(\cE_c)$. 
Then, the algebra of classical local operators at~$0 \in \C^n$ of the factorization algebra~$\cF$ is quasi-isomorphic to 
\begin{align*}
\cF(0) & \simeq {\rm Sym} \left(\Omega^{n,hol}(\Hat{D}^n,V_0^*)^\vee [-n]\right) \\ & \cong \cO\left(\Omega^{n,hol}(\Hat{D}^n,V_0^*) [n] \right) 
\end{align*}
where the differential on the right hand side is~$Q^{hol}$.
\end{lem}

\begin{proof}
First, notice that as graded topological vector spaces one has an isomorphism for any open set $U \subset M$ 
\beqn\label{eqn:dist}
\left(\overline{\cE}^!(U)\right)^\vee \simeq \cE_c(U) 
\eeqn
%
This implies there is an isomorphism
\beqn\label{eqn:dist2}
\Sym(\cE_c(U)) \simeq \cO\left(\overline{\cE}^!(U)\right) 
\eeqn
for any open set $U$.
By assumption, the linear differential $[\cdot]_1$ is elliptic, in particular the embedding of smooth sections into distributional sections
\beqn\label{eqn:dist3}
\cE^!(U) \hookrightarrow \overline{\cE}^! (U)
\eeqn
is a quasi-isomorphism for any open set~$U$. 

We can assume that $U \subset \C^n$ is a Stein open set containing~$0 \in \C^n$.
Then we have a sequence of quasi-isomorphisms
\begin{align*}
\overline{\cE}^! (U) & \simeq \cE^!(U) \\ & \simeq \Omega^{n,\bu}(U, V^*)[n].
\end{align*}
The result now follows from Lemma~\ref{lem:taylor}.

\end{proof}

\subsection{Local characters for topological-holomorphic theories}\label{s:localchar}

Suppose that $\cF$ is the factorization algebra of observables of a topological-holomorphic theory on $\R^m \times \C^n$. 
We will restrict our attention to cases where $\cF$, as a graded vector space, is of the form $\Sym(\cE^*)$ or $\Sym(\cE_c)$ where $\cE$ is of the form \eqref{eqn:cE}.

The local character $\chi_\cF ({\bf q})$ is, by definition, the graded character of algebra of local operators $\cF(0)$ with respect to some group of symmetries $H$, see \cite{SWchar}.
The particular group of symmetries depends on the theory, and we will present some examples momentarily. 

By assumption, as a graded algebra, the space of local operators $\cF(0)$ of a topological-holomorphic theory is of the form
\beqn
\cF(0) = \Sym (\lie{s})
\eeqn
where $\lie{s}$ is a graded topological vector space which we interpret as the linear local operators.

We will also assume that the group of symmetries $H$ acting on $\cF(0)$ arises from an action of $H$ on the linear local operators $\lie{s}$. 
Denote by $f_{\cF}({\bf q})$ the character of $\lie{s}$ with respect to this group action---this is the so-called `single particle' character. 
The full character of $\cF(0)$ is then given as the plethystic exponential of this single particle character
\beqn
\chi_{\cF}({\bf q}) = {\rm PExp}\left[f_{\cF}({\bf q}) \right] .
\eeqn

\subsection{Examples}

We present some simple examples. 

\begin{eg}
Suppose that $V$ is the trivial bundle on $\C^n$ and consider the theory whose fields are
\[
\cE = \Omega^\bu(\R^m) \otimes \Omega^{0,\bu}(\C^n) 
\]
where the differential is just $\d_{dR} + \dbar$. 
Then, the space of local operators is the symmetric algebra on the topological vector space which is linear dual to 
\[
\cO^{hol}(\Hat{D}^n) = \C[[z_1,\ldots,z_n]] .
\]
Via the $n$-dimensional residue one can identify the algebra of local operators with 
\[
\Sym\left(\frac{\d z_1}{z_1} \cdots \frac{\d z_n}{z_n}  \C[z_1^{-1}, \ldots , z_n^{-1}]\right) ,
\]
where $\lie{s} \simeq \frac{\d z_1}{z_1} \cdots \frac{\d z_n}{z_n}  \C[z_1^{-1}, \ldots , z_n^{-1}]$ is (equivalent to) the space of linear local operators. 

Consider the standard torus action $\C^\times \times \cdots \times \C^\times$ on $\C^n$. 
We would like to observe that the character of local operators with respect to this symmetry would be given by the plethystic exponential of the single particle index (the character of the space of linear local operators) which is immediate to compute:
\[
\frac{1}{(1-q_1)\cdots (1-q_n)} .
\]
However, the plethystic exponential cannot be applied to such an expression since as a power series in $q_1,\ldots,q_n$ there is a nonzero constant term.
This is related to the fact that there is an infinite number of operators for which the fugacities satisfy $q_1=\ldots=q_n=1$, so counting local operators in this way is ill-defined. 
One can remedy this by introducing a single extra variable fugacity $y$ and modify the single particle index to 
\[
\frac{y}{(1-q_1)\cdots (1-q_n)} .
\]
The plethystic exponential of such an expression returns the local character
\[
\chi(q_1,\ldots,q_n,y) = \prod_{k_1,\ldots,k_n \geq 0} \frac{1}{1-y q_1^{k_1}\cdots q_n^{k_n}}
\]
which now makes sense as a power series in the variables $y,q_1,\ldots,q_n$.
\end{eg}

Its instructive to see how local operators differ between $!$-dual factorization algebras.
Let us first point out a simple example. 
\begin{eg}
Consider the sheaf of cochain complexes
\[
\cE = \Omega^{0,\bu}\left(\C, K_{\C}^{\otimes r}\right),
\]
where $r \in \Z$ and the differential is~$\dbar$. 
Then, we can consider both the factorization algebra $\Obs = \cO(\cE)$ and its $!$-dual $\Obs^! = \Sym(\cE_c)$. 

The $\infty$-jets at $0 \in \C$ of $\cE$ is quasi-isomorphic to $\Gamma(\Hat{D}^n, K^{\otimes r}) = \d z^{\otimes r} \C[[z]]$. 
Thus the algebra of local operators $\Obs(0)$ is quasi-isomorphic to 
\[
\Obs(0) \simeq \cO \left(\Gamma(\Hat{D}, K^{\otimes r})\right) .
\]
In particular, the character of local operators $\Obs(0)$ is the plethystic exponential of
\[
\frac{q^{r}}{1-q} 
\]
where $q$ represents the fugacity for the standard~$\C^\times$ action on~$\C$.
Notice that when $r = 0$ we run into a similar problem as in the previous example. 
It is therefore convenient to introduce an extra fugacity $y$ which enters the single particle character as
\[
\frac{y q^{r}}{1-q}  .
\]



On the other hand, by Lemma \ref{lem:envelope} we see that the local operators associated to the $!$-dual $\Obs^!(0)$ is identified with the vector space
\[
\cO\left(\Gamma(\Hat{D}, K^{1-r})[1]\right) .
\]
In particular, the character of local operators $\Obs^!(0)$ is the plethystic exponential of
\[
-\frac{q^{1-r}}{1-q} 
\]
where $q$ represents the fugacity for the standard $\C^\times$ action on $\C$.
This time, when $r=1$ there is a problem with defining the plethystic exponential. 
To get an expression that makes sense for all $r$ we can again introduce a variable $y$ which enters the single particle character as
\[
- \frac{y q^{1-r}}{1-q} .
\]
\end{eg}

\subsection{Local characters for twisted superconformal theories}\label{s:localchar}

Suppose that $\cF$ is the factorization algebra of observables of a topological-holomorphic theory on $\R^m \times \C^n$. 
The local character $\chi_\cF ({\bf q})$ is, by definition, the graded character of algebra of local operators $\cF(0)$ with respect to some group of symmetries \cite{SWchar}.
The particular group of symmetries depends on the theory.
In this section we focus on local characters of factorization algebras that arise as twists of six-dimensional $\cN=(2,0)$ supersymmetric theories.

The (complexified) superconformal algebra in dimension six is $\lie{osp}(8|4)$. 
The holomorphic twist of this superconformal algebra is $\lie{osp}(6|2)$. 
We will consider the symmetry by the bosonic subalgebra
\beqn\label{eqn:cartan3}
\lie{sl}(3) \times \lie{sl}(2) \times \lie{gl}(1) \subset \lie{osp}(6|2)  .
\eeqn
The corresponding generators of the Cartan, as in \S \ref{sec:states}, were denoted $h_1,h_2,h,$ and $\Delta$ and the respective fugacities $t_1,t_2,r,q$.

We have described how this subalgebra embeds as fields in the twist of eleven-dimensional supergravity in \S \ref{s:ads7}. 
In particular, the holomorphic twist of any six-dimensional superconformal theory will have as a symmetry the subalgebra \eqref{eqn:cartan3}.
If the corresponding factorization algebra is $\cF$, and the local operators $\cF(0)$, the local character is then defined by the formal expression
\beqn
\chi_{\cF}(t_1,t_2,r,q) = {\rm Tr}_{\cF(0)} \left((-1)^F t_1^{h_1} t_2^{h_2} r^h q^\Delta\right) .
\eeqn
In the next section we will compute these characters in the case that the factorization algebra $\cF$ is $\clie_\bu(\cG_{N,c})$ where $N = 1,2,\ldots$.

We pointed out in \S \ref{sec:states} an alternative parametrization of the fugacities in terms of the parameters $y_1,y_2,y_3,y,q$ which satisfy the constraint $y_1 y_2 y_3 = 1$.
These parameters are related by $y_1=t_1^{-1}, y_2 = t_1 t_2^{-1}, y_3 = t_2$ and $y = q^{1/2} r$. 
We will also consider formulas for the local character $\chi(y_i,y,q)$ in terms of these variables.
  

\subsection{A relationship to the superconformal index}
\label{sec:sucaindex}
The local character for the holomorphic twist of a six-dimensional $\cN=(2,0)$ supersymmetric theory agrees with the well-known superconformal index.
Generally, in any dimension, the superconformal index counts states $\cH^Q$ which are annihilated by a particular supercharge~$Q$.
The index is defined as a function on the Cartan of a commuting subalgebra with respect to $Q$.
For six-dimensional superconformal theories, a natural choice of a supercharge is the holomorphic twisting supercharge. 
Then the index is sensitive to the so-called $\tfrac{1}{16}$-BPS states.

Recall that the odd part of the $\cN=(2,0)$ supersymmetry algebra is $S_+ \otimes R$ where $S_+$ is the positive irreducible spin representation of $\lie{so}(6)$ and $R \cong \C^4$.
Square-zero supercharges $Q \in S_+ \otimes R$ are stratified by the rank of the corresponding map $R \to (S_+)^* \cong S_-$.
A holomorphic supercharge $Q$ has rank one (such elements automatically square to zero). 
Thus, the superconformal index counts precisely the states in the holomorphic twist.
In the terminology above these states comprise the algebra of local operators $\cH^Q = \Obs(0)$ in the holomorphic twist of the six-dimensional $\cN=(2,0)$ theory.

The six-dimensional superconformal algebra (before twisting) is~$\lie{osp}(8|4)$.
The Cartan of the Lie super algebra is six-dimensional generated by elements
\[
H, J_1,J_2,J_3,R_1,R_2 .
\]
The holomorphic twisting supercharge $Q \in \lie{osp}(8|4)$ and the (super) commuting subalgebra is $\lie{osp}(6|2)$ together with the element 
\[
\Delta \define [Q,Q^\dagger] = H - (J_1 + J_2 + J_3) - 2 (R_1 + R_2) 
\]
where $Q^\dagger$ denotes the superconformal partner to the supercharge $Q$. 
The superconformal index counts states which saturate the BPS bound $\Delta \geq 0$ as a representation for the subalgebra $\lie{osp}(6|2)$. 
To fit with the notation used in this paper, the superconformal index can be written as
\beqn
\cI(y_i,y,q) = \Tr_{\cH^Q} (-1)^F q^{H + \frac13 (J_1+J_2+J_3)} y_1^{J_1} y_2^{J_2} y_3^{J_3} y^{R_1 - R_2} .
\eeqn
This agrees precisely with the local character $\chi(y_i,y,q)$ with the evident change of coordinates for the Cartan of $\lie{osp}(6|2)$. 

\subsection{Exceptional symmetry and a finite $N$ conjecture}

Generally speaking, after twisting there are enhancements of symmetries which are present in the original theory. 
We expect that the same occurs for any six-dimensional superconformal theory. 
In \cite{SW6d} we have shown that at the level of the holomorphic twist the twisted superconformal algebra $\lie{osp}(6|2)$ gets enhanced to the infinite-dimensional exceptional super Lie algebra $E(3|6)$ \cite{KacClass}. 
For the case of the theory on a stack of $N$ fivebranes, whose factorization algebra we denote by $\Obs_N$, this implies that the local operators $\Obs_{N}(0)$ form a representation for $E(3|6)$.

Our goal is to gain knowledge of the structure of $\Obs_N(0)$ as an $E(3|6)$-representation from our holographic analysis of the previous section.
Indeed, in \S \ref{s:fact} we have expressed the restriction of the factorization algebra of observables of twisted eleven-dimensional supergravity to the three-fold $Z$ as the Chevalley--Eilenberg cochains of a local $L_\infty$ algebra $\cG$. 
Recall that we have a decomposition of local Lie algebras $\cG = \oplus_{j \geq -1} \cG^{(j)}$ on the three-fold $Z$. 
From this decomposition we have defined a family of local Lie algebras $\cG_{N}$ on $Z = \C^3$ for $N=1,2,\ldots$.

In \S \ref{sec:factsummary} we explained the expectation that to an open set $U \subset \C^3$, the Lie algebra cohomology of $\cG_{N,c}(U)$ is equivalent to the observables $\Obs_N(U)$ of the six-dimensional theory supported on $U$. 
Each $\cG_N$ is acted on by the local Lie algebra $\cG^{(0)} = \cE(3|6)$.
The $\infty$-jets of $\cE(3|6)$ at $0 \in \C^3$ is exactly the exceptional super Lie algebra $E(3|6)$.  
Thus, for every $N$, the space of local operators of the factorization algebra $\clie_\bu(\cG_{N,c})$ is naturally and $E(3|6)$-representation. 
At the level of local operators we can make the following conjecture, which we will further elucidate at the level of characters for $E(3|6)$ in the next section.

\begin{conj}
\label{conj:ops}
Let $\Obs_N(0)$ be the local operators of the theory on a stack of $N$ fivebranes wrapping $\C^3$ in $\R \times \C^5$. 
There is an equivalence of $E(3|6)$-representations
\beqn
\Obs_N(0) \simeq \clie_\bu (\cG_{N,c}) (0) .
\eeqn
Similarly, let $\til \Obs_{N}(0)$ be the local operators of the theory on a stack of $N \geq 2$ fivebranes with the center of mass degrees of freedom removed. 
There is an equivalence of $E(3|6)$-representations
\beqn
\til \Obs_N(0) \simeq \clie_\bu (\til\cG_{N,c}) (0) .
\eeqn
\end{conj}

\section{Conjectures for indices of operators on fivebranes}
\label{s:finite}

In conjecture \ref{conj:ops} we have formulated a conjectural description of the space of local operators $\Obs_{N}(0)$ associated to the worldvolume theory on a stack of $N$ fivebranes in the holomorphic twist.
As we reviewed just in the previous section, the space of local operators is what categorifies the superconformal index that we study in this paper.
In this section we begin to provide some evidence for this description at the level of characters.

For a stack of $N=1$ fivebranes, which corresponds to the abelian six-dimensional superconformal field theory, we find that our local character matches exactly with the expressions in the literature. 
This is not a surprise as we have shown that even at the level of factorization algebras $\clie_\bu(\cG_{1,c})$ is quasi-isomorphic to the classical limit of~$\Obs_1$, see Proposition~\ref{prop:factabelian}.

The main computation of this section is a closed formula for the local character of the factorization algebra $\clie_\bu(\cG_{N,c})$ for $N > 1$, see Theorem \ref{thm:finite}. 
Following conjecture \ref{conj:ops} and the general discussion of \S \ref{sec:sucaindex} we are led to hypothesize a closed formula for the superconformal index for the theory on a finite number of fivebranes (in flat space).
As far as the authors are aware of there is no closed formula for the refined superconformal index (with four independent fugacities) for the theory on a stack of $N > 1$ fivebranes.
For small values of $N$ we expand our closed formulas to low orders in the fugacity $q$ (which roughly counts instanton charge) to match exactly with expressions in the literature. 

\subsection{Operators on a single fivebrane}

We deduce the character of the holomorphic twist of the theory on a single fivebrane and will find an exact match with the index of the six-dimensional superconformal theory associated to the abelian Lie algebra $\lie{gl}(1)$.
By the proposition \ref{prop:factabelian} we can compute this character either from a first principles description of the theory, or holographically by focusing on the weight $(-1)$ part of the decomposition of $\Bar{\pi}_* \Obs_{sugra}$.

\begin{lem}
\label{lem:single}
The $\Z \times \Z/2$ graded algebra of local operators $\Obs_{1}(0)$ of the holomorphic twist of the worldvolume theory of a single fivebrane is quasi-isomorphic to the graded symmetric algebra on the linear dual of the topological vector space
\beqn\label{eqn:localfree}
V_0[[z_1,z_2,z_3]] \simeq \Omega^{2}_{cl} (\Hat{D}^3)[1] \oplus \Pi \Omega^0(\Hat{D}^3, K^{1/2}) \otimes \C^2 [1].
\eeqn
\end{lem}

\begin{proof}
The jet expansion at $0 \in \C^3$ determines a map from the sections of the abelian holomorphic-topological local Lie algebra on $\C^3$ to the cochain complex
\beqn
\begin{tikzcd}
\ul{-1} & \ul{0} \\
\Omega^{2}(\Hat{D}^3) \ar[r,"\del"] & \Omega^{3}(\Hat{D}^3) \\
\Pi \Omega^0(\Hat{D}^3, K^{1/2}_{\Hat{D}^3}\otimes \C^2) . 
\end{tikzcd} 
\eeqn
On the formal disk all closed two-forms are automatically exact, which implies the lemma.
\end{proof}

We present the character of $\Obs_1(0)$ as the plethystic exponential of the character $f_1(t_1,t_2,r,q)$ of the space of linear local operators
\beqn
\chi_{1} (t_1,t_2,r,q) = {\rm PExp} \big[f_1(t_1,t_2,r,q) \big] .
\eeqn
According to the weights listed above and using the description of local operators in Lemma \ref{lem:single} we have the following contributions to the single particle character~$f_{1}(t_1,t_2,r,q)$.

\begin{itemize}
\item Single particle operators on the odd copy of holomorphic two-forms $\Pi \Omega^{2,hol}$ contribute
\[
- q^2 \frac{\chi^{\lie{sl}(3)}_{[0,1]}(t_1,t_2)}{(1-t_1^{-1}q) (1-t_1 t_2^{-1} q) (1-t_2 q)} = - q^2 \frac{t_1  + t_1^{-1} t_2  + t_2^{-1} }{(1-t_1^{-1}q) (1-t_1 t_2^{-1} q) (1-t_2 q)}
\]
where $\chi^{\lie{sl}(3)}_{[1,0]}(t_1,t_2)$ is the $\lie{sl}(3)$ character of highest weight $[1,0]$.
\item Single particle operators on the even copy of holomorphic three-forms $\Omega^{3,hol}$ contribute
\[
q^3 \frac{1}{(1-t_1^{-1}q) (1-t_1 t_2^{-1} q) (1-t_2 q)} 
\]
\item Single particle operators on $K^{1/2} \otimes \C^2$ contribute
\[
q^{3/2} \frac{\chi_{1}^{\lie{sl}(2)} (r)}{(1-t_1^{-1}q) (1-t_1 t_2^{-1} q) (1-t_2 q)} = q^{3/2}\frac{(r + r^{-1})}{(1-t_1^{-1}q) (1-t_1 t_2^{-1} q) (1-t_2 q)}
\]
where $\chi_{1}^{\lie{sl}(2)} (r)$ is the $\lie{sl}(2)$ character of highest weight one.
\end{itemize}

Putting this all together we obtain the following.

\begin{prop}
\label{prop:6done}
The local character $\chi_{1}(t_1,t_2,r,q)$ of the holomorphic twist of the theory on a single fivebrane is given by the plethystic exponential of the single particle character
\beqn\label{eqn:6done}
f_{1} (t_1,t_2,r,q) = \frac{q^{3/2}(r + r^{-1}) - q^2(t_1 + t_1^{-1} t_2 + t_2^{-1} ) + q^3}{(1-t_1^{-1}q) (1-t_1 t_2^{-1} q) (1-t_2 q)} .
\eeqn
\end{prop}

In terms of the parameters $y_1,y_2,y_3,y,q$ this single particle character reads
\beqn
\label{eqn:6done1}
f_{1} (y_i,y,q) = \frac{qy + q^2y^{-1} - q^2(y^{-1}_1+y^{-1}_2+y^{-1}_3) + q^3}{(1-y_1q) (1-y_2 q) (1-y_3 q)} .
\eeqn
The expression matches exactly with the index of the abelian six-dimensional superconformal theory.
For example, compare with \cite[Eq. (3.1)]{Kim:2013nva} or \cite[Eq. (3.35)]{Bhattacharya:2008zy}.
From now on, we will give all formulas for the index in terms of the parameters $y_1,y_2,y_3,y,q$.

In Proposition \ref{lem:single} we have shown that $\clie_\bu (\cG_{c}^{(-1)})$ is equivalent to the factorization algebra $\Obs^{cl}_{1}$ encoding the classical observables  of the holomorphic twist on a single fivebrane.
On $\C^3$, the global sections of the local Lie algebra $\cG^{(0)}$ is closely related to $E(3|6)$---the $\infty$-jets of $\cG^{(0)}$ at $0 \in \C^3$ is quasi-isomorphic to $E(3|6)$.
Combining these facts we see that $\Obs_{1}(0)$ is a module for $E(3|6)$. 
This module turns out to be a one-dimensional extension of an irreducible $E(3|6)$-module which was classified in~\cite{KR2}. 


%
%

\parsec
There are various degenerations, or specializations, of this character which are interesting to consider.
A particularly meaningful one is related to two different deformations of the theory by elements in the (twisted) superconformal algebra and is known as the Schur limit of the index.

Recall that after performing the holomorphic twist the residual superconformal algebra is~$\lie{osp}(6|2)$.
We have recalled in \S\ref{s:global1} how the bosonic part of this algebra is represented by fields of the eleven-dimensional theory. 
There are two types of odd elements of~$\lie{osp}(6|2)$ that also have a natural interpretation in the eleven-dimensional theory.
The odd part of~$\lie{osp}(6|2)$ can be identified with the twelve-dimensional space
\[
\C^3 \otimes \C^2 \oplus \wedge^2(\C^3) \otimes \C^2 
\]
where $\C^3, \C^2$ are the fundamental $\lie{sl}(3)$ and $\lie{sl}(2)$ representations, respectively. 
The $\lie{gl}(1)$ factor in the bosonic part of $\lie{osp}(6|2)$ acts with weight $1/2$ on both summands. 

\begin{itemize}
\item The summand $\C^3 \otimes \C^2$ embeds into the ghosts of twisted supergravity via the $\gamma$-type fields which satisfy
\[
\del \gamma = \d w_a \d z_i .
\]
where $i=1,2,3$ and $a = 1,2$.
Note that $\gamma$ appears to be ambiguous up to a closed holomorphic one-form, but since there is a linear gauge symmetry which sends $\beta \mapsto \del \beta$, it implies that $\gamma$ is unique up to a BRST exact term. 
Since in our model all closed one-forms are rendered trivial in cohomology
\item The summand $\wedge^2(\C^3) \otimes \C^2$ embeds as another $\gamma$-type field which satisfies 
\[
\del \gamma = w_a \d z_i \d z_j .
\]
\end{itemize}

Both deformations break the global Cartan subalgebra down to $\lie{gl}(1) \times \lie{gl}(1)$ according to the specializations
\beqn\label{eqn:special1}
y=1 , \quad y_3 = 1 .
\eeqn
Notice that due to the constraint $y_1y_2y_3=1$ this forces $y_1 = y_2^{-1}$.
As one can easily check, this specialization yields the following single particle index
\[
f_{1}(y_1, y_1^{-1},y_3=1, y=1, q) = \frac{q}{1-q} 
\]
which recovers the single particle index of a single chiral boson on the Riemann surface $\Sigma = \C_{z_1}$. 
Notice that the dependence on the parameter $y_1$ has completely dropped out even though we have not specialized it to any value.

\subsection{A conjectural description of operators on a stack of two fivebranes}

In \S\ref{sec:factsummary} we saw that the decomposition of the local $L_\infty$ algebra $\cG = \cG_Z$ on $Z$ induces a filtration of the factorization algebra $\clie_\bu(\cG_c)$. 
\[
\clie_{\bu}(\cG_{1,c}) \subset \clie_{\bu}(\cG_{2,c}) \subset \cdots .
\]
We now turn to the factorization algebra $\clie_{\bu}(\cG_{2,c})$.

Recall that $\cG_{2}$ is the local $L_\infty$ algebra on $Z$ defined as $\cG_{2} = \cG / \cG^{\geq 1}$. 
Since $\cG$ is concentrated in weights $\geq -1$ we see that $\til \cG_{2}$ is of the form
\[
\cG_2 = \til \cG_2 \ltimes \cG_1 
\]
where $\cG_1 = \cG^{(-1)}$ is the weight $(-1)$ piece and $\til \cG_2 = \cG^{\geq 0} / \cG^{\geq 1} = \cG^{(0)}$.  
We focus mostly on the factorization algebra $\clie_\bu(\til \cG_{2,c})$.

We have already characterized the local dg Lie algebra $\til \cG_{2} = \cG^{(0)}$ as the weight zero part of $\cG$ on on any threefold $Z$ in \S\ref{s:weight0}. 
We have also shown that $\cG^{(0)}$ is equivalent to the local Lie algebra $\cE(3|6)$. 
The even part of $\cE(3|6)$ is
\[
\Omega^{0,\bu}(Z, \T_Z) \oplus \Omega^{0,\bu}(Z) \otimes \lie{sl}(2) 
\]
with its natural cohomological grading by Dolbeault form type. 
The odd part of $\cE(3|6)$ is
\[
\Omega^{1,\bu}(Z, K_Z^{-1/2}) \otimes \C^2 .
\]
The differential is $\dbar$ and the Lie bracket has been described in \S\ref{s:weight0}.


\parsec

We continue by computing the character of local operators associated to the factorization algebra $\clie_\bu(\cG_{2,c})$ using Lemma~\ref{lem:envelope}.
For simplicity we will use the fugacities $y_i, y, q$.

\begin{itemize}
\item Single particle operators coming from the copy of holomorphic vector fields $\Vect^{hol}(\C^3)$ contribute
\[
q^4 \frac{\chi^{\lie{sl}(3)}_{[1,0]}(y_i)}{(1-y_1q) (1-y_2 q) (1-y_3 q)}  = q^4 \frac{y_1 + y_2 + y_3}{(1-y_1q) (1-y_2 q) (1-y_3 q)} 
\]
\item Single particle operators coming from $\lie{sl}(2)$-valued holomorphic functions $\lie{sl}(2) \otimes \cO^{hol}(\C^3)$ contribute
\[
q^3 \frac{\chi_2^{\lie{sl}(2)} (q^{-1/2}y)}{(1-y_1q) (1-y_2 q) (1-y_3 q)}  = \frac{q^2 y^2 + q^3 + q^4 y^{-2}}{(1-y_1q) (1-y_2 q) (1-y_3 q)} 
\]
\item Single particle operators coming from the odd piece of $E(3|6)$ which is $\Omega^{1,hol} \otimes K^{-1/2} \otimes \C^2$ contribute
\[
q^{7/2}\frac{\chi^{\lie{sl}(2)}_{1}(q^{-1/2} y) \, \chi_{[0,1]}^{\lie{sl}(3)} (y_i)}{(1-y_1q) (1-y_2 q) (1-y_3 q)} = q^{3}\frac{(y + q y^{-1})(y_1^{-1} + y_2^{-1} + y_3^{-1})}{(1-y_1q) (1-y_2 q) (1-y_3 q)} 
\]
\end{itemize}

Combining these expressions we obtain the following.

\begin{prop} \label{prop:6dtwo}
The character of local operators of the factorization algebra $\clie_\bu(\til \cG_{2,c})$ on $\C^3$ is given by the plethystic exponential of the following expression
\beqn\label{eqn:6dtwo}
\til f_{2} (y_i,y,q) = \frac{q^4(y_1+y_2+y_3) + q^2 (y^2 + q + q^2 y^{-2}) - q^{3} (y + q y^{-1})(y_1^{-1} + y_2^{-1} + y_3^{-1})}{(1-y_1q) (1-y_2 q) (1-y_3 q)}.
\eeqn
\end{prop}

Recall that our conjecture for the space of local operators associated to the holomorphic twist of the six-dimensional worldvolume theory on a stack of two fivebranes is $\Obs_2 (0) \simeq \clie_\bu(\cG_{2,c})(0) \cong \clie_\bu(\cG_{1,c})(0) \otimes \clie_\bu(\til \cG_{2,c})(0)$. 
And after removing the center of mass degrees of freedom, our conjecture is $\til \Obs_2 \simeq \clie_\bu(\til \cG_{2,c})$.

Just as in the abelian case, the local operators $\til \Obs_2(0)$ form a module over $E(3|6)$.
It turns out that this module is irreducible~\cite{KR2}.

We can now state a decategorified version of conjecture \ref{conj:ops} at the level of superconformal indices, or local characters.

\begin{conj}\label{conj:6dtwo}
The superconformal index of the six-dimensional superconformal theory of type $A_1$ is given by
\[
\til \chi_{2} (y_i,y,q) = {\rm PExp} \left[\til f_2(y_i,y,q) \right] .
\]
where $\til f_2(y_i,y,q)$ is as in \eqref{eqn:6dtwo}.
\end{conj}

Similarly, the index associated to the $\lie{gl}(2)$ theory, which is the local character of $\clie_\bu(\cG_{2,c})$, is conjectured to be simply the product 
\[
\chi_{2} (y_i,y,q) = \chi_{2} (y_i,y,q) \cdot \chi_{1}(y_i,y,q)
\]
where the character $\chi_{1}$ for the $\lie{gl}(1)$ theory is given in proposition~\ref{prop:6done}
Equivalently, $\chi_2$ is the plethystic exponential of $f_2 = f_1 + \til f_2$. 

%

\parsec[]

The Schur limit $y=1, y_3=1$ of $\til f_2$ in \eqref{eqn:special1} yields 
\[
\til f_{2}(y_1, y_2, y_3=1, y=1, q) = \frac{q^2}{1-q} 
\]
which is the single particle index of Virasoro vacuum module on the Riemann surface $\Sigma = \C_{z_1}$. 

\subsection{A closed formula for the finite $N$ index}

Before exhibiting the general formula for the local character of the factorization algebra $\clie_\bu(\cG_{N,c})$ on $\C^3$ we set up some notation. 
As above, we let $\chi_k^{\lie{sl}(2)}$ and $\chi^{\lie{sl}(3)}_{[k,l]}$ denote the highest weight $\lie{sl}(2)$ and $\lie{sl}(3)$ characters. 
We also define the following expression which appears in the denominator in all of our characters
\beqn
d(y_i,y,q) = (1-y_1 q)(1-y_2q)(1-y_3q) .
\eeqn 
To simplify formulas, we will temporarily denote the single particle character for the $N=1$ theory $\Obs_1$ by 
\beqn
g_{-1} (y_i,y,q) = f_1(y_i,y,q)
\eeqn
where $f_1(y_i,y,q)$ is as in equation \eqref{eqn:6done1} and also denote by 
\beqn
g_0 (y_i,y,q) = \til f_2(y_i,y,q)
\eeqn
where $\til f_2(y_i,y,q)$ is as in equation \eqref{eqn:6dtwo}. 
Thus $g_2$ is the single particle local character of $\clie_\bu(\til \cG_{2,c}) = \clie_{\bu}(\cG_c^{(0)})$.
Finally, for $k \geq 1$ let
\beqn
\label{eqn:gk}
\begin{array}{lllll}
g_k (y_i,y,q) \define & q^{3} \left(q^{1 + 3 k/2} \chi^{\lie{sl}(2)}_{k}(q^{-1/2} y)\chi^{\lie{sl}(3)}_{[1,0]}(y_i) + q^{3k/2} \chi^{\lie{sl}(2)}_{k+2}(q^{-1/2} y) \right. \\
&\displaystyle \frac{\left.  - q^{3(k+1)/2} \chi^{\lie{sl}(2)}_{k-1}(q^{-1/2}y) - q^{-1 + 3(k+1)/2} \chi^{\lie{sl}(2)}_{k+1} (q^{-1/2} y) \chi^{\lie{sl}(3)}_{[0,1]}(y_i) \right)}{d(y_i,y,q)} .
\end{array}
\eeqn

\begin{thm}
\label{thm:finite}
Let $N \geq 3$. 
The local character of the factorization algebra $\clie_{\bu}(\cG_{N,c})$ is
\beqn
\chi_{N}(y_1,y_2,y_3,y,q) = \text{PExp}\left[\sum_{k=-1}^{N-2} g_k(y_1,y_2,y_3,y,q)\right].
\eeqn
Similarly, the local character of the factorization algebra $\clie_\bu(\til{\cG}_{N,c})$ is 
\beqn
\til{\chi}_{N}(y_1,y_2,y_3,y,q) = \text{PExp}\left[\sum_{k=0}^{N-2} g_k(y_1,y_2,y_3,y,q)\right].
\eeqn
\end{thm}
\begin{proof}
By Lemma~\ref{lem:envelope} the character of $\clie_\bu (\cG_{N,c})$ is given by 
\beqn
\chi_N = \text{PExp} \left[f_N\right]
\eeqn
where $f_N$ is the single particle local character.
Thus, it suffices to show that $f_N = \sum_{k = -1}^{N-2} g_k$.
Recall that from the description \eqref{eqn:gN} we have, as local Lie algebras:
\beqn
\cG_N = \cG / \cG^{(\geq N-2)} ,
\eeqn 
for $N \geq 1$. 
In particular, as a super vector bundle on the threefold $Z = \C^3$ we have
\[
\cG_N = \cG^{(-1)} \oplus \cG^{(0)} \oplus \cdots \oplus \cG^{(N-2)} .
\]
So, it suffices to observe that $g_k$ is the single particle index of the factorization algebra $\clie_\bu(\cG^{(k)}_c)$, which is a direct observation using the description of $\cG^{(k)}$ we have given in Proposition \ref{prop:Vj}.
\end{proof}

We thus arrive at the following conjecture for the index of the worldvolume theory on a stack of a finite number of fivebranes which we phrase in terms of the six-dimensional superconformal theory associated to the Lie algebra $\lie{sl}(N)$.

\begin{conj} 
The superconformal index of the six-dimensional superconformal theory associated to the Lie algebra of type $A_{N-1}$ is $\til \chi_{N}(y_1,y_2,y_3,y,q)$. 
\end{conj}

We proceed to give some concrete evidence for this conjecture.
First, we show that when we take the limit as $N \to \infty$ that we recover the index computed from the gravitational side.

\parsec

It follows from the limit description in \eqref{eqn:lim} that the large $N$ limit of $\chi_N$ is precisely the multiparticle supergravity index we computed in proposition~\ref{prop:sugraindex1}. 
Alternatively, we have the following direct proof of this fact. 

\begin{prop}
One has
\beqn
\chi_{sugra}(y_i, y, q) = \lim_{N \to \infty} \chi_N(y_i,y,q)
\eeqn
\end{prop}

\begin{proof}
It suffices to show that at the level of single particle indices one has
\beqn
f_{sugra}(y_i, y, q) = \lim_{N \to \infty} f_N(y_i,y,q) ,
\eeqn
where $f_N = \sum_{k = -1}^{N-2} g_k$. 

We will use the following identity 
\beqn
\sum_{k=0}^\infty q^{3k/2} \chi_{k}^{\lie{sl}(2)}(q^{-1/2}y) = \frac{1}{(1-q y)(1-q^2 y^{-1})} .
\eeqn
We will denote this expression by $S(y,q)$.

Using this identity one can directly see that the result reduces to observing that
\begin{multline}
\left(q^4 (y_1+y_2+y_3) + 1 - q^6 - q^2 (y_1^{-1} + y_2^{-1} + y_3^{-1})\right)S(y,q) - 1 + q^3= \\
\left(q^4(y_1+y_2+y_3)-q^2(y_1^{-1} + y_2^{-1} + y_3^{-1})+(1-q^3)(yq + y^{-1} q^2) \right) S(y,q) .
\end{multline}


\end{proof}

As an immediate corollary we have the following result.
\begin{cor}
For any $N \geq 1$ one has
\beqn
\chi_{sugra}(y_i,y,q) = \til{\chi}_N(y_i,y,q) \mod q^{N+1} .
\eeqn
\end{cor}
\begin{proof}
This follows from observing that at the level of single particle states $f_N$ is of order $q^{N}$.
\end{proof}

\parsec

We can also apply the Schur limit  $y,y_3\to 1$ to $\chi_N$.
\begin{prop}
Upon specializing $y=1,y_3=1$ (so that $y_1 y_2 = 1$) one has the following single particle index
\beqn
f_N (y_1,y_2, y_3=1,y=1,q) = \frac{q^2 + q^3 + \cdots + q^{N}}{1-q} 
\eeqn
The plethystic exponential of the right hand side agrees with the vacuum character of the $W_{N}$ vertex algebra.\footnote{By this we mean the principal $W$-algebra of type $A_{N-1}$.}
\end{prop}
\begin{proof}
By induction it suffices to show that the specialization of the single particle local character $g_k$ of the factorization algebra $\clie_\bu(\cG^{(k)})$ is $q^{k+2} / (1-q)$. 
We have already seen this in the case $k=-1,0$, so it suffices to show this when $k \geq 1$.

First observe that the denominator becomes
\beqn
d(y_1,y_2,y_3=1,y=1,q) = (1-y_1 q)(1-y_2q) (1-q) .
\eeqn

Next, we observe that the numerator of $g_k (y_1,y_2,y_3=1,y=1,q)$ can be factored as
\begin{align*}
q^{3 + 3k/2} \left(q^{-(k+2)/2} + q^{-(k-2)/2} - q^{-k/2} (y_1+y_2) \right) 
& = q^{k+2} (1 + q^2 - q (y_1 + y_2)) \\
& = q^{k+2} (1 - y_1 q) (1-y_2 q) 
\end{align*}
where in the last line we have used $y_1 y_2 = 1$.
The result follows.
\end{proof}

\parsec[]

We would also like to point out compatibility of our expression with a certain ``minimally reduced'' index considered in \cite{Gaiotto:2021xce}. 
This minimal reduction is the result of sending certain parameters to zero while keeping some expression in the fugacities fixed.
To consider it it is useful to make the following change of variables: 
\beqn
z_i = y_i q, \quad w_1 = yq, \quad w_2 = y^{-1} q^2 .
\eeqn
These variables satisfy the constraint $z_1 z_2 z_3 = w_1 w_2$. 

This minimally reduced index corresponds to taking the following limit in the new fugacities
\beqn\label{eqn:limitgaitto}
z_3 , w_2 \to 0 .
\eeqn
This yields an index which only accounts for operators which transform trivially with respect to the symmetries that the fugacities $z_3,w_2$ correspond to. 
This will result in an index which has three remaining fugacities.

\begin{prop}
\label{prop:gaiotto}
The limit $z_3 , w_2 \to 0$ of the expression $\chi_{N}(z_i,w_a)$ is 
\beqn
\prod_{a=1}^N \prod_{b,c \geq 0} \frac{1-w_1^{a-1}z_1^{b+1} z_2^{c+1}}{1-w_1^a z_1^b z_2^c} .
\eeqn
\end{prop}
\begin{proof}
It is easy to see that the $z_3,w_2 \to 0$ limit of $g_{-1}$ is 
\beqn
g_{-1}(z_1,z_2,w_1) = \frac{w_1 - z_1 z_2}{(1-z_1)(1-z_2)} 
\eeqn
and the $z_3,w_2 \to 0$ limit of $g_0$ is 
\beqn
g_0(z_1,z_2,w_3) = w_1 g_{-1}(z_1,z_2,w_1) .
\eeqn

In the coordinates $z_i,w_a$ the expression $g_k$, for $k \geq 1$, in \eqref{eqn:gk} can be written as
\beqn
\label{eqn:gk}
\begin{array}{lllll}
g_k (z_i,w_a) \define & \left( z_1z_2z_3 p_k(w_1,w_2) (z_1+z_2+z_3) + p_{k+2}(w_1,w_2)  \right. \\
&\displaystyle \frac{\left.  -z_1z_2z_3 p_{k-1}(w_1,w_2) - p_{k+1}(w_1,w_2) (z_1z_2+z_2z_3+z_1z_3) \right)}{(1-z_1)(1-z_2)(1-z_3)} .
\end{array}
\eeqn
Here $p_k(w_1,w_2) = \sum_{i+j=k} w_1^i w_2^j$. 

From this expression it is easy to see that $\lim_{z_3,w_2 \to 0} g_k$ is 
\beqn
g_k(z_1,z_2,w_1) = w_1^{k+1} g_{-1}(z_1,z_2,w_1) .
\eeqn
The result follows from applying the plethystic exponential.
\end{proof}

The $z_3,w_2 \to 0$ limit of our index is quite similar, though not exactly, the index of a four-dimensional $\cN=1$ theory on $\C^2$ where the fugacities $z_1,z_2$ count holomorphic derivatives in each of the complex directions.
Also, note that this minimally reduced index further reduces to the Schur limit (so the character of the $W_N$ vertex algebra) by specializing $z_1 = w_1$.

\subsection{Comparisons to expansions of superconformal indices}

In the final section we would like to exhibit a series of direct consistency checks with our conjectural exact formula for the index of the non-abelian six-dimensional superconformal theory with a number of expansions that have appeared in recent literature. 

\parsec
Let us first focus on the superconformal theory associated to the Lie algebra $\lie{sl}(2)$ (so type $A_1$).
Our conjecture for the superconformal index in this case is the plethystic exponential of $\til f_2 (y_i,y,q)$ from equation \eqref{eqn:6dtwo}.
We expand the formal single particle index $\til f_2 (y_i, y, q)$ as a series in the variable~$q$, yielding
\begin{align*}
\til f_2 (y_i,y,q) & = y^2 q^2 + \left(1 - \chi_{[0,1]}(y_i) y + \chi_{[1,0]} y^2 \right) q^3 \\
& + \left(y^{-2} - \chi_{[0,1]}(y_i) y^{-1} + 2 \chi_{[1,0]}(y_i) - \chi_{[0,1]} (y_i) y + \chi_{[2,0]}(y_i) y^2 \right) q^4 + O(q^5) .
\end{align*}
From this expression, we obtain the $q$-expansion of the index $\til \chi_2(y_i,y,q) = \text{PExp}[\til f_2]$ as 
\begin{align*}
\til \chi_2(y_i,y,q) & = 1 + y^2 q^2 + \left(1-\chi_{[0,1]}(y_i) y + \chi_{[1,0]}(y_i)y^2\right)q^3 \\ 
& + \left(y^{-2} - \chi_{[0,1]}(y_i) y^{-1} + 2 \chi_{[1,0]}(y_i) - \chi_{[0,1]} (y_i) y + \chi_{[2,0]}(y_i) y^2 + y^4\right)q^4 + O(q^5)
\end{align*}

Similarly, for the $\lie{gl}(2)$ theory $\chi_2 = \text{PExp}[f_1 + \til f_2] = \chi_1 \cdot \til \chi_2$ we find the expansion
\begin{align*}
\chi_2 (y_i,y,q) & = y q + \left(y^{-1} - \chi_{[0,1]}(y_i) + \chi_{[1,0]}(y_i) y + 2y^2 \right) q^2 \\ 
& + \left( \chi_{[1,0]}(y_i) y^{-1} - (\chi_{[1,1]}(y_i)-2) + (\chi_{[2,0]}(y_i) - 2 \chi_{[0,1]}(y_i)) y + 2 \chi_{[1,0]}(y_i) y^2 + 2y^3\right) q^3 \\ & + O(q^4) .
\end{align*}

We observe that these $q$-expansions agree precisely with the expansions in \cite{Kim:2013nva} for the $\lie{gl}(2)$ theory  
(see equations (3.51) and (3.65) of \textit{loc. cit.}).

\parsec

We proceed to compare expansions of our exact expression for the $\lie{gl}(3)$ theory to those in \cite{Kim:2013nva}. 
Recall that our conjectural $\lie{gl}(3)$ index is given by the local character of the holomorphic factorization algebra $\Obs_3$:
\beqn
\chi_3 (y_i,y,q) = \text{PExp}[f_3(y_i,y,q)] = \chi_2(y_i,y,q) \cdot \text{PExp}[g_3(y_i,y,q)]  .
\eeqn
Here, $f_3(y_i,y,q)$ is the single particle local character for the holomorphic factorization algebra $\Obs_3$ and $g_3(y_i,y,q)$ is given in equation \eqref{eqn:gk}. 

Since $g_3(y_i,y,q) = y^3 q^3 + O(q^4)$ we see that $\chi_3$ and $\chi_2$ agree up to order $q^2$ and the difference at order $q^3$ is simply
\beqn
\chi_3(y_i,y,q) - \chi_2(y_i,y,q) = y^3q^3 + O(q^4) .
\eeqn
This is again in exact agreement with the index for the $\lie{gl}(3)$ theory computed \cite{Kim:2013nva} up to order $q^3$ (see equation (3.79) of \textit{loc. cit.}). 

\parsec

Next, we compare to expansions for the $\lie{sl}(N)$ theory computed in \cite{Imamura}, where the method of the `giant graviton' expansion is used.
It will be convenient to change the variables $(y_i, y, q) \to (y_i,x,q)$ where 
\beqn
x = qy .
\eeqn 
We will again expand in powers of $q$.\footnote{To match precisely with the equations in \cite{Imamura} we note that it is necessary to relable the variables $y_i \leftrightarrow u_i$, $x \leftrightarrow \check{x}$, and $q \leftrightarrow y$ where the variable $y$ is distinct from the one we use in this paper!}

Starting with the $\lie{sl}(2)$ theory we find that up to order $q^4$ the single particle index is
\begin{align*}
\til f_2 (y_i,x,q) & = x^2 + \chi_{[1,0]}(y_i) x^2 q \\
& + \left(-\chi_{[0,1]}(y_i) x + \chi_{[2,0]}(y_i) x^2 \right) q^2 +  \left(1-x-\chi_{[1,1]}(y_i) x + \chi_{[3,0]}(y_i) x^2 \right) q^3 \\
& + \left(2 \chi_{[1,0]}(y_i) - \chi_{[2,1]}(y_i) x + \chi_{[4,0]}(y_i)x^2 \right)q^4 + O(q^5) .
\end{align*}

It follows that the plethystic exponential $\til \chi_2(y_i,y,q)$ of this expression has $q$-expansion
\begin{align*}
\til \chi_2(y_i,y,q) & = \frac{1}{1-x^2} +  \frac{x^2}{1-x^2} \chi_{[1,0]}(y_i) q \\
& + \left(- x \chi_{[0,1]}(y_i) +  x^2 (1+x^2)\chi_{[2,0]}(y_i)\right) \frac{1}{1-x^2} q^2 \\ 
& + \left( 1 - x - x^3 + x^6 + (-x - x^3 + x^4) \chi_{[1,1]}(y_i)  + (x^2 + x^4 + x^6)\chi_{[3,0]}(y_i)  \right) \frac{1}{1-x^2} q^3 \\
& + O(q^4) 
\end{align*}
This agrees with the expansion in \cite{Imamura} (see equation (68)) except for the $\lie{sl}(3)$-scalar term at order $q^3$. 
We find $(1-x-x^3+x^6) / (1-x^2) = (1-x^3-x^4-x^5) / (1+x)$ whereas Imamura's result is $1 / (1+x)$. 

Similarly, we can obtain the $q$-expansions for the local character $\chi_3(y_i,x,q)$ of the factorization algebra $\Obs_3$ and compare it to the $q$-expansion for the superconformal index of the $\lie{sl}(3)$ theory in \cite{Imamura}. 
Up to order $q^2$ we have
\begin{align*}
\chi_3(y_i,x,q) & = \frac{1}{(1-x^2)(1-x^3)} + \frac{x^2}{(1-x)(1-x^3)} \chi_{[1,0]}(y_i) q \\
& + \left((-x -x^2 + x^5)\chi_{[0,1]}(y_i) + (x^2 + x^3 + x^4 + x^5+ x^6) \chi_{[2,0]}(y_i) \right) \frac{1}{(1-x^2)(1-x^3)} q^2 \\
& + O(q^3) .
\end{align*}
Again, we find a discrepancy of our expansion compared to \cite{Imamura} at order~$q^3$.
It would be interesting to explain the physical or representation theoretic sources of these discrepancies in each of these cases.



\printbibliography

\end{document}